\RequirePackage{fix-cm}
\documentclass[12pt,a4paper]{article}

\usepackage{authblk}
\usepackage[T1]{fontenc}
\usepackage{lmodern}
\usepackage[a4paper,margin=1in]{geometry}
\usepackage{amsmath,amssymb,amsthm,mathtools,mathrsfs}
\usepackage{array,booktabs}
\usepackage{graphicx,caption}
\usepackage{float}
\usepackage{needspace}
\usepackage{placeins}
\usepackage{xcolor}
\usepackage{enumitem}
\usepackage{tikz}
\usetikzlibrary{arrows.meta,positioning,fit}
\usepackage{hyperref}

\definecolor{LinkBlue}{RGB}{0,76,153}
\hypersetup{
  colorlinks=true,
  linkcolor=LinkBlue,
  citecolor=LinkBlue,
  urlcolor=LinkBlue,
  linktoc=all,
  bookmarksopen=true,
  bookmarksnumbered=true,
  pdftitle={Systematic construction of multiqubit unextendible product bases},
  pdfsubject={Unextendible product bases and unextendible orthogonal matrices},
  pdfkeywords={unextendible product basis, unextendible orthogonal matrix,
  size spectrum, one-factorization, matching, finite verification}
}

\setlist{itemsep=.25ex,topsep=.5ex}
\newtheorem{theorem}{Theorem}
\newtheorem{lemma}[theorem]{Lemma}
\newtheorem{proposition}[theorem]{Proposition}
\newtheorem{corollary}[theorem]{Corollary}
\theoremstyle{definition}
\newtheorem{definition}[theorem]{Definition}
\newtheorem{example}[theorem]{Example}

\theoremstyle{remark}
\newtheorem{remark}[theorem]{Remark}
\newtheorem{question}[theorem]{Question}
\newtheorem{memo}[theorem]{Memo}
\theoremstyle{plain}

\newtheorem{conjecture}[theorem]{Conjecture}

\newcommand{\R}{\mathcal R}
\newcommand{\F}{\mathcal F}
\newcommand{\Sfam}{\mathscr S}
\newcommand{\bcap}{\operatorname{bcap}}

\newcommand{\Comp}{\operatorname{Comp}}
\newcommand{\Sym}{\operatorname{Sym}}

\newcommand{\disj}{\mathbin{\dot\cup}}

\newcommand{\certsite}{\href{https://github.com/chechche123/upb-research}{online resource}~\cite{UPBOnlineSource}}

\newcommand{\bcj}{\begin{conjecture}}
\newcommand{\ecj}{\end{conjecture}}
\newcommand{\bcr}{\begin{corollary}}
\newcommand{\ecr}{\end{corollary}}
\newcommand{\bd}{\begin{definition}}
\newcommand{\ed}{\end{definition}}
\newcommand{\bea}{\begin{eqnarray}}
\newcommand{\eea}{\end{eqnarray}}
\newcommand{\bem}{\begin{enumerate}}
\newcommand{\eem}{\end{enumerate}}
\newcommand{\bex}{\begin{example}}
\newcommand{\eex}{\end{example}}
\newcommand{\bim}{\begin{itemize}}
\newcommand{\eim}{\end{itemize}}
\newcommand{\bl}{\begin{lemma}}
\newcommand{\el}{\end{lemma}}
\newcommand{\bma}{\begin{bmatrix}}
\newcommand{\ema}{\end{bmatrix}}
\newcommand{\bpf}{\begin{proof}}
\newcommand{\epf}{\end{proof}}
\newcommand{\bpp}{\begin{proposition}}
\newcommand{\epp}{\end{proposition}}
\newcommand{\bqu}{\begin{question}}
\newcommand{\equ}{\end{question}}
\newcommand{\br}{\begin{remark}}
\newcommand{\er}{\end{remark}}
\newcommand{\bt}{\begin{theorem}}
\newcommand{\et}{\end{theorem}}
\newcommand{\bmm}{\begin{memo}}
\newcommand{\emm}{\end{memo}}
\newcommand{\btb}{\begin{tabular}}
\newcommand{\etb}{\end{tabular}}
\newcommand{\beq}{\begin{equation}}
\newcommand{\eeq}{\end{equation}}
\newcommand{\bal}{\begin{aligned}}
\newcommand{\eal}{\end{aligned}}

\let\MentorTexta\a
\let\MentorTextb\b
\let\MentorTextc\c
\let\MentorTextd\d
\let\MentorTexti\i
\let\MentorTextk\k
\let\MentorTextl\l
\let\MentorTextr\r
\let\MentorTextt\t
\let\MentorTextu\u
\let\MentorTextL\L
\let\MentorTextO\O
\let\MentorTextP\P
\let\MentorTextS\S
\DeclareRobustCommand{\a}{\ifmmode\alpha\else\expandafter\MentorTexta\fi}
\DeclareRobustCommand{\b}{\ifmmode\beta\else\expandafter\MentorTextb\fi}

\DeclareRobustCommand{\d}{\ifmmode\delta\else\expandafter\MentorTextd\fi}
\newcommand{\e}{\epsilon}

\DeclareRobustCommand{\t}{\ifmmode\theta\else\expandafter\MentorTextt\fi}
\DeclareRobustCommand{\i}{\ifmmode\iota\else\expandafter\MentorTexti\fi}
\DeclareRobustCommand{\k}{\ifmmode\kappa\else\expandafter\MentorTextk\fi}
\DeclareRobustCommand{\l}{\ifmmode\lambda\else\expandafter\MentorTextl\fi}

\newcommand{\x}{\xi}

\DeclareRobustCommand{\r}{\ifmmode\rho\else\expandafter\MentorTextr\fi}

\DeclareRobustCommand{\u}{\ifmmode\upsilon\else\expandafter\MentorTextu\fi}

\DeclareRobustCommand{\c}{\ifmmode\chi\else\expandafter\MentorTextc\fi}

\DeclareRobustCommand{\L}{\ifmmode\Lambda\else\expandafter\MentorTextL\fi}
\newcommand{\X}{\Xi}
\DeclareRobustCommand{\P}{\ifmmode\Pi\else\expandafter\MentorTextP\fi}
\DeclareRobustCommand{\S}{\ifmmode\Sigma\else\expandafter\MentorTextS\fi}

\DeclareRobustCommand{\O}{\ifmmode\Omega\else\expandafter\MentorTextO\fi}

\newcommand{\Dbar}{\leavevmode\lower.6ex\hbox to 0pt{\hskip-.23ex\accent"16\hss}D}

\begin{document}

\title{Systematic construction of multiqubit unextendible product bases}

\author{%
Caohan Cheng$^{1,*}$ and
Lin Chen$^{1,*,\dagger}$\\
{\normalsize\itshape
$^{1}$LMIB(Beihang University), Ministry of Education,\\
School of Mathematical Sciences, Beihang University,
Beijing 100191, China\\
{\upshape$^{\dagger}$Corresponding author:
\href{mailto:linchen@buaa.edu.cn}{linchen@buaa.edu.cn}}}
}

\date{}

\maketitle

\begin{abstract}
The construction of multiqubit unextendible product bases (UPBs) has been
a long-standing problem.  We give a systematic construction of multiqubit
UPBs by studying their unextendible orthogonal matrices (UOMs) with
graph-theoretic methods.  Row decompositions of small UOMs produce the main
intervals.  Completion graphs settle the remaining small UPB sizes, except
for one size obtained from a good column structure.  Nonexistence of the
other UPB sizes follows from counts of fibres and row pairs.  The
computer is used only for a few finite UOM checks, whose exact tests are
stated in the certificate appendix.  These arguments determine the complete
multiqubit UPB size spectrum in every dimension.
\end{abstract}

Keyword: multiqubit system, unextendible product basis, graph theory

\section{Introduction}
\label{sec:introduction}

Quantum-information processing such as teleportation \cite{Bennett1993Teleporting} relies heavily on multiqubit systems \cite{Horodecki2009Quantum}. The realization of multiqubit states has received a lot of attentions in past decades \cite{Ransford_2026}. Theoretically, multiqubit positive-partial-transpose (PPT) entangled states \cite{Horodecki1998Mixed} can be constructed by multiqubit unextendible product bases (UPBs).
Actually, UPBs were firstly proposed to show the nonlocality without entanglement \cite{bennett1999quantum,bennett1999unextendible,divincenzo2003unextendible}. Bipartite UPBs have been used to construct genuinely entangled subspaces \cite{demianowicz2018unextendible,wang20194,Demianowicz2024Completely}, which further help construct distillable entanglement across every bipartition  \cite{Agrawal2019Genuinely}. Besides, 
multipartite UPBs can be used to construct completely entangled subspace \cite{bhat2006completely} and study Bell inequalities with no quantum violation  \cite{Augusiak2011Bell,PhysRevA.85.042113}. Recently, the so-called  
strong quantum non-
locality for UPB has been proposed in heterogeneous systems \cite{Shi2022Strong,Shi2022Strongly,He2024Strong}.

The construction of multiqubit UPBs have received lots of attentions over the past decades. Three-qubit UPBs of size four are the firstly constructed three-qubit UPB 
\cite{bravyi2004unextendible}, used to construct some three-qubit PPT entangled states of rank four. Further, four-qubit PPT entangled symmetric states have been constructed  \cite{tura2012four}. On the other hand, multiqubit UPBs of minimum size and up to seven qubits have been also studied \cite{johnston2013minimum, johnston2014structure,chen2015minimum}. Further, a full classification of orthogonal product bases of four-qubit system has been given \cite{chen2017orthogonal}, continued with a method of formally orthogonal matrices  \cite{chen2018multiqubit}. Next, it has been shown that $n$-qubit unextendible
product bases of size $2^n-5$ do not exist \cite{chen2018nonexistence}. Seven-qubit UPB of size ten has been constructed \cite{wang2020construction}. UPBs can also be built from tile structures and their local entanglement-assisted distinguishability \cite{Shi2020Unextendible,You2023Unextendible}. In spite of these achievements, a complete determination of the possible multiqubit UPB sizes is still missing and remains a long-standing challenge. This is the main motivation and problem studied in this paper. The construction of UPBs has been intimately related to graph theory \cite{Alon2001Unextendible,feng2006unextendible,shi2023graph}, and we will present some novel notions from graph theory for the construction of UPBs.

For the convenience of readers, we briefly introduce the findings of this
paper.  An $m\times n$ UOM is the formal-matrix representation of an
$n$-qubit UPB with $m$ product vectors; the precise definition and its
relation with UPBs are given in Section~\ref{sec:preliminaries}
\cite{chen2018multiqubit}.  For every positive integer $n$, we set
\begin{equation}
 \Theta_n=\{m\in\mathbb Z_{>0}:\text{there is an }m\times n\text{ UOM}\}.
 \label{eq:intro-theta}
\end{equation}
Every $m\in\Theta_n$ satisfies $n+1\le m\le 2^n$.  In this paper a complete
orthogonal product basis is also regarded as a UOM, so that
$2^n\in\Theta_n$ \cite{chen2017orthogonal}.  If the term UPB is used only
for an incomplete basis, one should delete $2^n$ from each spectrum below.

For integers $a\leq b$, write
$[a,b]=\{a,a+1,\ldots,b\}$; for a positive integer $s$, abbreviate
$[s]=[1,s]$.  Every interval throughout this paper is an integer interval.

Several parts of the answer were known earlier.  For $n\ge9$, the minimum
size is $n+1$ for odd $n$, $n+2$ when $n\equiv2\pmod4$, and $n+4$ when
$n\equiv0\pmod4$ \cite{johnston2013minimum}.  In the odd case, the size
$n+2$ is impossible \cite{johnston2014structure}.  At the other end of the
spectrum, for $n\ge3$, the sizes $2^n-4$ and $2^n$ occur, whereas
$2^n-5,2^n-3,2^n-2$, and $2^n-1$ do not occur
\cite{johnston2014structure,chen2018nonexistence}.
The spectra for $n\le4$ were determined in
\cite{johnston2014structure}.  Chen and \DJ okovi\'c completed the cases
$n=5,6,8$ in \cite{chen2018multiqubit}.  For $n=7$, the last two missing
sizes, $10$ and $11$, were supplied in
\cite{wang2020construction,SunWang2024}.  The new constructions and
obstructions in this paper determine all the sizes left between these known
endpoints.  The first display below records the earlier results; the remaining
three formulas are completed in this paper.

\begin{theorem}
\label{thm:complete-spectrum}
\begin{enumerate}[label=(\roman*)]
\item
For $1\le n\le8$,
\begin{equation}
\begin{aligned}
\Theta_1&=\{2\},\qquad
\Theta_2=\{4\},\qquad
\Theta_3=\{4,8\},\\
\Theta_4&=[6,10]\cup\{12,16\},\\
\Theta_5&=\{6\}\cup[8,26]\cup\{28,32\},\\
\Theta_6&=[8,58]\cup\{60,64\},\\
\Theta_7&=\{8\}\cup[10,122]\cup\{124,128\},\\
\Theta_8&=[11,250]\cup\{252,256\}.
\end{aligned}
 \label{eq:small-spectra}
\end{equation}
\item
For odd $n\ge9$,
\begin{equation}
 \begin{aligned}
 \Theta_n
  &=\{n+1\}\cup[n+3,2^n-6]\\
  &\quad\cup\{2^n-4,2^n\}.
 \end{aligned}
 \label{eq:odd-spectrum}
\end{equation}
\item
For $n\ge10$ with $n\equiv2\pmod4$,
\begin{equation}
 \begin{aligned}
 \Theta_n
  &=\{n+2\}\cup[n+4,2^n-6]\\
  &\quad\cup\{2^n-4,2^n\}.
 \end{aligned}
 \label{eq:two-spectrum}
\end{equation}
\item
For $n\ge12$ with $n\equiv0\pmod4$,
\begin{equation}
 \begin{aligned}
 \Theta_n
  &=[n+4,2^n-6]\\
  &\quad\cup\{2^n-4,2^n\}.
 \end{aligned}
 \label{eq:zero-spectrum}
\end{equation}
\end{enumerate}
\qed
\end{theorem}

Theorem~\ref{thm:complete-spectrum} determines the possible cardinalities
for every number of qubits.  The equivalence classes at each cardinality
remain a separate problem.  The earlier continuous interval in
\cite[Theorem~4]{johnston2014structure} starts at
$(n^2+3n-30)/2$ for $n\ge7$.
In \eqref{eq:odd-spectrum}--\eqref{eq:zero-spectrum}, the continuous
interval starts at $n+3$, $n+4$, and $n+4$, respectively.
Each new starting point is exact: the preceding integer is excluded by
the known minimum or parity result, or by
Theorems~\ref{thm:even-obstruction} and \ref{thm:ten-boundary}.
Thus the remaining problem is at the lower end, even though the earlier
interval already contains almost all sizes as $n$ grows.

Theorem~\ref{thm:row-decomposition-construction} gives a condition under
which row decompositions can be joined.  Its part~(ii) gives a necessary
and sufficient condition when the added columns come from matchings.
The condition uses the joint block capacity in
\eqref{eq:joint-capacity}.
It records which blocks can be covered in the original columns with
disjoint column supports.  The two-input construction for
$(N+1)\times N$ UOMs in \cite[Proposition~14]{chen2018multiqubit}
is the case with singleton blocks and matching added columns.
Theorem~\ref{thm:row-decomposition-construction} permits larger blocks,
more inputs, and weighted added columns.
These choices give the two-, four-, and six-block constructions.

Theorem~\ref{thm:square-transfer}(i) gives a second construction.
A maximal completion clique of size $c$ for a square kernel of odd order
$p$ gives a $(q+c)\times q$ UOM for every $q\in Q_p$ in
\eqref{eq:Qp}.  The proof shows that the extended clique is still maximal
in the larger completion graph.  It also gives a decomposition satisfying
the block cover condition in Definition~\ref{def:block-cover-condition}.
Part~(ii) combines two such inputs and reduces the number of added columns.
Its capacity bound is an application of
Theorem~\ref{thm:row-decomposition-construction}(i).
The completion-clique sizes therefore give whole ranges of UOM sizes.
The padding and good-column constructions supply the remaining endpoints.

Theorems~\ref{thm:even-obstruction} and \ref{thm:ten-boundary} give the
new exclusion $n+3\notin\Theta_n$ for $n\ge10$ with
$n\equiv2\pmod4$.  This exclusion applies to every UOM.
The fibre-union bound in Lemma~\ref{even:lem-union-inequality} restricts
the column types.  Their capacities bound the total number of orthogonal
row pairs.  The equality cases require a further count of pairs repeated
in different columns.  The finite matching checks enter only after these
reductions.  These exclusions, the existence constructions, and the known
endpoint results give the complete spectrum.

For $n\ge10$ with $n\equiv2\pmod4$, restricting
\eqref{eq:two-spectrum} to the interval $[n+2,2^n-6]$ gives
\[
 \Theta_n\cap[n+2,2^n-6]
 =[n+2,2^n-6]\setminus\{n+3\}.
\]
Hence $n+3$ is the only missing integer between the minimum $n+2$ and the
last point $2^n-6$ of the continuous interval.  The two further sizes
$2^n-4$ and $2^n$ are the isolated upper points in
\eqref{eq:two-spectrum}.
Figure~\ref{fig:intro-spectrum-atlas} summarizes the three spectra. Further, Table~\ref{tab:intro-result-guide} gives a short guide to the main methods.
The last column points to the figures which explain the corresponding
construction or synthesis.

\begin{table}[htbp]
\centering
\renewcommand{\arraystretch}{1.12}
\begin{tabular}{@{}
>{\raggedright\arraybackslash}p{.25\textwidth}
>{\raggedright\arraybackslash}p{.37\textwidth}
>{\raggedright\arraybackslash}p{.27\textwidth}@{}}
\toprule
Method & Main result & Visual guide \\
\midrule
Completion graphs
 & Lemma~\ref{lem:full-link}
 & Figures~\ref{fig:completion-example} and \ref{fig:full-link-picture} \\
Row decompositions
 & Theorem~\ref{thm:row-decomposition-construction}
 & Figures~\ref{fig:two-added-columns} and \ref{twon:fig-four-group-lift} \\
Square kernels
 & Theorem~\ref{thm:square-transfer}
 & Figures~\ref{fig:group-perfect-square-core-lift} and
   \ref{fig:group-perfect-two-parent-join} \\
Odd padding
 & Theorem~\ref{thm:odd-padding}
 & Figure~\ref{fig:group-perfect-universal-odd-lift} \\
Interval synthesis
 & Theorem~\ref{thm:existence-intervals}
 & Figure~\ref{fig:proof-architecture} \\
\bottomrule
\end{tabular}
\caption{Where the principal constructions and their pictures are used.}
\label{tab:intro-result-guide}
\end{table}

\begingroup
\newcommand{\SpecDot}[2]{\fill (#1,#2) circle[radius=1.8pt];}
\newcommand{\SpecCross}[2]{%
  \draw[line width=.75pt] (#1-.085,#2-.085) -- (#1+.085,#2+.085);
  \draw[line width=.75pt] (#1-.085,#2+.085) -- (#1+.085,#2-.085);}
\newcommand{\SpecTick}[3]{%
  \draw[line width=.38pt] (#1,#2-.08) -- (#1,#2+.08);
  \node[tick] at (#1,#2-.28) {#3};}
\newcommand{\LowerAxis}[1]{%
  \draw[axis] (-3.55,#1) -- (2.35,#1);
  \draw[axis] (2.75,#1) -- (6.45,#1);
  \node[font=\rmfamily\scriptsize,inner sep=0pt] at (2.55,#1) {$\cdots$};
  \SpecTick{-3.15}{#1}{$n+1$}
  \SpecTick{-1.90}{#1}{$n+2$}
  \SpecTick{-.65}{#1}{$n+3$}
  \SpecTick{.60}{#1}{$n+4$}
  \SpecTick{6.15}{#1}{$2^n-6$}}
\newcommand{\LowerInterval}[2]{%
  \draw[present] (#1,#2) -- (2.35,#2);
  \draw[present] (2.75,#2) -- (6.15,#2);
  \SpecDot{#1}{#2}
  \SpecDot{6.15}{#2}}

\begin{figure}[!t]
\centering
\begin{tikzpicture}[
  x=.92cm,y=1cm,
  font=\rmfamily\small,
  axis/.style={draw=black!65,line width=.42pt},
  present/.style={draw=black,line width=2.2pt,line cap=butt},
  tick/.style={font=\rmfamily\scriptsize,inner sep=0pt},
  class/.style={font=\rmfamily\bfseries\small,anchor=east,align=right},
  panel/.style={font=\rmfamily\bfseries\small,anchor=west}
]
\node[panel] at (-6.55,4.55) {(a) Lower end, separated by the parity of $n$};

\def\yodd{3.55}
\node[class] at (-4.05,\yodd) {odd $n\ge9$};
\LowerAxis{\yodd}
\SpecDot{-3.15}{\yodd}
\SpecCross{-1.90}{\yodd}
\LowerInterval{-.65}{\yodd}

\def\ytwo{2.25}
\node[class] at (-4.05,\ytwo) {$n\equiv2\pmod4$\\$n\ge10$};
\LowerAxis{\ytwo}
\SpecCross{-3.15}{\ytwo}
\SpecDot{-1.90}{\ytwo}
\SpecCross{-.65}{\ytwo}
\LowerInterval{.60}{\ytwo}

\def\yzero{.95}
\node[class] at (-4.05,\yzero) {$n\equiv0\pmod4$\\$n\ge12$};
\LowerAxis{\yzero}
\SpecCross{-3.15}{\yzero}
\SpecCross{-1.90}{\yzero}
\SpecCross{-.65}{\yzero}
\LowerInterval{.60}{\yzero}

\node[panel] at (-6.55,-.05) {(b) The common upper end for all $n\ge9$};
\def\yupper{-1.05}
\draw[axis] (-3.45,\yupper) -- (6.45,\yupper);
\SpecTick{-3.15}{\yupper}{$2^n-6$}
\SpecTick{-1.60}{\yupper}{$2^n-5$}
\SpecTick{-.05}{\yupper}{$2^n-4$}
\SpecTick{1.50}{\yupper}{$2^n-3$}
\SpecTick{3.05}{\yupper}{$2^n-2$}
\SpecTick{4.60}{\yupper}{$2^n-1$}
\SpecTick{6.15}{\yupper}{$2^n$}
\SpecDot{-3.15}{\yupper}
\SpecCross{-1.60}{\yupper}
\SpecDot{-.05}{\yupper}
\SpecCross{1.50}{\yupper}
\SpecCross{3.05}{\yupper}
\SpecCross{4.60}{\yupper}
\SpecDot{6.15}{\yupper}
\end{tikzpicture}
\caption{The exact spectra for $n\ge9$.  Panel (a) shows the lower end and
the continuous interval ending at $2^n-6$; the centered ellipsis suppresses
the long middle part of each axis.  Panel (b) shows the upper end, which is the same in all three
parity classes.  A thick segment contains every integer between its
endpoints, a filled point is included, and a cross is excluded.  These are
the three formulas in \eqref{eq:odd-spectrum}--\eqref{eq:zero-spectrum}.}
\label{fig:intro-spectrum-atlas}
\end{figure}
\endgroup
\FloatBarrier

\begingroup
The rest of the paper is organized as follows.
Section~\ref{sec:preliminaries} introduces UOMs and fibres in
Definitions~\ref{def:uom}(iii) and \ref{def:equality-fibre}(i), proves the fibre
cover test in Theorem~\ref{thm:fibre-cover}, and collects the graph tools in
Theorem~\ref{thm:standard-graph-tools} before defining row decompositions,
added column systems, completion graphs, and square kernels in
Definitions~\ref{def:row-decomposition}, \ref{def:added-column-system},
\ref{def:completion-graph}(ii), and \ref{def:square-kernel}.
Figures~\ref{fig:completion-example} and \ref{fig:full-link-picture} show
the completion-graph correspondence.
Section~\ref{sec:master-constructions} proves the row-decomposition
construction in Theorem~\ref{thm:row-decomposition-construction} and the
square-kernel transfer in Theorem~\ref{thm:square-transfer}, including its
compressed two-input form in part~(ii).  These constructions are pictured in
Figures~\ref{fig:two-added-columns},
\ref{fig:group-perfect-square-core-lift}, and
\ref{fig:group-perfect-two-parent-join}.
Section~\ref{sec:existence-synthesis} records the finite inputs in
Lemmas~\ref{lem:five-kernel-cliques} and
\ref{lem:certified-row-decompositions}, then combines the square-kernel
ranges, padding in Theorem~\ref{thm:odd-padding}, the six-block construction,
the \(n+5\) construction, and the upper interval in
Corollaries~\ref{cor:low-kernel-intervals} and
\ref{cor:padded-endpoints}, Proposition~\ref{prop:six-block-interval}, and
Theorems~\ref{thm:n-plus-five-endpoint} and \ref{thm:uniform-tail}.  The
result is Theorem~\ref{thm:existence-intervals};
Figure~\ref{fig:group-perfect-universal-odd-lift} shows the padding step,
and Figure~\ref{fig:proof-architecture} shows how the component ranges
overlap.
Section~\ref{sec:nonexistence} establishes the two new obstructions in
Theorems~\ref{thm:even-obstruction} and \ref{thm:ten-boundary}.
Section~\ref{sec:finite-ledger} handles the remaining finite dimensions in
Theorems~\ref{thm:finite-new-certificates} and
\ref{thm:finite-boundary-spectra}, and completes the spectra.

Appendix~\ref{app:row-decomposition-constructions} gives the two-, four-, and
six-block constructions and the upper interval; see
Figures~\ref{twon:fig-four-group-lift} and
\ref{fig:row-replacement-picture}.  Appendix~\ref{app:square-kernel}
proves the square-kernel lemmas and odd padding.  Appendix~\ref{app:good-column-lift}
proves the good-column lift and verifies its \(13\times8\) initial matrix.
Appendices~\ref{even:sec-global-obstruction} and \ref{sec:13-by-10} contain
the full proofs of Theorems~\ref{thm:even-obstruction} and
\ref{thm:ten-boundary}, respectively.  Appendix~\ref{app:certificate-contracts}
states the exact finite tests; Appendix~\ref{app:explicit-matrices} gives
three matrices.  The data and three programs are in the \certsite.
\par\endgroup
\section{Preliminaries}
\label{sec:preliminaries}

The fibre cover test in Theorem~\ref{thm:fibre-cover} is the main tool
for unextendibility.  It uses row orthogonality from
Definition~\ref{def:uom}(i) and fibres from
Definition~\ref{def:equality-fibre}(i).
Definition~\ref{def:graph-notation}(i)--(iv) fixes the graph terms.
The remaining notions prepare the two constructions:
row decompositions with added columns, and square kernels with completion graphs.

For a positive integer $k$, put
\begin{equation}
 \mathbb Z_k=\{0,\ldots,k-1\},
 \label{eq:basic-index-notation}
\end{equation}
where all calculations in $\mathbb Z_k$ are taken modulo $k$.  For a finite
set $V$, let $\Sym(V)$ denote its permutation group.  When $V=[k]$, write
$\Sym_k=\Sym([k])$; thus $\Sym_k$ permutes the whole set $[k]$, not the
one-element set $\{k\}$.
The notation $A\disj B$ is used for the union of two disjoint sets.  If $t$
is a positive integer and $A,B$ are sets of integers, put
\[
 A+B=\{a+b:a\in A,\ b\in B\},
 \qquad
 tA=\underbrace{A+\cdots+A}_{t\text{ times}}.
\]

\subsection{Formal matrices and fibres}

Every column has its own alphabet.  This is the set of symbols allowed in
that column.  The symbols are divided into pairs.  We call them mate pairs.
If $x$ is one symbol in a pair, the other one is denoted by $\bar x$.
When more symbols are needed, take a pair not used earlier in this column.
The mate rule is
\begin{equation}
 \bar{\bar x}=x,
 \qquad
 \bar x\ne x.
 \label{eq:mate-rule}
\end{equation}
Alphabets in different columns are separate.  The same printed name in two
columns does not give any relation between their symbols.

\medskip
\noindent\emph{Formal rows.}
A formal row with $N$ columns is a tuple $r=(r_1,\ldots,r_N)$.  Its entry
$r_j$ belongs to the alphabet of column $j$.  A list of formal rows is a
formal matrix.  The row set of a formal matrix $X$ is denoted by $R(X)$.

\begin{definition}
\label{def:uom}
\begin{enumerate}[label=(\roman*)]
\item Two rows $r$ and $s$ are orthogonal, written $r\perp s$, if
\begin{equation}
 r_j=\overline{s_j}
 \qquad\text{for some column }j.
 \label{eq:row-orthogonal}
\end{equation}
\item A formal matrix is orthogonal if every two different rows are
orthogonal.
\item An orthogonal formal matrix $X$ is an \emph{unextendible orthogonal
matrix} (UOM) if
\[
 \nexists\,y\quad \text{such that} \quad y\perp r\text{ for every }r\in R(X).
\]
\end{enumerate}
\qed
\end{definition}

\medskip
\noindent\emph{Realizations.}
A formal matrix records only the mate relation.  It does not yet assign
qubit states.  A realization supplies this assignment.  Each symbol $x$ in
column $j$ is replaced by a qubit ray $\phi_j(x)$, where a qubit ray means a
one-dimensional subspace of $\mathbb C^2$.  Mate symbols must give orthogonal
rays:
\[
 \phi_j(\bar x)=\phi_j(x)^\perp.
\]
The realization is called generic if
\[
 \phi_j(x)\ne\phi_j(y),
 \qquad
 \phi_j(x)\not\perp\phi_j(y)
\quad\text{if }y\notin\{x,\bar x\}.
\]
Under a generic realization, the rows of a UOM form a multiqubit UPB.
Conversely, every multiqubit UPB has a UOM description
\cite{chen2018multiqubit}.  The rest of the paper uses these formal matrices.

The next definition records the row sets that one entry of an extension
can cover.  It leads to Theorem~\ref{thm:fibre-cover}.

\begin{definition}
\label{def:equality-fibre}
\begin{enumerate}[label=(\roman*)]
\item Let $X$ have $N$ columns. For a symbol $a$ in the alphabet of column
$j$, we put
\begin{equation}
\begin{aligned}
 F_{X,j}(a)&=\{r\in R(X):r_j=a\},\\
 \F_{X,j}&=\{F_{X,j}(a):a\text{ is in the alphabet of column }j,\\[-1mm]
 &\hspace{42mm}F_{X,j}(a)\ne\varnothing\}.
\end{aligned}
 \label{eq:equality-fibre}
\end{equation}
We call $F_{X,j}(a)$ the fibre of $a$ in column $j$.  It may be empty.
The family $\F_{X,j}$ consists of the nonempty fibres in column $j$.
For fixed $j$, its members form a partition of $R(X)$:
\[
 R(X)=\bigsqcup_{F\in\F_{X,j}}F.
\]

\item For $T\subseteq R(X)$, a set $S\subseteq[N]$ is a fibre support of
$T$ if one fibre from each column in $S$ can be chosen, and they cover
$T$. The family of all fibre supports is
\begin{equation}
 \Sfam_X(T)
  =\bigl\{S\subseteq[N]:
       \exists F_j\in\F_{X,j}\ (j\in S),\ 
       T\subseteq\bigcup_{j\in S}F_j\bigr\}.
\label{eq:support-family}
\end{equation}

\Needspace{6\baselineskip}
\item If $T$ has a fibre support, its cover cost is the least size of one:
\begin{equation}
 \kappa_X(T)=\min\{|S|:S\in\Sfam_X(T)\}.
 \label{eq:cover-cost}
\end{equation}
If $T$ has no fibre support, put $\kappa_X(T)=\infty$.
\end{enumerate}
\qed
\end{definition}

A row $r$ belongs to the fibre $F_{X,j}(r_j)$.
The family $\F_{X,j}$ contains these fibres, rather than their individual
rows; see Definition~\ref{def:equality-fibre}(i).

Here is an example of fibre support. Let
$X=\left(\begin{smallmatrix}a&b\\a&\bar b\\\bar a&\bar b\end{smallmatrix}\right)$.
Number its rows by $r_1,r_2,r_3$.  Take $T=R(X)$.  Then
$F_{X,1}(a)=\{r_1,r_2\}$ and
$F_{X,2}(\bar b)=\{r_2,r_3\}$.  These two fibres cover $T$, so
$S=\{1,2\}$ is a fibre support of $T$ by
Definition~\ref{def:equality-fibre}(ii).  No single fibre covers
all three rows.  Thus $\kappa_X(T)=2$. Note that the set $S$ consists of column numbers.

Next, a formal row $y=(y_1,\ldots,y_N)$ extends $X$ if $y\perp r$ for every
$r\in R(X)$, as in Definition~\ref{def:uom}(iii).  Fix column $j$.  The entry
$y_j$ has only one mate, $\bar y_j$.  Thus, in this column, $y$ is
orthogonal exactly to the rows in $F_{X,j}(\bar y_j)$; this fibre may be
empty.  This fact gives the following test.

\Needspace{7\baselineskip}
\begin{theorem}
\phantomsection
\label{thm:fibre-cover}
\noindent\textup{(i)}
An orthogonal $m\times N$ matrix $X$ is extendible if and only if the
following condition holds:

\begin{equation}
\begin{gathered}
 R(X)=\bigcup_{j\in J}F_{X,j}(a_j)
 \quad\text{for some }J\subseteq[N],\\
 \text{with at most one fibre from each column}.
\end{gathered}
 \label{eq:fibre-cover}
\end{equation}

\medskip
\noindent\textup{(ii)}
Put $\R_0=\{\varnothing\}$. We introduce two families
\[
\begin{aligned}
 \R_1={}&\{\varnothing\}\cup\F_{X,1},\\
 \R_2={}&\{\varnothing\}\cup\F_{X,1}\cup\F_{X,2}\\
 &\quad\cup\{F_1\cup F_2:
 F_1\in\F_{X,1},\ F_2\in\F_{X,2}\}.
\end{aligned}
\]
In general,
\[
 \R_j=
 \left\{\bigcup_{\ell\in S}F_\ell:
 S\subseteq[j],\ F_\ell\in\F_{X,\ell}\text{ for every }\ell\in S\right\},
\]
where the union is $\varnothing$ when $S=\varnothing$.
Equivalently, these families satisfy
\begin{equation}
 \R_j=\R_{j-1}\cup
 \{U\cup F:U\in\R_{j-1},\ F\in\F_{X,j}\},
 \quad 1\le j\le N.
 \label{eq:fibre-dp}
\end{equation}
Then
\begin{equation}
 X\text{ is a UOM}
 \quad\Longleftrightarrow\quad
 X\text{ is orthogonal and }R(X)\notin\R_N.
 \label{eq:fibre-dp-uom}
\end{equation}
\end{theorem}

\begin{proof}
For (i), an extension row $y$ gives the cover in \eqref{eq:fibre-cover}.
Indeed, \eqref{eq:row-orthogonal} puts each row $r$ in
$F_{X,j}(\bar y_j)$ for some $j$.
These fibres therefore cover $R(X)$.

Conversely, suppose that \eqref{eq:fibre-cover} holds.
Set $y_j=\bar a_j$ in each selected column.
Fill every unused position with a symbol from a new mate pair.
The chosen fibres show that $y$ is orthogonal to every row of $X$.

For (ii), \eqref{eq:fibre-dp} lists exactly the unions obtained from the
first $j$ columns.  This is the induction claim.
It holds at $j=0$.
If column $j$ is unused, the union remains in $\R_{j-1}$.
Otherwise, it is $U\cup F$ for some $U\in\R_{j-1}$ and
$F\in\F_{X,j}$.
These are exactly the choices in \eqref{eq:fibre-dp}, so induction gives
the claim.  Part (i) now gives \eqref{eq:fibre-dp-uom}.
\end{proof}

\Needspace{10\baselineskip}
\begin{example}
Let
\[
 X_0=\begin{pmatrix}a&b\\ a&\bar b\end{pmatrix}.
\]
The rows are orthogonal in the second column. In the first column,
$F_{X_0,1}(a)=R(X_0)$, so Theorem~\ref{thm:fibre-cover}(i) says that $X_0$ is
extendible.  Indeed, if $c$ is a new symbol in the second column, then
$(\bar a,c)$ is orthogonal to both rows.  Thus an orthogonal matrix need not
be a UOM.
\end{example}

In the finite calculation, repeated row sets in \eqref{eq:fibre-dp}
are listed once.  Appendix~\ref{app:certificate-contracts} gives the
implementation of \eqref{eq:fibre-dp-uom} and the checks used later.

\subsection{Graph notations and results}

All graphs in this paper are simple and undirected.  Those used in the
following results are finite.

\medskip
\noindent\emph{Basic graph notation.}
For a graph $G$, its vertex and edge sets are denoted by $V(G)$ and $E(G)$.
The degree of a vertex $v$ is $\deg_G(v)$.  The maximum vertex degree is
$\Delta(G)$.  When $u$ and $v$ are joined, write $u\sim v$; the edge
$\{u,v\}$ may also be written as $uv$.  For $U\subseteq V(G)$, the set
$N_G(U)$ contains the vertices joined to at least one vertex of $U$.  A
$d$-regular graph has degree $d$ at every vertex.

On a finite vertex set $V$, the complete graph is denoted by $K_V$.  We use
$K_n$ when the vertex set is $[n]$.  The complete bipartite graph with parts
$A$ and $B$ is $K_{A,B}$.  If $|A|=r$ and $|B|=s$, the notation
$K_{r,s}[A,B]$ may be used instead.  The symbol $K_{q_1,\ldots,q_t}$ denotes
the complete multipartite graph with part sizes $q_1,\ldots,q_t$.  Its
vertices are joined exactly when they lie in different parts.

If $H$ is a subgraph of $G$, the graph $G-H$ is obtained by deleting the
edges of $H$.  The notation $mG$ means the disjoint union of $m$ copies of
$G$, while $C_\ell$ is a cycle of length $\ell$.  An even cycle with parts
$A$ and $B$ may also be written as $C_{2\ell}[A,B]$.

\begin{definition}
\label{def:graph-notation}
\begin{enumerate}[label=(\roman*)]
\item A matching is a set of edges with no common endpoints.  A matching
covering every vertex is a one-factor, or perfect matching.  A
one-factorization partitions $E(G)$ into one-factors. When $|V(G)|$ is odd,
a near-one-factor misses one vertex, and a near-one-factorization partitions
$E(G)$ into such factors.

\item A clique is a set of pairwise joined vertices. 
It is inclusion-maximal if no outside vertex can be added.  Throughout the paper,
``maximal'' has this meaning, not largest.  The order of a clique $C$ is
$|C|$.

\item An edge colouring assigns a colour to every edge.  It is proper if
edges with a common endpoint have different colours.  The edges with one
fixed colour form its colour class.  In a proper edge colouring, every
nonempty colour class is a matching.

The chromatic index $\chi'(G)$
is the least number of colours in a proper edge colouring.

\item For an indexed family $\mathcal E=(E_i)_{i\in I}$ of edge sets on one
vertex set, a transversal matching of size $r$ consists of distinct indices
$i_1,\ldots,i_r$ and pairwise disjoint edges $e_j\in E_{i_j}$.  The indices
remain distinct even if some members of $\mathcal E$ are equal.
\end{enumerate}
\qed
\end{definition}
A transversal matching also requires distinct family indices; see
Definition~\ref{def:graph-notation}(iv).
For example, $\{12,34\}$ is a matching in $E_1=\{12,34\}$.
It cannot give a transversal matching of size two if $E_1$ is the only
indexed set.

The graph $G=K_4$ gives an example of all four parts of the definition.
Its vertex set is a maximal clique of order four by
Definition~\ref{def:graph-notation}(ii).

Put
$M_1=\{12,34\}$, $M_2=\{13,24\}$, and $M_3=\{14,23\}$.  Each $M_i$
is a one-factor.  These three sets form a one-factorization of $G$ by
Definition~\ref{def:graph-notation}(i). Assign colour $i$ to the edges
of $M_i$.  This gives a proper edge colouring with three colour classes by
Definition~\ref{def:graph-notation}(iii).  For
Definition~\ref{def:graph-notation}(iv), take two indexed copies
$E_1=E_2=\{12,34\}$.  Then $12\in E_1$ and $34\in E_2$ form a
transversal matching of size $2$.  The edge sets are equal.  Their indices
are different.

For a family of colour classes, the distinct indices in
Definition~\ref{def:graph-notation}(iv) mean distinct colours.

For later use, we collect the following standard results from graph
theory.

\begin{theorem}
\label{thm:standard-graph-tools}
Let $h,q$ and $t$ be positive integers, and let $d$ be a nonnegative
integer.
\begin{enumerate}[label=(\roman*),itemsep=.15em,parsep=0pt]
\item The graph $K_{2h}$ has a one-factorization into $2h-1$ one-factors;
one may require it to contain any prescribed one-factor $M$
\cite{LovaszPlummer1986,feng2006unextendible}.

\item The graph $K_{2h+1}$ has a near-one-factorization into $2h+1$ factors,
labelled by their missed vertices \cite{LovaszPlummer1986}.

\item Every finite $d$-regular bipartite graph has a one-factorization into
$d$ one-factors \cite{LovaszPlummer1986}.

\item Let $t\ge2$.  If $qt$ is even, then the complete multipartite graph
with equal part sizes
\begin{equation}
 K_{\underbrace{q,\ldots,q}_{t\text{ parts}}}
 \label{eq:equal-multipartite-graph}
\end{equation}
has a one-factorization with $q(t-1)$ one-factors
\cite{HoffmanRodger1992}.

\item Every finite simple graph $G$ has a proper edge colouring with at most
$\Delta(G)+1$ colours \cite{Vizing1964}.

\item Let $\mathcal E$ be a nonempty family of edges in a bipartite graph.
If every two edges in $\mathcal E$ have a common endpoint, then there is
one vertex which is an endpoint of every edge in $\mathcal E$.
\end{enumerate}
\end{theorem}
Parts (i)--(v) are standard.
The prescribed factor in (i) is obtained by relabelling a fixed
one-factorization.  Deleting one vertex from a factorization of
$K_{2h+2}$ gives (ii).

\begin{proof}[Proof of {\rm(vi)}]
The case of at most one edge is immediate.  Otherwise, choose two
edges in $\mathcal E$ with common endpoint $v$.  Every edge in
$\mathcal E$ contains $v$.  Indeed, an edge not containing
$v$ would have to join the other endpoints of these two edges.  Those
endpoints lie in the same part of the bipartite graph, a contradiction.
\end{proof}

\begin{example}
The three sets
\[
 \{12,34\},\qquad \{13,24\},\qquad \{14,23\}
\]
form a one-factorization of $K_4$.  Indeed, every set is a matching which
covers all four vertices, and every edge of $K_4$ appears in exactly one
set.  Apart from the single edge $K_2$, this is the smallest example of
Theorem~\ref{thm:standard-graph-tools} (i).
\end{example}

\begin{lemma}
\label{lem:balanced-color-classes}
Let $G$ be a finite graph and let $k$ be a positive integer.  If $G$ has
a proper edge colouring using at most $k$ colours, then $G$ also has a
proper edge colouring with $k$ colour classes
$P_1,\ldots,P_k$ such that
\begin{equation}
 \bigl||P_i|-|P_j|\bigr|\le1
 \qquad(i,j\in[k]).
 \label{eq:balanced-color-classes}
\end{equation}
where some colour classes are allowed to be empty.
\end{lemma}
\begin{proof}
The proof corrects a violation of \eqref{eq:balanced-color-classes} by
exchanging two colours along a path.  First find a path with one extra edge in the larger class.
Next check that the exchange preserves properness.
Finally, strict decrease of the nonnegative integer in
\eqref{eq:balanced-color-potential} will show that the exchanges stop with
\eqref{eq:balanced-color-classes}.

Indeed, add empty classes if needed, so that there are $k$ classes
$P_1,\ldots,P_k$.
If \eqref{eq:balanced-color-classes} fails, choose two colour classes with
$|P_i|\ge |P_j|+2$.

The exchange needs a path with one more edge of colour $i$ than of
colour $j$.  Consider the subgraph $H$ on $V(G)$ with edge set
\begin{eqnarray}
\label{eq:Pi cup Pj}
P_i\cup P_j.
\end{eqnarray}
Every vertex of $H$ has degree at most two.
Indeed, each colour class is a matching by
Definition~\ref{def:graph-notation}(iii).
Each nontrivial component, meaning one with an edge, is therefore a path
or a cycle.  Properness makes the two colours alternate.
Thus every cycle in this two-colour subgraph has even length.

\Needspace{10\baselineskip}
We claim that there are at least two paths with the required excess. Actually, for each component $C$, put
$d(C)=|E(C)\cap P_i|-|E(C)\cap P_j|$.
Alternation gives $d(C)=0$ for cycles and
$d(C)\in\{-1,0,1\}$ for paths.
Isolated vertices contribute zero.
Each positive summand is one.
The component sum therefore proves the claim:
\begin{equation*}
 \sum_C d(C)=|P_i|-|P_j|\ge2.
 \tag{\ref*{eq:Pi cup Pj}a}\label{eq:component-imbalance}
\end{equation*}

\Needspace{9\baselineskip}
We now change the colouring to reduce the difference between the two class sizes.
Choose one component path $C$ with $d(C)=1$ from
\eqref{eq:component-imbalance}.
On this path, change every edge of colour $i$ to colour $j$.
Change every edge of colour $j$ to colour $i$ at the same time.
All edges outside $C$ keep their colours.
Since $d(C)=1$, this exchange reduces $|P_i|$ by one.
It increases $|P_j|$ by one.

The new colouring is proper by Definition~\ref{def:graph-notation}(iii).
Indeed, each internal vertex of $C$ still meets both colours.
An endpoint meets no other edge of colour $i$ or $j$, because $C$ is a whole
component of $H$.
Edges of other colours are unchanged.

Repeat this change whenever \eqref{eq:balanced-color-classes} fails.
At each step, choose two classes with $|P_i|\ge |P_j|+2$.
To show that this process stops, define
\begin{equation}
 Q(P_1,\ldots,P_k)=\sum_{c=1}^k|P_c|^2.
 \label{eq:balanced-color-potential}
\end{equation}
Put $a=|P_i|$ and $b=|P_j|$ before an exchange.
The two class sizes become $a-1$ and $b+1$.
All other class sizes stay the same.
Thus the decrease in \eqref{eq:balanced-color-potential} is
\[
 a^2+b^2-(a-1)^2-(b+1)^2=2(a-b-1)>0.
\]
The number $Q$ is a nonnegative integer, so it cannot strictly decrease forever.
Hence the process stops after finitely many exchanges.
At that point, \eqref{eq:balanced-color-classes} holds, since an exchange is
possible whenever it fails.
\end{proof}

\begin{lemma}
\label{lem:split-factorization}
Let $S$ and $B$ be disjoint sets with even sizes $r,b>0$ respectively, where $r\le b$. Let $H$ be a simple $(b-r)$-regular graph on $B$.  Then
\begin{equation}
 K_{S,B}\cup H
 \label{eq:split-factorization}
\end{equation}
has a one-factorization consisting of $b$ one-factors.
\end{lemma}

The proof in Appendix~\ref{app:square-kernel} first balances the colour
classes using Theorem~\ref{thm:standard-graph-tools}(v) and
Lemma~\ref{lem:balanced-color-classes}.
It then extends each class to a one-factor.

\subsection{Row decompositions and added columns}

An input UOM is a UOM used to construct a larger one.  Its existing columns
are called original columns.  The rows are divided into groups, which will be
treated as blocks when new columns are added.

\begin{definition}
\label{def:row-decomposition}
A row decomposition of an input UOM $X$ is a partition
\begin{equation}
 R(X)=A_1\disj\cdots\disj A_q,
 \label{eq:row-decomposition}
\end{equation}
into nonempty sets.  The sets $A_i$ are called decomposition blocks.
\qed
\end{definition}

\medskip
\noindent\emph{Block supports.}
For $I\subseteq[q]$, write
\begin{equation}
\Sfam_X(I)=\Sfam_X\!\left(\bigcup_{i\in I}A_i\right).
\end{equation}
The family on the right is the family of fibre supports in
Definition~\ref{def:equality-fibre}(ii).
Here \eqref{eq:row-decomposition} specifies the rows to be covered.

\medskip
\noindent\emph{Joint block capacity.}
Let $X^{(1)},\ldots,X^{(t)}$ have the same $N$ original columns and $q$
decomposition blocks.  Give the input UOMs disjoint alphabets.  Their
joint block capacity is
\begin{equation}
\begin{aligned}
 &\bcap(X^{(1)},\ldots,X^{(t)})\\
 &\quad=\max\Bigl\{\sum_{v=1}^t|I_v|:
 \begin{array}{l}
 I_v\subseteq[q],\ S_v\in\Sfam_{X^{(v)}}(I_v),\\[-1mm]
 S_u\cap S_v=\varnothing\quad(u\ne v)
 \end{array}\Bigr\}.
\end{aligned}
\label{eq:joint-capacity}
\end{equation}
Here $I_v$ lists the covered blocks.  The set $S_v$ is their fibre support
in Definition~\ref{def:equality-fibre}(ii).
The input alphabets are disjoint, so one original column can cover
rows from only one input.  This explains why the supports must be disjoint.
Thus \eqref{eq:joint-capacity} counts the largest total number of whole
blocks that the original columns can cover.
The data $(I_v,S_v)_{v=1}^t$ are called feasible when the conditions in the
maximum hold.

\Needspace{8\baselineskip}
\medskip
\noindent\emph{Demand graph.}
The added columns must make rows from different inputs orthogonal.
For $t$ inputs with $q$ blocks, these requirements form the demand graph
\begin{equation}
 D_{q,t}=K_{\underbrace{q,\ldots,q}_{t\text{ parts}}}.
 \label{eq:demand-graph}
\end{equation}
Each part represents one input UOM.  An edge therefore joins exactly two
blocks from different inputs.

\begin{definition}
\label{def:added-column-system}
An added column system is a family of new formal columns, each constant
on every decomposition block.  Every two blocks joined in $D_{q,t}$ receive
mate symbols in at least one added column.
\qed
\end{definition}

\medskip
\noindent\emph{Column weight.}
The weight $\mu_c$ of an added column $c$ is the largest number of whole
decomposition blocks contained in one fibre of $c$.

\medskip
\noindent\emph{Matching added columns.}
Given a one-factorization $(M_c)$ of $D_{q,t}$, its matching added column
system has one column $c$ for each $M_c$.  The endpoint blocks of each edge
of $M_c$ receive a new mate pair.  Different edges and columns use different
mate pairs.

Every fibre is one decomposition block.
Thus every matching added column has weight one.

\begin{definition}
\label{def:block-cover-condition}
Let $R(X)=A_1\disj\cdots\disj A_{N+1}$ be a row decomposition of an
$N$-column input UOM.  It satisfies the block cover condition if
\begin{equation}
 \kappa_X\!\left(\bigcup_{i\in I}A_i\right)\ge |I|
 \qquad(\varnothing\ne I\subseteq[N+1]).
 \label{eq:block-cover-condition}
\end{equation}
\qed
\end{definition}

Thus covering all rows of any $k$ decomposition blocks needs at
least $k$ original columns.

If all row decompositions satisfy the block cover condition, the joint
capacity is at most $N$.
Indeed, let $(I_v,S_v)_{v=1}^t$ be a feasible choice in
\eqref{eq:joint-capacity}.
Equation~\eqref{eq:block-cover-condition} gives $|I_v|\le|S_v|$ for each $v$.
The disjoint supports in \eqref{eq:joint-capacity} give
\begin{equation}
 \sum_v|I_v|\le\sum_v|S_v|
 =\left|\bigcup_vS_v\right|\le N.
 \label{eq:block-cover-condition-capacity}
\end{equation}

\medskip
\noindent\emph{An $8\times5$ input.}
Consider the following $8\times5$ input UOM.  To keep the display compact,
a number $s$ in column $j$ is a local label for a symbol $s_j$.
The column subscript is omitted in the display.
In each column, the mate pairs are $1\leftrightarrow2$,
$3\leftrightarrow4$, $5\leftrightarrow6$, and $7\leftrightarrow8$,
as in \eqref{eq:mate-rule}.
\begin{equation}
Y_5=\begin{pmatrix}
1&1&1&1&1\\
1&2&3&3&3\\
2&3&5&5&5\\
2&4&7&7&7\\
3&3&6&2&4\\
3&4&8&4&2\\
4&1&2&6&8\\
4&2&4&8&6
\end{pmatrix}.
\label{eq:Y5}
\end{equation}
\Needspace{7\baselineskip}
Number the rows by $1,\ldots,8$ from top to bottom.  Using the
row-decomposition convention in Definition~\ref{def:row-decomposition},
divide them as
\begin{equation}
\begin{aligned}
 A_1&=\{1\},& A_2&=\{2,7\},& A_3&=\{3\},\\
 A_4&=\{4,5\},& A_5&=\{6\},& A_6&=\{8\}.
\end{aligned}
\label{eq:Y5-decomposition-blocks}
\end{equation}

\begin{example}
For the blocks above, take $I=\{1,2\}$, so that
$A_1\cup A_2=\{1,2,7\}$.  No fibre contains all three rows.
The two fibres
\[
 F_{Y_5,1}(1)=\{1,2\},
 \qquad F_{Y_5,3}(2)=\{7\}
\]
cover them.  Hence
\[
 \kappa_{Y_5}(A_1\cup A_2)=2=|I|.
\]
This is one instance of Definition~\ref{def:block-cover-condition}.

Lemma~\ref{lem:Y5-profile} proves the block cover condition for every
nonempty choice of blocks.  The joint-capacity calculations used later are
given with the general construction in
Lemma~\ref{lem:certified-row-decompositions}.
\end{example}

\subsection{Completion graphs, matching cores, and square kernels}

The next construction starts from an orthogonal matrix $K$.
It adds mutually orthogonal completion rows until no further row can be
added.  The completion graph records these choices.

\begin{definition}
\label{def:completion-graph}
\begin{enumerate}[label=(\roman*)]
\item Let $K$ be an orthogonal matrix; it need not be a UOM.  A completion
row of $K$ is a row orthogonal to every row of $K$.  The set of all such rows
is
\begin{equation}
 \Comp(K)=\{r:r\perp s\text{ for every }s\in R(K)\}.
 \label{eq:completion-row-set}
\end{equation}
\item The completion graph $\Gamma(K)$ has vertex set $\Comp(K)$.  Two
vertices are joined exactly when their rows are orthogonal.
\end{enumerate}
\qed
\end{definition}

A completion graph can be infinite.  The square-kernel completion graphs
used here are finite.

\begin{lemma}
\label{lem:full-link}
\begin{enumerate}[label=(\roman*)]
\item
For a finite set $C\subseteq\Comp(K)$,
\begin{equation}
\begin{gathered}
 K\cup C\text{ is a UOM}\\[-1mm]
 \Longleftrightarrow\\[-1mm]
 C\text{ is an inclusion-maximal clique of }\Gamma(K).
\end{gathered}
 \label{eq:full-link}
\end{equation}

\item
If $K$ is $N\times N$ and every fibre has exactly one row, then
\begin{equation}
 \Comp(K)=\{r_\sigma:\sigma\in\Sym_N\},
 \qquad
 r_\sigma(j)=\overline{K_{\sigma(j),j}}.
 \label{eq:completion-permutation}
\end{equation}
In this case, each permutation gives one completion row.  Different
permutations give different rows.
\end{enumerate}
\end{lemma}

\begin{proof}
For (i), $K\cup C$ is orthogonal exactly when $C$ is a clique.
When this holds, an extension of the matrix is a vertex of $\Gamma(K)$
joined to every vertex of $C$, by
Definition~\ref{def:completion-graph}(i)--(ii).
Such a vertex exists exactly when $C$ is not maximal.
This proves \eqref{eq:full-link}.

For (ii), we first show that each column of a completion row must cover
a different core row.  Indeed, one column can cover at most one core row because its fibres
are singletons; see Definition~\ref{def:equality-fibre}(i).
Covering all $N$ rows with $N$ entries therefore covers
each row exactly once.  Let $\sigma(j)$ be the row covered in column $j$.
Then $\sigma\in\Sym_N$, with entries given by
\eqref{eq:completion-permutation}.

Conversely, that formula makes $r_\sigma$ orthogonal to core row
$\sigma(j)$ in column $j$.  A permutation lists every core row, so
$r_\sigma\in\Comp(K)$.
Singleton fibres also make the rows from distinct permutations distinct.
\end{proof}

\medskip
\noindent\emph{Matching cores.}
A matching core assigns every pair of rows to a column that makes them
orthogonal.  Index the rows by $[N]$.
Let $F_1,\ldots,F_N$ be matchings on $[N]$ which satisfy
\begin{equation}
 E(K_N)=F_1\disj\cdots\disj F_N.
 \label{eq:matching-core-partition}
\end{equation}
Construct an $N\times N$ matrix $K$ as follows.  Fix a column $j$.  For
every edge of $F_j$, choose a new mate pair.  Put its two symbols in the
rows at the endpoints of that edge.  Use different mate pairs for different
edges.  If a vertex is not covered by $F_j$, put a new symbol in its row;
do not use its mate in column $j$.  The resulting matrix $K$ is called the
matching core of $F_1,\ldots,F_N$.

\begin{theorem}
\label{thm:matching-core-full-link}
\begin{enumerate}[label=(\roman*)]
\item
The matching core $K$ is orthogonal, with one row in every fibre.
For $\sigma\in\Sym_N$, put
\begin{equation}
 r_\sigma(j)=\overline{K_{\sigma(j),j}}.
 \label{eq:matching-core-completion}
\end{equation}
Then $\Comp(K)=\{r_\sigma:\sigma\in\Sym_N\}$.

\item
For $\sigma,\tau\in\Sym_N$,
\begin{equation}
 r_\sigma\perp r_\tau
 \quad\Longleftrightarrow\quad
 \{\sigma(j),\tau(j)\}\in F_j
 \quad\text{for some }j.
 \label{eq:matching-core-adjacency}
\end{equation}

\item
For $C\subseteq\Sym_N$, the following two statements are equivalent:
\begin{equation}
\begin{gathered}
 K\cup\{r_\sigma:\sigma\in C\}\text{ is a UOM}\\[-1mm]
 \Longleftrightarrow\\[-1mm]
 C\text{ is an inclusion-maximal clique under }
 \eqref{eq:matching-core-adjacency}.
\end{gathered}
 \label{eq:matching-core-maximality}
\end{equation}
\end{enumerate}
\end{theorem}

\begin{proof}
For (i), \eqref{eq:matching-core-partition} assigns every pair of core
rows to a matching $F_j$.  Column $j$ makes that pair orthogonal.
Different rows in one column have distinct symbols, so its fibres are
singletons.  Lemma~\ref{lem:full-link}(ii) now gives
\eqref{eq:matching-core-completion}.

For (ii), substitute \eqref{eq:matching-core-completion} into
\eqref{eq:row-orthogonal}.  The entries in column $j$ are mates exactly
when $\{\sigma(j),\tau(j)\}\in F_j$.
This proves \eqref{eq:matching-core-adjacency}.

For (iii), Lemma~\ref{lem:full-link}(i) applies to this adjacency rule.
The result is \eqref{eq:matching-core-maximality}.
\end{proof}

Let $p$ be a positive odd integer.  Take a set $A$ of $p$ elements, numbered
$1,\ldots,p$.  For a single vertex, take the empty matching as the only
near-one-factor.  For larger odd orders,
Theorem~\ref{thm:standard-graph-tools}(ii) gives a
near-one-factorization
\begin{equation}
 E(K_A)=\disj_{a\in A}F_a,
 \qquad F_a\text{ misses }a,
 \label{eq:near-factorization}
\end{equation}
of $K_A$.

\begin{definition}
\label{def:square-kernel}
The matching core of the near-one-factorization
\eqref{eq:near-factorization} is denoted by $\mathcal K_p$ and is called a
square kernel.
\qed
\end{definition}

Theorem~\ref{thm:matching-core-full-link} shows that every fibre of
$\mathcal K_p$ contains one row.
Equations~\eqref{eq:near-factorization} and
\eqref{eq:matching-core-adjacency} describe its completion graph on
the vertex set $\Sym_p$.

Figure~\ref{fig:completion-example} shows the graph used in the next
example.  Its two vertices are orthogonal completion rows.
Figure~\ref{fig:full-link-picture} gives the general construction in
Lemma~\ref{lem:full-link}(i).  The drawings follow the style of
\cite{johnston2014structure}.

\begin{figure}[H]
\centering
\begin{tikzpicture}[x=1cm,y=1cm]
  \coordinate (v) at (-4.55,.22);
  \coordinate (w) at (-2.45,.22);

  \draw[line width=.45pt] (v)--(w);
  \fill (v) circle[radius=2pt];
  \fill (w) circle[radius=2pt];

  \node[anchor=south,inner sep=2pt] at (v) {$v=(2,4)$};
  \node[anchor=south,inner sep=2pt] at (w) {$w=(1,2)$};
  \node[anchor=north,inner sep=3pt] at (-3.50,.10)
    {$C=\{v,w\}=V\!\left(\Gamma(K)\right)$};

  \draw[-{Latex[length=2.0mm,width=1.2mm]},line width=.6pt]
    (-1.65,.22)--(-.70,.22);

  \node[anchor=west] at (-.42,.22) {$
    K\cup C=
    \left[
    \begin{array}{c|cc}
      \text{row}&1&2\\ \hline
      1&1&1\\
      2&2&3\\ \hline
      v&2&4\\
      w&1&2
    \end{array}
    \right]
  $};
\end{tikzpicture}
\caption{The completion graph for
$K=\left(\begin{smallmatrix}1&1\\2&3\end{smallmatrix}\right)$.
Its vertices are $v=(2,4)$ and $w=(1,2)$; the line means that these rows are
orthogonal.}
\label{fig:completion-example}
\end{figure}

\begin{example}
Figure~\ref{fig:completion-example} gives a $2\times2$ core.  The two rows
are orthogonal in the first column.  Each fibre contains one row.
The two permutations in $\Sym_2$ give, by
\eqref{eq:completion-permutation}, the completion rows
$v=(2,4)$ and $w=(1,2)$.

These are all completion rows by Lemma~\ref{lem:full-link}(ii).
They are orthogonal, so $\{v,w\}=V(\Gamma(K))$ is a maximal clique.
Part (i) of that lemma gives the UOM $K\cup\{v,w\}$.
\qed
\end{example}

\begin{figure}[!htbp]
  \centering
  \begingroup
  \setlength{\arraycolsep}{3.2pt}
  \renewcommand{\arraystretch}{1.16}

  {\small\textbf{The matrix when every fibre has one row}}
  \[
  Y_C=K\cup\{r_{\sigma_1},\ldots,r_{\sigma_s}\}
  =
  \left(
  \begin{array}{c|cccc}
   \text{row}\backslash\text{column}&1&2&\cdots&N\\ \hline
   1&k_{11}&k_{12}&\cdots&k_{1N}\\
   2&k_{21}&k_{22}&\cdots&k_{2N}\\
   \vdots&\vdots&\vdots&\ddots&\vdots\\
   N&k_{N1}&k_{N2}&\cdots&k_{NN}\\ \hline
   r_{\sigma_1}
    &\overline{k_{\sigma_1(1),1}}
    &\overline{k_{\sigma_1(2),2}}
    &\cdots&\overline{k_{\sigma_1(N),N}}\\
   \vdots&\vdots&\vdots&\ddots&\vdots\\
   r_{\sigma_s}
    &\overline{k_{\sigma_s(1),1}}
    &\overline{k_{\sigma_s(2),2}}
    &\cdots&\overline{k_{\sigma_s(N),N}}
  \end{array}
  \right).
  \tag*{(a)}
  \]
  \vspace{-1.0ex}
  \[
    \begin{aligned}
    r_\sigma(j)&=\overline{k_{\sigma(j),j}},\qquad
    C=\{r_{\sigma_1},\ldots,r_{\sigma_s}\}\subseteq \Comp(K),\\
    &\hspace{2.2cm}\sigma_1,\ldots,\sigma_s\in\Sym_N .
    \end{aligned}
    \tag*{(b)}
  \]
  Every lower entry is the mate of one original entry in
  the same column.  For an arbitrary core, write
  \(Y_C=\bigl(\begin{smallmatrix}K\\ C\end{smallmatrix}\bigr)\), where
  \(C\subseteq \Comp(K)\).  A lower matrix row and its vertex in the
  completion graph are denoted by the same label \(r_\sigma\).

  \vspace{1.0ex}
  \begin{tikzpicture}[x=1cm,y=1cm]
    \tikzset{
      vertex/.style={circle,fill=black,inner sep=0pt,minimum size=4pt},
      gedge/.style={draw=black,line width=.45pt}
    }

    \begin{scope}[xshift=-3.30cm]
      \node[font=\small] at (0,1.50) {(c) a completion row can be added};
      \coordinate (c1) at (-1.10,.68);
      \coordinate (c2) at (-1.10,-.68);
      \coordinate (c3) at (.10,0);
      \coordinate (u)  at (1.42,0);
      \draw[gedge] (c1)--(c2)--(c3)--cycle;
      \draw[gedge] (u)--(c1) (u)--(c2) (u)--(c3);
      \draw[dashed,line width=.45pt,rounded corners=7pt]
        (-1.42,-.98) rectangle (.43,.98);
      \node[vertex] at (c1) {};
      \node[vertex] at (c2) {};
      \node[vertex] at (c3) {};
      \node[vertex] at (u) {};
      \node[font=\scriptsize,anchor=east] at (-1.48,0) {$C$};
      \node[font=\scriptsize,anchor=west] at (1.50,0) {$r_\pi$};
      \node[font=\footnotesize,align=center] at (0,-1.40)
        {$r_\pi\perp r_\sigma$ for every $r_\sigma\in C$\\[2pt]
         $\Longrightarrow\
         \left[\begin{smallmatrix}Y_C\\ r_\pi\end{smallmatrix}\right]$
         is orthogonal};
    \end{scope}

    \draw[-{Latex[length=2.1mm,width=1.25mm]},line width=.65pt]
      (-.48,0)--(.48,0);
    \node[font=\scriptsize,align=center] at (0,.45)
      {no row outside $C$\\is adjacent to all of $C$};

    \begin{scope}[xshift=3.30cm]
      \node[font=\small] at (0,1.50) {(d) $Y_C$ is unextendible};
      \coordinate (d1) at (-.60,.68);
      \coordinate (d2) at (-.60,-.68);
      \coordinate (d3) at (.60,0);
      \draw[gedge] (d1)--(d2)--(d3)--cycle;
      \draw[dashed,line width=.45pt,rounded corners=7pt]
        (-.95,-.98) rectangle (.95,.98);
      \node[vertex] at (d1) {};
      \node[vertex] at (d2) {};
      \node[vertex] at (d3) {};
      \node[font=\scriptsize,anchor=east] at (-1.02,0) {$C$};
      \node[font=\footnotesize,align=center] at (0,-1.40)
        {$\nexists r_\pi\in \Comp(K)\setminus C:\
          r_\pi\perp r_\sigma\ (r_\sigma\in C)$};
    \end{scope}
  \end{tikzpicture}

  
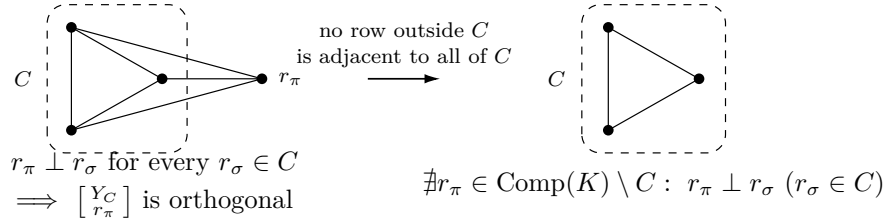
\vspace{.4ex}
  \[
    K\cup C\text{ is a UOM}
    \quad\Longleftrightarrow\quad
    C\text{ is an inclusion-maximal clique of }\Gamma(K).
    \tag*{(e)}
  \]
  \endgroup
  \caption{Construction from the completion graph.  In (a)--(b), \(k_{ij}\)
  denotes the original entry in row \(i\) and column \(j\).
  \(\overline{k_{ij}}\) denotes its mate in the same column.  Hence every new
  entry is the mate of the original entry given by the two subscripts.  In
  (c)--(d), every black point represents a completion row.  A line joins
  two orthogonal rows.  The set \(C\) is marked by the dashed box.}
  \label{fig:full-link-picture}
\end{figure}

\section{Two main construction theorems}
\label{sec:master-constructions}

The next lemma separates the fibres chosen in the original columns
from those chosen in the added columns.  It applies the test in
\eqref{eq:fibre-cover} to stacked inputs.

\medskip
\noindent\emph{Fibre choices.}
Let $X^{(1)},\ldots,X^{(t)}$ be input UOMs with the same $N$ original
columns.  In each original column, suppose that their alphabets are pairwise disjoint.  Put these
matrices one below another.  Add columns making rows from different
inputs orthogonal.  Call the resulting matrix $Y$.
Let $\sigma$ select at most one fibre from every added column.
Let $T_v(\sigma)$ be the rows of input $v$ left uncovered.

\Needspace{7\baselineskip}
\begin{lemma}
\label{lem:added-column-reduction}
The matrix $Y$ is extendible if and only if there
exist such a choice $\sigma$ and pairwise disjoint sets $S_v$ of original
columns satisfying
\begin{equation}
 S_v\in\Sfam_{X^{(v)}}(T_v(\sigma))\qquad(v\in[t]),
 \label{eq:added-column-reduction}
\end{equation}
by using \eqref{eq:support-family}.
\end{lemma}

\begin{proof}
First suppose that an extension exists.
Its added-column entries select $\sigma$.
In an original column, its entry can cover rows from at most one input,
because the input alphabets are disjoint.
Assign each used column to that input.
The assigned sets $S_v$ are disjoint.  Their fibres cover $T_v(\sigma)$
as required by \eqref{eq:support-family}.

Conversely, take $\sigma$ and disjoint supports satisfying
\eqref{eq:added-column-reduction}.
Choose fibres that give these supports.
Together with the fibres selected by $\sigma$, they cover all rows
using at most one fibre per column.  Equation~\eqref{eq:fibre-cover} gives an extension.
\end{proof}

\subsection{Construction from row decompositions}

The next theorem uses the row decompositions and added column systems in
Definitions~\ref{def:row-decomposition} and
\ref{def:added-column-system}.

\medskip
\noindent\emph{Stacked construction.}
Let $X^{(1)},\ldots,X^{(t)}$ be input UOMs with $N$ columns.  Decompose
the rows of every matrix into $q$ blocks.  Give the input UOMs pairwise
disjoint alphabets in every original column.  Put them one below another.
Attach an added column system containing $k$ columns.  The resulting matrix $Y$ has
\begin{equation}
 \sum_{v=1}^t|R(X^{(v)})|\text{ rows and }N+k\text{ columns}.
 \label{eq:row-decomposition-output}
\end{equation}

\Needspace{8\baselineskip}
\begin{theorem}
\label{thm:row-decomposition-construction}
For the matrix $Y$ above, the following statements hold.
\begin{enumerate}[label=(\roman*)]
\item If
\begin{equation}
 \sum_{c=1}^k\mu_c\le q(t-1),
 \qquad
 \bcap(X^{(1)},\ldots,X^{(t)})\le q-1,
 \label{eq:row-decomposition-hypothesis}
\end{equation}
then $Y$ is a UOM.

\item If the added column system is a matching added column system, then
$k=q(t-1)$ and
\begin{equation}
 Y\text{ is extendible}
 \quad\Longleftrightarrow\quad
 \bcap(X^{(1)},\ldots,X^{(t)})\ge q.
 \label{eq:row-decomposition-exactness}
\end{equation}

\item If $q$ is even and $t\ge2$, then there exists a matching added
column system.
\end{enumerate}
\end{theorem}

\begin{proof}
For (i) and (ii), compare the number of blocks left by the added columns
with the joint capacity of the original columns.
A one-factorization gives the added columns in (iii).

The matrix $Y$ is orthogonal.
Rows from one input remain orthogonal in the original columns.
Blocks from different inputs are joined in \eqref{eq:demand-graph}.
Definition~\ref{def:added-column-system} makes their rows orthogonal
in an added column.

For (i), suppose that an extension exists.
We claim that at least $q$ blocks must be covered by the original columns.
Indeed, a fibre in added column $c$ covers at most $\mu_c$ whole blocks.
It covers each block completely or not at all, by
Definition~\ref{def:added-column-system}.
Let $I_v$ index the blocks of input $v$ left for the original columns.
The first bound in \eqref{eq:row-decomposition-hypothesis} gives
\begin{equation}
 \sum_v|I_v|\ge qt-\sum_c\mu_c\ge q.
 \label{eq:decomposition-block-count}
\end{equation}
Lemma~\ref{lem:added-column-reduction} gives disjoint supports for these
blocks.  Equations~\eqref{eq:joint-capacity} and
\eqref{eq:decomposition-block-count} therefore give
\begin{equation}
 \bcap(X^{(1)},\ldots,X^{(t)})
 \ge\sum_v|I_v|\ge q.
 \label{eq:block-capacity-contradiction}
\end{equation}
This contradicts \eqref{eq:row-decomposition-hypothesis}, proving (i).

For (ii), a matching added column system has $k=q(t-1)$ with every $\mu_c=1$.
The same count gives the forward implication, without using
the bound $\bcap\le q-1$.

For the reverse implication, we show that a cover of $q$ blocks in the
original columns can be completed to a cover of all rows.
Suppose that
$\bcap(X^{(1)},\ldots,X^{(t)})\ge q$.  Choose feasible sets
$(I_v,S_v)_{v=1}^t$ as in \eqref{eq:joint-capacity}, with
$\sum_v|I_v|\ge q$.  Remove indices from the sets $I_v$ until their total
size is $q$.  The same sets $S_v$ still support them.  Let the corresponding
original fibres cover these blocks.  The number of blocks outside the
selected sets $I_v$ is
\begin{equation}
 qt-q=q(t-1)=k.
 \label{eq:remaining-block-count}
\end{equation}
Each block is a fibre in every matching added column, since different
factor edges use different mate pairs.
Assign the $k$ blocks counted in \eqref{eq:remaining-block-count} to
distinct added columns.  These fibres complete the original cover.
Equation~\eqref{eq:fibre-cover} proves the reverse implication in
\eqref{eq:row-decomposition-exactness}.

For (iii), the graph in \eqref{eq:demand-graph} has $t$ parts of even
size $q$.  Theorem~\ref{thm:standard-graph-tools}(iv) gives the
one-factorization required for matching added columns.
\end{proof}

The cases with two, four, and six decomposition blocks are given in
Appendix~\ref{app:row-decomposition-constructions}.
Figure~\ref{fig:two-added-columns} shows the smallest construction using
two input UOMs.  Figure~\ref{twon:fig-four-group-lift} gives the block
matrix for four decomposition blocks.

\begin{figure}[!htbp]
  \centering
  \begingroup
  \setlength{\arraycolsep}{4.0pt}
  \renewcommand{\arraystretch}{1.22}

  {\small\textbf{The matrix for \(q=t=2\)}}\par
  \smallskip
  Let the two input UOMs be \(\widetilde X^{(1)}\) and
  \(\widetilde X^{(2)}\).  Their decomposition blocks are denoted by
  \(A_0,A_1\) and \(B_0,B_1\), respectively.  The two added
  columns are denoted by \(s_0,s_1\).
  \[
  Y=
  \left[
  \begin{array}{c|c|cc}
   \text{decomposition block}&\text{original columns }[N]&s_0&s_1\\ \hline
   A_0&\widetilde X^{(1)}|_{A_0}
      &\alpha_0\mathbf 1_{|A_0|}&\beta_0\mathbf 1_{|A_0|}\\
   A_1&\widetilde X^{(1)}|_{A_1}
      &\alpha_1\mathbf 1_{|A_1|}&\beta_1\mathbf 1_{|A_1|}\\ \hline
   B_0&\widetilde X^{(2)}|_{B_0}
      &\overline{\alpha_0}\mathbf 1_{|B_0|}
      &\overline{\beta_1}\mathbf 1_{|B_0|}\\
   B_1&\widetilde X^{(2)}|_{B_1}
      &\overline{\alpha_1}\mathbf 1_{|B_1|}
      &\overline{\beta_0}\mathbf 1_{|B_1|}
  \end{array}
  \right].
  \tag*{(a)}
  \]
  Here every input UOM \(\widetilde X^{(v)}\) uses a separate alphabet,
  and no entry in its original columns is changed.  The notation
  \(\gamma\mathbf1_{|A|}\) means that the added column has the same value
  \(\gamma\) on all rows of the decomposition block \(A\).

  
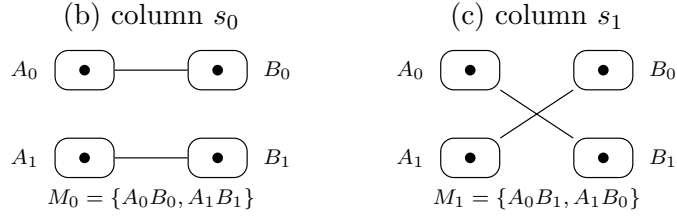
\vspace{1.2ex}
  \begin{tikzpicture}[x=1cm,y=1cm]
    \tikzset{
      block/.style={draw=black,line width=.45pt,
        rounded corners=5pt,minimum width=.78cm,minimum height=.52cm,
        inner sep=0pt},
      factor/.style={draw=black,line width=.45pt}
    }
    \newcommand{\drawblocks}{%
      \node[block] (Azero) at (-.88,.58) {};
      \node[block] (Aone)  at (-.88,-.58) {};
      \node[block] (Bzero) at (.88,.58) {};
      \node[block] (Bone)  at (.88,-.58) {};
      \fill (Azero.center) circle[radius=2pt];
      \fill (Aone.center) circle[radius=2pt];
      \fill (Bzero.center) circle[radius=2pt];
      \fill (Bone.center) circle[radius=2pt];
      \node[font=\scriptsize,anchor=east] at (-1.34,.58) {$A_0$};
      \node[font=\scriptsize,anchor=east] at (-1.34,-.58) {$A_1$};
      \node[font=\scriptsize,anchor=west] at (1.34,.58) {$B_0$};
      \node[font=\scriptsize,anchor=west] at (1.34,-.58) {$B_1$};
    }

    \begin{scope}[xshift=-2.55cm]
      \node[font=\small] at (0,1.28) {(b) column \(s_0\)};
      \drawblocks
      \draw[factor] (Azero)--(Bzero) (Aone)--(Bone);
      \node[font=\scriptsize] at (0,-1.12)
        {$M_0=\{A_0B_0,A_1B_1\}$};
    \end{scope}
    \begin{scope}[xshift=2.55cm]
      \node[font=\small] at (0,1.28) {(c) column \(s_1\)};
      \drawblocks
      \draw[factor] (Azero)--(Bone) (Aone)--(Bzero);
      \node[font=\scriptsize] at (0,-1.12)
        {$M_1=\{A_0B_1,A_1B_0\}$};
    \end{scope}
  \end{tikzpicture}

  \vspace{.7ex}
  \[
   \begin{aligned}
    s_0:&\quad \alpha_i\leftrightarrow\overline{\alpha_i},
       &&\{A_i,B_i\}\in M_0,\\
    s_1:&\quad \beta_i\leftrightarrow\overline{\beta_i},
       &&\{A_i,B_{1-i}\}\in M_1,
       \qquad i\in\{0,1\}.
   \end{aligned}
   \tag*{(d)}
  \]
  The four mate pairs in (d) are all new, and every pair is used on one
  edge in one column.  Therefore, any two decomposition blocks which come
  from different input UOMs become orthogonal in exactly one added column. For two blocks which come from the same input UOM, their orthogonality is
  still given by \(\widetilde X^{(v)}\).
  \endgroup
  \caption{The construction with two added columns.  The two input UOMs
  use different alphabets.  In each added column, the entries are constant
  on every decomposition block.  The bars indicate the mate relation between the
  entries.  The two factor graphs used for the added columns are shown in
  panels (b) and (c).}
  \label{fig:two-added-columns}
\end{figure}

\subsection{Square kernel transfer and compressed construction}

The next construction enlarges the square kernels in
Definition~\ref{def:square-kernel}.  It uses the split factorization in
Lemma~\ref{lem:split-factorization}, whose proof is in
Appendix~\ref{app:square-kernel}.
The resulting UOMs satisfy the block cover condition in
Definition~\ref{def:block-cover-condition}.
This condition allows two such inputs to be combined with fewer added
columns.

For an odd integer $p$, put
\begin{equation}
 Q_p=\{p\}\cup\{q:q\ge2p+1,\ q\text{ odd}\}.
 \label{eq:Qp}
\end{equation}

\medskip
\noindent\emph{Square-kernel data.}
Let $\mathcal K_p$ be an odd square kernel.  Take its completion graph
from Definition~\ref{def:completion-graph}(ii).  By
\eqref{eq:matching-core-completion}, its vertices are the rows $r_\sigma$
with $\sigma\in\Sym_p$.  Write $\sigma$ for the vertex $r_\sigma$.  Let
$C,D\subseteq\Sym_p$ be two nonempty inclusion-maximal cliques of this graph,
with inclusion-maximal understood as in
Definition~\ref{def:graph-notation}(ii).  Write $|C|=c$ and $|D|=d$.  Define
\begin{equation}
 C^\sharp=\{r_\sigma:\sigma\in C\},
 \qquad
 D^\sharp=\{r_\sigma:\sigma\in D\}.
 \label{eq:square-clique-rows}
\end{equation}

\Needspace{13\baselineskip}
\begin{theorem}
\label{thm:square-transfer}
With the data above, the following statements hold.
\begin{enumerate}[label=(\roman*)]
\item For every $q\in Q_p$, there exists a $(q+c)\times q$ UOM which has a
row decomposition into $q+1$ blocks satisfying the block cover condition.

\item For every $q\in Q_p$ and every integer $s$ satisfying
$0\le s\le(q+1)/2$, there exists a
\begin{equation}
 (2q+c+d)\times(2q+1-s)\text{ UOM}.
 \label{eq:compressed-shape}
\end{equation}
\end{enumerate}
\end{theorem}

\begin{proof}
Part (i) starts from the UOM given by the original kernel and maximal clique.
For a larger kernel, the main step is to prove that the extended clique is
still maximal.  The same count of covered core rows gives the block cover
condition in both cases.
For (ii), pairs of matching added columns are replaced by single columns.
Their total weight is unchanged, so
Theorem~\ref{thm:row-decomposition-construction}(i) still applies.

First take $q=p$ in (i).
Equation~\eqref{eq:full-link} makes
$X(C)=\mathcal K_p\cup C^\sharp$ a UOM.
Its blocks are the $p$ singleton core rows and the completion block
$C^\sharp$ from \eqref{eq:square-clique-rows}.
It remains to prove \eqref{eq:block-cover-condition}.
Covering $u$ core blocks needs at least $u$ columns, since a fibre meets
the core in at most one row.  This proves the block cover condition for
sets of blocks which omit $C^\sharp$.

A set consisting of $C^\sharp$ and $u$ core blocks needs at least
$u+1$ columns.  Indeed, suppose that at most $u$ columns cover their union.
At least $p-u$ columns remain unused.  Assign the $p-u$ remaining core
rows to different unused columns.  In each assigned column, select the
fibre containing that row.  These fibres cover $X(C)$ using at most $p$
columns.  This contradicts Theorem~\ref{thm:fibre-cover}(i), since $X(C)$
is a UOM.  Thus these $u+1$ blocks need at least $u+1$ columns.
Both types of block sets satisfy \eqref{eq:block-cover-condition}.

For odd $q\ge2p+1$, first enlarge the kernel by extending its
near-one-factorization.  Take $A=[p]$ as the vertex
set of the original kernel.  Write
\begin{equation}
 V=A\disj B,\qquad |A|=p,\qquad |B|=b=q-p.
 \label{eq:square-split}
\end{equation}
Here $b$ is even and $b\ge p+1$.  The graph $K_B$ has a factorization into
$b-1$ one-factors by Theorem~\ref{thm:standard-graph-tools}(i).  Since
$b-1\ge p$, choose distinct factors $M_a$.  Put
\begin{equation}
 \widehat F_a=F_a\disj M_a\quad(a\in A).
 \label{eq:old-lifted-factors}
\end{equation}
Add a new vertex $\omega$.  Put
\begin{equation}
 S=A\cup\{\omega\},\qquad
 H=K_B-\disj_{a\in A}M_a.
 \label{eq:square-residual}
\end{equation}
The graph $H$ is $(b-p-1)$-regular.  Lemma~\ref{lem:split-factorization}
factors $K_{S,B}\cup H$ into $b$ one-factors.  Each factor has exactly one
edge incident with $\omega$.  Write it as $\omega x$.  Removing the
auxiliary vertex gives the near-one-factor
$\widehat F_x$, which misses $x$.
Every edge at the removed vertex occurs in exactly one factor.
Thus the new factors miss distinct vertices.
Together with the factors in \eqref{eq:old-lifted-factors}, the new factors
form a near-one-factorization of $K_V$.  Definition~\ref{def:square-kernel}
now gives a square kernel $\mathcal K_q$.

The square kernel $\mathcal K_q$ is orthogonal by
Theorem~\ref{thm:matching-core-full-link}(i).
Its entries are as follows.  If $u,a\in A$, put
\[
 k_{u,a}=(\mathcal K_p)_{u,a}.
\]
For every $\{x,y\}\in M_a$, choose a mate pair not used elsewhere in
column $a$.  Write
\[
 \xi_{x,a}=\overline{\xi_{y,a}}.
\]
For $z\in B$ and $\{v,w\}\in\widehat F_z$, choose a mate pair not used
elsewhere in column $z$.  Write
\[
 \eta_{v,z}=\overline{\eta_{w,z}}.
\]
The factor $\widehat F_z$ misses $z$.  Put a new symbol
$\eta_{z,z}$ in row $z$ and column $z$.  Its mate is not used in column
$z$.

Extend each $\sigma\in C$ to $V$ by putting
\begin{equation}
 \widehat\sigma(a)=\sigma(a)\ (a\in A),\qquad
 \widehat\sigma(x)=x\ (x\in B).
 \label{eq:lifted-permutation}
\end{equation}
The extended matrix has the block form
\[
 \widehat X(C)=
 \left[
 \begin{array}{c|cc}
  \text{row}\backslash\text{column}&a\in A&z\in B\\ \hline
  u\in A&k_{u,a}&\eta_{u,z}\\
  x\in B&\xi_{x,a}&\eta_{x,z}\\ \hline
  r_{\widehat\sigma},\ \sigma\in C
    &\overline{k_{\sigma(a),a}}&\overline{\eta_{z,z}}
 \end{array}
 \right].
\]
The first two block rows form $\mathcal K_q$.  The last block row contains
the $|C|$ extended completion rows.

The extended completion rows are pairwise orthogonal in the original columns.
They form a clique by Definition~\ref{def:completion-graph}(ii).
We now prove that this clique is maximal in the larger completion graph.
Suppose that an outside vertex $\pi\in\Sym(V)$ is
adjacent to every $\widehat\sigma$.  Fix
$\sigma\in C$.  A new column $x\in B$ cannot make $\pi$ adjacent to
$\widehat\sigma$: $\widehat\sigma(x)=x$, but $\widehat F_x$ misses $x$.
The adjacency must come from an original column $a\in A$.

We claim that this outside vertex gives a permutation of $A$ adjacent
to every member of $C$.  Indeed, for each $\sigma\in C$, choose an
original column $a\in A$ that makes $\pi$ adjacent to $\widehat\sigma$.
In this column, \eqref{eq:lifted-permutation} gives
$\widehat\sigma(a)=\sigma(a)\in A$.
Equation~\eqref{eq:matching-core-adjacency} implies
$\{\pi(a),\widehat\sigma(a)\}\in\widehat F_a$.  The factor
$\widehat F_a$ has no $A$--$B$ edge.  Hence $\pi(a)\in A$.  Put
\[
 J=\{a\in A:\pi(a)\in A\}.
\]
The restriction $\pi|_J$ is a one-to-one map from $J$ into $A$.  Extend it
to a permutation $\widetilde\pi\in\Sym(A)$.  For every $\sigma\in C$, some
original column in $J$ makes $\pi$ adjacent to $\widehat\sigma$.  That
column also makes $\widetilde\pi$ adjacent to $\sigma$.
Thus $\widetilde\pi$ is adjacent to every vertex of $C$.
It cannot itself belong to $C$, since the graph has no loops.
This contradicts maximality in Definition~\ref{def:graph-notation}(ii).

The maximality just proved gives a $(q+c)\times q$ UOM by
\eqref{eq:full-link}.  Its blocks again consist of singleton core rows
and one completion block.
The argument for $q=p$ proves \eqref{eq:block-cover-condition}, completing (i).

For (ii), the two inputs from (i) have joint capacity at most $q$ by
\eqref{eq:block-cover-condition-capacity}.
To apply \eqref{eq:row-decomposition-hypothesis}, it remains to construct
the added columns with total weight $q+1$.
Take the two extended input UOMs constructed from $C^\sharp$ and
$D^\sharp$.  Relabel their mate pairs so that the two alphabets are
disjoint in every original column.  This preserves orthogonality and
all fibre supports.  Each input has $A_0=q+1$ decomposition blocks.
Write $L_i$ and $R_i$, with $i\in\mathbb Z_{A_0}$, for the blocks of the
first and second inputs, respectively.  The cyclic factors
\begin{equation}
 M_e=\{L_iR_{i+e}:i\in\mathbb Z_{A_0}\}
 \qquad(e\in\mathbb Z_{A_0})
 \label{eq:cyclic-cross-factors}
\end{equation}
give a factorization of $K_{A_0,A_0}$.  Two opposite factors satisfy
\begin{equation}
 \begin{aligned}
 M_e\cup M_{e+A_0/2}
 &=\disj_{i\in\mathbb Z_{A_0/2}}\!
 K_{2,2}\bigl[\{L_i,L_{i+A_0/2}\},\\[-1mm]
 &\hspace{38mm}\{R_{i+e},R_{i+e+A_0/2}\}\bigr].
 \end{aligned}
 \label{eq:opposite-factors}
\end{equation}
Equation~\eqref{eq:opposite-factors} allows each opposite pair to be
put in one column of weight two.  These opposite pairs are disjoint.
Combine $s$ such pairs into $s$
compressed columns.  For every
$K_{2,2}$, choose a mate pair not used elsewhere in that column.
Put one symbol on the two blocks in one part.
Put its mate on the two blocks in the other part.  Keep
every other factor in a separate column.  Any two blocks from different input UOMs now have mate symbols in
exactly one added column.  There are $A_0-s$ added columns.  Their total weight is
\begin{equation}
 (A_0-2s)+2s=A_0=q+1.
 \label{eq:compressed-capacity}
\end{equation}

Equations~\eqref{eq:compressed-capacity} and
\eqref{eq:block-cover-condition-capacity} give
\eqref{eq:row-decomposition-hypothesis} with $q+1$ blocks per input.
Theorem~\ref{thm:row-decomposition-construction}(i) gives a UOM.
Its dimensions are \eqref{eq:compressed-shape}.
\end{proof}

\phantomsection\label{par:general-compression}
The compression in Theorem~\ref{thm:square-transfer}(ii) applies to
any two UOMs with $q$ columns, disjoint alphabets, and $q+1$ blocks
satisfying \eqref{eq:block-cover-condition}.
The inputs need not come from square kernels.  For any such inputs with
$m_1,m_2$ rows and odd $q$, the same construction gives
\[
 (m_1+m_2)\times(2q+1-s)\text{ UOMs},
 \qquad 0\le s\le(q+1)/2.
\]
The two bounds used above are still
\eqref{eq:block-cover-condition-capacity} and
\eqref{eq:compressed-capacity}, so
Theorem~\ref{thm:row-decomposition-construction}(i) applies unchanged.

\Needspace{4\baselineskip}
The two maps used in the square kernel transfer are shown in
Figure~\ref{fig:group-perfect-square-core-lift}.  The matrix of the
compressed construction is given in
Figure~\ref{fig:group-perfect-two-parent-join}.

\begin{figure}[!htbp]
  \centering
  \begingroup
  \small
  \setlength{\arraycolsep}{4.2pt}
  \renewcommand{\arraystretch}{1.16}

  {\small\textbf{The lift of a square kernel}}\par
  \smallskip

  \begin{tikzpicture}[x=.88cm,y=.88cm]
    \tikzset{
      point/.style={circle,fill=black,inner sep=0pt,minimum size=4pt},
      edge/.style={draw=black,line width=.45pt},
      lift/.style={-{Latex[length=2.2mm,width=1.3mm]},line width=.68pt},
      lab/.style={font=\scriptsize,inner sep=1pt}
    }

    \node[font=\small] at (-1.35,3.22)
      {(a) Add a one-factor to an original factor};
    \node[font=\footnotesize] at (-1.35,2.83)
      {$F_1=\{\{2,3\}\}$,\quad
       $M_1=\{\{x,y\},\{z,w\}\}$,\quad
       $\widehat F_1=F_1\disj M_1$};

    \foreach \xx/\name/\lab in {-6.10/Aone/1,-5.25/Atwo/2,-4.40/Athree/3}{
      \node[point] (\name) at (\xx,2.05) {};
      \node[lab,anchor=south] at (\xx,2.17) {$\lab$};
    }
    \draw[edge] (Atwo)--(Athree);
    \node[lab] at (-5.25,1.53) {$A$};
    \node[lab,anchor=north] at (-5.25,1.35) {$F_1$ misses vertex $1$};

    \draw[lift] (-3.85,2.05)--(-2.85,2.05);

    \foreach \xx/\name/\lab in {-2.10/Bone/1,-1.25/Btwo/2,-.40/Bthree/3}{
      \node[point] (\name) at (\xx,2.05) {};
      \node[lab,anchor=south] at (\xx,2.17) {$\lab$};
    }
    \foreach \xx/\name/\lab in {.75/Bx/x,1.60/By/y,2.45/Bz/z,3.30/Bw/w}{
      \node[point] (\name) at (\xx,2.05) {};
      \node[lab,anchor=south] at (\xx,2.17) {$\lab$};
    }
    \draw[edge] (Btwo)--(Bthree) (Bx)--(By) (Bz)--(Bw);
    \draw[dashed,line width=.38pt] (.18,1.55)--(.18,2.55);
    \node[lab] at (-1.25,1.55) {$A$};
    \node[lab] at (2.03,1.55) {$B$};

    \node[font=\small] at (-1.35,.83)
      {(b) Extend an original permutation};
    \node[font=\footnotesize] at (-1.35,.43)
      {$\sigma:1\mapsto2,\ 2\mapsto3,\ 3\mapsto1,
       \qquad \widehat\sigma=\sigma\sqcup\operatorname{id}_B$};

    \foreach \xx/\name/\lab in {-6.10/Ctone/1,-5.25/Cttwo/2,-4.40/Ctthree/3}{
      \node[point] (\name) at (\xx,-.25) {};
      \node[lab,anchor=south] at (\xx,-.13) {$\lab$};
    }
    \foreach \xx/\name/\lab in {-6.10/Cbone/1,-5.25/Cbtwo/2,-4.40/Cbthree/3}{
      \node[point] (\name) at (\xx,-1.45) {};
      \node[lab,anchor=north] at (\xx,-1.57) {$\lab$};
    }
    \draw[edge] (Ctone)--(Cbtwo) (Cttwo)--(Cbthree)
      (Ctthree)--(Cbone);
    \node[lab] at (-5.25,-1.93) {$A$};

    \draw[lift] (-3.85,-.85)--(-2.85,-.85);

    \foreach \xx/\name/\lab in {-2.10/Done/1,-1.25/Dtwo/2,-.40/Dthree/3}{
      \node[point] (Dt\name) at (\xx,-.25) {};
      \node[point] (Db\name) at (\xx,-1.45) {};
      \node[lab,anchor=south] at (\xx,-.13) {$\lab$};
      \node[lab,anchor=north] at (\xx,-1.57) {$\lab$};
    }
    \foreach \xx/\name/\lab in {.75/Dx/x,1.60/Dy/y,2.45/Dz/z,3.30/Dw/w}{
      \node[point] (Dt\name) at (\xx,-.25) {};
      \node[point] (Db\name) at (\xx,-1.45) {};
      \node[lab,anchor=south] at (\xx,-.13) {$\lab$};
      \node[lab,anchor=north] at (\xx,-1.57) {$\lab$};
      \draw[edge] (Dt\name)--(Db\name);
    }
    \draw[edge] (DtDone)--(DbDtwo) (DtDtwo)--(DbDthree)
      (DtDthree)--(DbDone);
    \draw[dashed,line width=.38pt] (.18,-1.87)--(.18,.17);
    \node[lab] at (-1.25,-1.93) {$A$};
    \node[lab] at (2.03,-1.93) {$B$};
  \end{tikzpicture}

  \vspace{1.0ex}
  {\small\textbf{The resulting matrix}}
  \[
  \widehat X(C)=
  \left[
  \begin{array}{c|cc}
   \text{row}\backslash\text{column}&a\in A&z\in B\\ \hline
   u\in A&k_{u,a}&\eta_{u,z}\\
   x\in B&\xi_{x,a}&\eta_{x,z}\\ \hline
   r_{\widehat\sigma},\ \sigma\in C
     &\overline{k_{\sigma(a),a}}&\overline{\eta_{z,z}}
  \end{array}
  \right].
  \tag*{(c)}
  \]

  \vspace{.2ex}
  \[
  \begin{array}{c@{\quad\Longrightarrow\quad}l}
    F_a & k_{u,a}\ (u\in A),\\
    M_a & \xi_{x,a}\ (x\in B),\\
    \widehat F_z & \eta_{v,z}\ (v\in A\disj B),\\
    \widehat\sigma &
      \overline{k_{\sigma(a),a}}\text{ and }\overline{\eta_{z,z}}.
  \end{array}
  \tag*{(d)}
  \]
  If \(\{v,w\}\) is an edge of \(\widehat F_z\), put
  \(\eta_{v,z}=\overline{\eta_{w,z}}\).  Since the factor
  \(\widehat F_z\) misses \(z\), the symbol \(\eta_{z,z}\) is the only
  unmatched symbol in column \(z\).
  \endgroup
  \caption{The lift of a square kernel.  In panel (a), the one-factor
  \(M_1\) is added to the original factor \(F_1\).  In panel (b), the
  permutation on \(A\) is kept, while every element of \(B\) is fixed.
  The full matrix obtained from the construction is given in panel (c),
  and panel (d) explains where each family of its entries comes from.}
  \label{fig:group-perfect-square-core-lift}
\end{figure}
\begin{figure}[!htbp]
  \centering
  \begingroup
  \small
  \setlength{\arraycolsep}{3.5pt}
  \renewcommand{\arraystretch}{1.18}

  {\small\textbf{The matrix after two factor columns are merged}}\par
  \smallskip
  Let \(\widetilde X_L\) and \(\widetilde X_R\) be the two input UOMs,
  and let their original alphabets be different.  Their decomposition
  blocks are denoted by \(L_i\) and \(R_i\), where
  \(i\in\mathbb Z_{A_0}\), \(A_0=q+1\), and \(h=A_0/2\).  For
  \(P=\{e_P,e_P+h\}\in\mathcal P\), the two factor columns marked by
  this set are replaced by one merged column.  The factors which are not
  merged form \(\mathcal U\), and \([i]\) denotes the class of \(i\) in
  \(\mathbb Z_h\).  In the matrix below, the four column blocks give the
  row labels, the original columns, the separate factor columns, and the
  merged columns, respectively.
  \[
  Y_s=
  \left[
  \begin{array}{c|c|c|c}
   \text{decomposition block}
     &[q]&(s_e)_{e\in\mathcal U} &(z_P)_{P\in\mathcal P}\\ \hline
   L_i
     &\widetilde X_L|_{L_i}
     &\bigl(\alpha_{e,i}\mathbf1_{|L_i|}\bigr)_{e\in\mathcal U}
     &\bigl(\gamma_{P,[i]}\mathbf1_{|L_i|}\bigr)_{P\in\mathcal P}\\
   R_j
     &\widetilde X_R|_{R_j}
     &\bigl(\overline{\alpha_{e,j-e}}\mathbf1_{|R_j|}\bigr)_{e\in\mathcal U}
     &\bigl(\overline{\gamma_{P,[j-e_P]}}\mathbf1_{|R_j|}\bigr)_{P\in\mathcal P}
  \end{array}
  \right].
  \tag*{(a)}
  \]

  {\small\textbf{One opposite pair before and after merging}}
  \[
  \left[
  \begin{array}{c|cc}
   &s_e&s_{e+h}\\ \hline
   L_i&\alpha&\delta\\
   L_{i+h}&\beta&\varepsilon\\
   R_{i+e}&\overline\alpha&\overline\varepsilon\\
   R_{i+e+h}&\overline\beta&\overline\delta
  \end{array}
  \right]
  \quad\longmapsto\quad
  \left[
  \begin{array}{c|c}
   &z_{\{e,e+h\}}\\ \hline
   L_i&\gamma_{[i]}\\
   L_{i+h}&\gamma_{[i]}\\
   R_{i+e}&\overline\gamma_{[i]}\\
   R_{i+e+h}&\overline\gamma_{[i]}
  \end{array}
  \right].
  \tag*{(b)}
  \]
  In the matrix on the left, the four matching edges use four different
  mate pairs.  After merging, as shown in the matrix on the right, one new
  mate pair is put on the two sets \(L_i\cup L_{i+h}\) and
  \(R_{i+e}\cup R_{i+e+h}\).

  \vspace{1.0ex}
  \begin{tikzpicture}[x=1cm,y=1cm]
    \tikzset{
      point/.style={circle,fill=black,inner sep=0pt,minimum size=4pt},
      block/.style={draw=black,line width=.45pt,rounded corners=5pt},
      edge/.style={draw=black,line width=.45pt}
    }
    \begin{scope}[xshift=-4.05cm]
      \node[font=\small] at (0,1.18) {(c) \(M_e\)};
      \node[point] (L0) at (-.72,.45) {};
      \node[point] (L1) at (-.72,-.45) {};
      \node[point] (R0) at (.72,.45) {};
      \node[point] (R1) at (.72,-.45) {};
      \draw[edge] (L0)--(R0) (L1)--(R1);
      \node[font=\scriptsize,anchor=east] at (-.84,.45) {$L_i$};
      \node[font=\scriptsize,anchor=east] at (-.84,-.45) {$L_{i+h}$};
      \node[font=\scriptsize,anchor=west] at (.84,.45) {$R_{i+e}$};
      \node[font=\scriptsize,anchor=west] at (.84,-.45) {$R_{i+e+h}$};
    \end{scope}
    \node[font=\small] at (-2.05,0) {$+$};
    \begin{scope}[xshift=-.05cm]
      \node[font=\small] at (0,1.18) {(d) \(M_{e+h}\)};
      \node[point] (LL0) at (-.72,.45) {};
      \node[point] (LL1) at (-.72,-.45) {};
      \node[point] (RR0) at (.72,.45) {};
      \node[point] (RR1) at (.72,-.45) {};
      \draw[edge] (LL0)--(RR1) (LL1)--(RR0);
      \node[font=\scriptsize,anchor=east] at (-.84,.45) {$L_i$};
      \node[font=\scriptsize,anchor=east] at (-.84,-.45) {$L_{i+h}$};
      \node[font=\scriptsize,anchor=west] at (.84,.45) {$R_{i+e}$};
      \node[font=\scriptsize,anchor=west] at (.84,-.45) {$R_{i+e+h}$};
    \end{scope}
    \draw[-{Latex[length=2.1mm,width=1.25mm]},line width=.65pt]
      (1.55,0)--(2.28,0);
    \begin{scope}[xshift=3.75cm]
      \node[font=\small] at (0,1.18) {(e) the merged column};
      \draw[block] (-1.08,-.76) rectangle (-.34,.76);
      \draw[block] (.34,-.76) rectangle (1.08,.76);
      \node[point] (ML0) at (-.71,.38) {};
      \node[point] (ML1) at (-.71,-.38) {};
      \node[point] (MR0) at (.71,.38) {};
      \node[point] (MR1) at (.71,-.38) {};
      \draw[edge] (ML0)--(MR0) (ML0)--(MR1)
        (ML1)--(MR0) (ML1)--(MR1);
      \node[font=\scriptsize,align=center] at (-.71,-1.08)
        {$L_i$\\$L_{i+h}$};
      \node[font=\scriptsize,align=center] at (.71,-1.08)
        {$R_{i+e}$\\$R_{i+e+h}$};
      \node[font=\scriptsize] at (0,-1.57) {$K_{2,2}$};
    \end{scope}
  \end{tikzpicture}
  \endgroup
  \caption{The joining of two lifted input UOMs by merging factor columns.
  All entries of the resulting matrix are given in panel (a).  Panel (b)
  compares a pair of factor columns with the merged column.  The same operation on the decomposition blocks is shown
  in panels (c)--(e).  The four edges in panel (e) form the graph
  \(K_{2,2}\).}
  \label{fig:group-perfect-two-parent-join}
\end{figure}
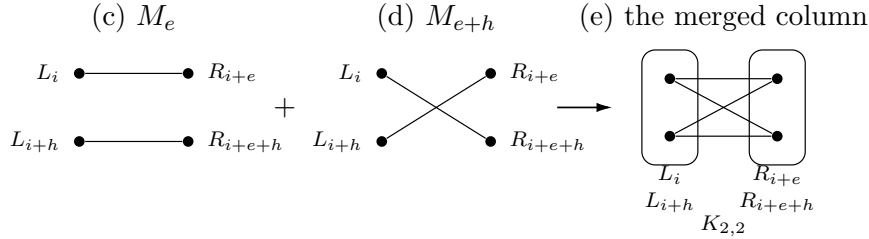

\section{Existence results}
\label{sec:existence-synthesis}

This section combines the constructions into intervals of UOM sizes.
The finite inputs are listed first.  Square kernels give the low sizes.
Row decompositions give intervals that reach the upper range.
Theorem~\ref{thm:existence-intervals} joins these intervals in each parity
class.

\subsection{Finite inputs}

Let $A=\mathbb Z_5$.  The cyclic five point square kernel in
Definition~\ref{def:square-kernel} uses the following factors:
\begin{equation}
\begin{aligned}
F_0&=\{14,23\},&F_1&=\{02,34\},&F_2&=\{04,13\},\\
F_3&=\{01,24\},&F_4&=\{03,12\}.&&
\end{aligned}
\label{eq:five-kernel}
\end{equation}
Write $\sigma\sim\tau$ when the completion rows indexed by
$\sigma,\tau\in\Sym(A)$ are orthogonal.  Equation~\eqref{eq:matching-core-adjacency}
gives
\begin{equation}
 \sigma\sim\tau
 \quad\Longleftrightarrow\quad
 \exists a\in\mathbb Z_5:
 \{\sigma(a),\tau(a)\}\in F_a.
 \label{eq:five-adjacency}
\end{equation}

\begin{lemma}
\label{lem:five-kernel-cliques}
The graph in \eqref{eq:five-adjacency} has an inclusion-maximal clique of
every order $c\in\Lambda_0$, where
\begin{equation}
\begin{aligned}
 \Lambda_0&=\{1\}\cup[4,19],\\
 \Lambda_0+\Lambda_0
 &:=\{c+d:c,d\in\Lambda_0\}=\{2\}\cup[5,38].
\end{aligned}
 \label{eq:kernel-orders}
\end{equation}
\end{lemma}

For each $c\in\Lambda_0$,
\nolinkurl{five_point_cliques.json} in the \certsite{}
gives a set $C_c\subset\Sym(A)$.
The program \nolinkurl{check_five_point_cliques.py}
checks its clique property by \eqref{eq:cert-five-clique}.
It checks maximality against all $120$ vertices by
\eqref{eq:cert-five-maximal}.
These checks prove the clique orders needed in \eqref{eq:kernel-orders}.
Appendix~\ref{app:cert-five-point} gives the full verification conditions.

\medskip
\noindent\emph{Repeated-input capacity.}
The constructions may use the same input more than once.
The following form of \eqref{eq:joint-capacity} includes these repetitions.
For a finite family $\mathcal C$ of $N$-column UOMs, each with $q$ blocks,
define
\begin{equation}
 \bcap^{\rm rep}(\mathcal C)
 =\max_{t\ge1}\ 
   \max_{X^{(1)},\ldots,X^{(t)}\in\mathcal C}
   \bcap(X^{(1)},\ldots,X^{(t)}).
 \label{eq:repetition-capacity}
\end{equation}
\Needspace{7\baselineskip}
\begin{lemma}
\phantomsection
\label{lem:certified-row-decompositions}
\noindent\textup{(i)}
There are families of input UOMs with the following row numbers and
decomposition properties.  In each row of the table, an input UOM may be
used more than once, and different listed UOMs may be used together.
\begin{equation}
\begin{array}{c|c|c|c}
N&q&\text{row numbers}&\text{condition}\\ \hline
4&2&[6,10]&\kappa_X(A_i)>2\quad(i\in[2])\\
5&2&\{6,8,9,10\}&\kappa_X(A_i)>5/2\quad(i\in[2])\\
4&4&\{8,9,10,12\}&\bcap^{\rm rep}\le3\\
5&4&[8,12]&\bcap^{\rm rep}\le3\\
6&4&\{8,9,11,12\}&\bcap^{\rm rep}\le3\\
7&4&\{8,10,11,12,13\}&\bcap^{\rm rep}\le3.
\end{array}
 \label{eq:row-decomposition-data}
\end{equation}

\medskip
\noindent\textup{(ii)}
There also exist input UOMs with five columns and six decomposition blocks,
and their row decompositions satisfy the block cover condition in
Definition~\ref{def:block-cover-condition}.  The row numbers of these UOMs
form the set $\Lambda_6$.  If $t$ input UOMs are used,
the sums of their row numbers form $t\Lambda_6$:
\begin{equation}
 \Lambda_6=\{6\}\cup[8,16],
 \qquad
 t\Lambda_6=\{6t\}\cup[6t+2,16t]
 \qquad(t\ge1).
 \label{eq:six-parent-sums}
\end{equation}
\end{lemma}

The finite matrices and their decomposition blocks are stored in
\nolinkurl{construction_data.json}.  The only exception is the $8\times5$
matrix $Y_5$ in \eqref{eq:Y5}.
Lemma~\ref{lem:Y5-profile} proves its required decomposition properties.
Equation~\eqref{eq:cert-min-supports} gives the minimal supports of the
listed matrices.  The recurrence \eqref{eq:cert-repetition-capacity}
accounts for arbitrary repetitions of the inputs.
The program \nolinkurl{check_construction_data.py} carries out these
calculations.  Appendix~\ref{app:certificate-contracts} states the checks.

\subsection{Results from square kernels}

\begin{corollary}
\label{cor:low-kernel-intervals}
The following inclusions hold:
\begin{alignat}{2}
 \{n+1\}\cup[n+4,n+19]&\subseteq\Theta_n
   &\;&(n\ge11\text{ odd}),
   \label{eq:odd-kernel-interval}\\
 \{n+2\}\cup[n+5,n+38]&\subseteq\Theta_n
   &\;&(n\ge22,\ n\equiv2\!\pmod4),
   \label{eq:two-kernel-interval}\\
 \{n+4\}\cup[n+7,n+40]&\subseteq\Theta_n
   &\;&(n\ge20,\ n\equiv0\!\pmod4).
   \label{eq:zero-kernel-interval}
\end{alignat}
\end{corollary}

\begin{proof}
Take $p=5$ in Theorem~\ref{thm:square-transfer}.
Relabel the vertices $\mathbb Z_5$ as $[5]$ in both the factors and
permutations.  This preserves adjacency in \eqref{eq:five-adjacency}.
The maximal clique orders in \eqref{eq:kernel-orders} are therefore unchanged.

For \eqref{eq:odd-kernel-interval}, set $q=n$.
Since $n\ge11$ is odd, $q\in Q_5$ by \eqref{eq:Qp}.
Theorem~\ref{thm:square-transfer}(i) gives every row number $n+c$ with $c\in\Lambda_0$.
Equation~\eqref{eq:kernel-orders} gives the stated set.

For \eqref{eq:two-kernel-interval}, set $q=n/2$ and $s=1$.
Here $q\ge11$ is odd, since $n\ge22$ and $n\equiv2\pmod4$.
Thus $q\in Q_5$.  Also, $s\le(q+1)/2$.
Theorem~\ref{thm:square-transfer}(ii) applies.
Equation~\eqref{eq:compressed-shape} gives $n$ columns
and $n+c+d$ rows.
The sumset in \eqref{eq:kernel-orders} proves the inclusion.

For \eqref{eq:zero-kernel-interval}, set $q=(n+2)/2$ and $s=3$.
Again $q\ge11$ is odd, since $n\ge20$ and $n\equiv0\pmod4$.
Thus $q\in Q_5$.  Also, $s\le(q+1)/2$.
Equation~\eqref{eq:compressed-shape} now gives $n$ columns
and $n+2+c+d$ rows.
Equation~\eqref{eq:kernel-orders} gives
$\{n+4\}\cup[n+7,n+40]$.
\end{proof}

The next result supplies the endpoints missing from
Corollary~\ref{cor:low-kernel-intervals}.
Its two padding constructions are proved in Appendix~\ref{app:square-kernel}.

\begin{theorem}
\label{thm:odd-padding}
If $q=5$, or $q\ge11$ is odd, then there exists a $(q+3)\times q$ UOM
whose row decomposition has $q+1$ blocks and satisfies the block cover
condition.  For every such $q$, there also exists a $(q+1)\times q$ UOM,
and its singleton decomposition blocks satisfy the block cover condition.
\end{theorem}

\begin{corollary}
\label{cor:padded-endpoints}
\begin{align}
 n+3&\in\Theta_n&&(n\ge11\text{ odd}),\label{eq:odd-endpoint}\\
 n+4&\in\Theta_n&&(n\ge22,\ n\equiv2\pmod4),\label{eq:two-endpoint}\\
 n+6&\in\Theta_n&&(n\ge20,\ n\equiv0\pmod4).\label{eq:zero-six-endpoint}
\end{align}
\end{corollary}

\begin{proof}
Equation~\eqref{eq:odd-endpoint} is Theorem~\ref{thm:odd-padding}
with $q=n$.

The other two inclusions use the compression
\hyperref[par:general-compression]{after Theorem~\ref{thm:square-transfer}}.
Theorem~\ref{thm:odd-padding} supplies inputs with $q+3$ and $q+1$ rows.
Both have $q$ columns and $q+1$ blocks satisfying
\eqref{eq:block-cover-condition}.  Give them disjoint column alphabets.

For \eqref{eq:two-endpoint}, take $q=n/2$ and $s=1$.
For \eqref{eq:zero-six-endpoint}, take $q=(n+2)/2$ and $s=3$.
In each case the assumed range gives odd $q\ge11$.
It also gives $0\le s\le(q+1)/2$.
The output has $2q+1-s=n$ columns.
The row number $2q+4$ is $n+4$ in the first case and $n+6$ in the second.
Figure~\ref{fig:group-perfect-universal-odd-lift} shows the padding
construction used for the two inputs.
\end{proof}

\subsection{Constructions from row decompositions and the interval from
\texorpdfstring{$2n$}{2n}}

The constructions with two and four decomposition blocks follow from
Theorem~\ref{thm:row-decomposition-construction}.

The upper interval also uses a direct sum.  This operation puts two UOMs
with the same number of columns one above the other.  It adds one column.
Its proof and the proof of the following theorem are given in
Appendix~\ref{app:row-decomposition-constructions}.

\begin{theorem}
\label{thm:uniform-tail}
For every $n\ge4$,
\begin{equation}
 [2n,2^n-6]\subseteq\Theta_n.
 \label{eq:intro-two-n-tail}
\end{equation}
\end{theorem}

The six-block construction connects the low sizes to the interval in
Theorem~\ref{thm:uniform-tail}.  Its parameters are
\begin{equation}
 t_{\mathrm{odd}}=\left\lceil\frac{n+1}{6}\right\rceil,
 \qquad
 t_{\mathrm{even}}=2\left\lceil\frac{n+1}{12}\right\rceil.
 \label{eq:six-parameters}
\end{equation}

\begin{proposition}
\label{prop:six-block-interval}
\begin{align}
 [6t_{\mathrm{odd}}+2,16t_{\mathrm{odd}}]&\subseteq\Theta_n
   && (n\ge11\text{ odd}),
 \label{eq:six-odd}\\
 [6t_{\mathrm{even}}+2,16t_{\mathrm{even}}]&\subseteq\Theta_n
   && (n\ge20\text{ even}).
 \label{eq:six-even}
\end{align}
\end{proposition}

\begin{figure}[!htbp]
  \centering
  \begingroup
  \setlength{\arraycolsep}{3.8pt}
  \renewcommand{\arraystretch}{1.13}

  {\small\textbf{The full \((8+t)\times(5+t)\) matrix}}
  \[
   \begin{gathered}
    t=q-5,\qquad B=\{\text{new single rows}\},\\[-1pt]
    A_1=\{1\},\quad A_2=\{2,7\},\quad A_3=\{3\},\\[-1pt]
    A_4=\{4,5\},\quad A_5=\{6\},\quad A_6=\{8\}.
   \end{gathered}
  \]
  Rows \(x\in B\) use the \(\delta\)-symbols.  The added columns use the
  \(\pi\)-symbols.
  \[
  \resizebox{.78\textwidth}{!}{$
  Y_q=
  \left[
  \begin{array}{c|c|c|c|c|c|c}
   \text{row}&1&2&3&4&5&\text{added columns }\ell\in[t]\\ \hline
   1&1&1&1&1&1&(\pi_{\ell,1})_{\ell\in[t]}\\
   2&1&2&3&3&3&(\pi_{\ell,2})_{\ell\in[t]}\\
   3&2&3&5&5&5&(\pi_{\ell,3})_{\ell\in[t]}\\
   4&2&4&7&7&7&(\pi_{\ell,4})_{\ell\in[t]}\\
   5&3&3&6&2&4&(\pi_{\ell,4})_{\ell\in[t]}\\
   6&3&4&8&4&2&(\pi_{\ell,5})_{\ell\in[t]}\\
   7&4&1&2&6&8&(\pi_{\ell,2})_{\ell\in[t]}\\
   8&4&2&4&8&6&(\pi_{\ell,6})_{\ell\in[t]}\\ \hline
   x\in B&\delta_{1,x}&\delta_{2,x}&\delta_{3,x}&\delta_{4,x}&\delta_{5,x}
      &(\pi_{\ell,x})_{\ell\in[t]}
  \end{array}
  \right].
  $}
  \tag*{(a)}
  \]
  The numerical block is \(Y_5\), with no change.  In each added column, the
  same \(\pi_{\ell,u}\) appears on every row of decomposition block \(A_u\).

  \vspace{-.2ex}
  \[
  \begin{array}{rcll}
   \delta_{j,x}&=&\overline{\delta_{j,y}},
      &\{x,y\}\in D_j,\quad j\in[5],\\
   \pi_{\ell,u}&=&\overline{\pi_{\ell,x}},
      &\{A_u,x\}\in P_\ell,\\
   \pi_{\ell,x}&=&\overline{\pi_{\ell,y}},
      &\{x,y\}\in P_\ell\cap H,
  \end{array}
  \qquad
  H=K_B-\disj_{j=1}^{5}D_j.
  \tag*{(b)}
  \]
  Different factor edges and columns use different symbols.  The factors
  \(D_j\) and \(P_\ell\) determine original column \(j\) and added column
  \(\ell\), respectively.

  \vspace{-.5ex}
  \begin{tikzpicture}[x=1cm,y=1cm]
    \tikzset{
      block/.style={draw=black,line width=.45pt,
        rounded corners=4pt,minimum width=1.00cm,minimum height=.54cm,
        inner sep=0pt},
      point/.style={circle,fill=black,inner sep=0pt,minimum size=4pt},
      edge/.style={draw=black,line width=.45pt}
    }
    \begin{scope}[xshift=-3.35cm]
      \node[font=\small] at (0,1.60) {(c) the new rows in original column \(j\)};
      \node[point] (dx) at (-.80,.48) {};
      \node[point] (dy) at (.80,.48) {};
      \node[point] (dz) at (-.80,-.48) {};
      \node[point] (dw) at (.80,-.48) {};
      \draw[edge] (dx)-- node[font=\scriptsize,above]
        {$\delta_{j,x}\leftrightarrow\delta_{j,y}$} (dy);
      \draw[edge] (dz)-- node[font=\scriptsize,below]
        {$\delta_{j,z}\leftrightarrow\delta_{j,w}$} (dw);
      \node[font=\scriptsize,anchor=east] at (-.92,.48) {$x$};
      \node[font=\scriptsize,anchor=west] at (.92,.48) {$y$};
      \node[font=\scriptsize,anchor=east] at (-.92,-.48) {$z$};
      \node[font=\scriptsize,anchor=west] at (.92,-.48) {$w$};
      \node[font=\scriptsize] at (0,-1.08) {$\{x,y\},\{z,w\}\in D_j$};
    \end{scope}
    \begin{scope}[xshift=2.35cm]
      \node[font=\small] at (0,1.60) {(d) one added column \(\ell\)};
      \node[block] (Au) at (-1.45,.48) {};
      \node[point] (Auone) at (-1.72,.48) {};
      \node[point] (Autwo) at (-1.45,.48) {};
      \node[point] (Authree) at (-1.18,.48) {};
      \node[font=\scriptsize,anchor=east] at (-2.02,.48) {$A_u$};
      \node[font=\scriptsize,anchor=south] at (-1.45,.79)
        {$\pi_{\ell,u}\mathbf1_{|A_u|}$};
      \node[point] (bx) at (.05,.48) {};
      \node[font=\scriptsize,anchor=south] at (.05,.58)
        {$x:\overline{\pi_{\ell,u}}$};
      \draw[edge] (Auone)--(bx) (Autwo)--(bx) (Authree)--(bx);
      \node[point] (by) at (.52,-.48) {};
      \node[point] (bz) at (1.62,-.48) {};
      \node[font=\scriptsize,anchor=north] at (.52,-.58)
        {$y:\pi_{\ell,y}$};
      \node[font=\scriptsize,anchor=north] at (1.62,-.58)
        {$z:\overline{\pi_{\ell,y}}$};
      \draw[edge] (by)--(bz);
      \node[font=\scriptsize] at (.10,-1.14)
        {$\{A_u,x\},\{y,z\}\in P_\ell$};
    \end{scope}
  \end{tikzpicture}
  \endgroup
  \caption{Odd padding used in Corollary~\ref{cor:padded-endpoints}.  The
  \(8\times5\) block is the original matrix \(Y_5\).  The \(\delta\)-entries give the five
  removed factors \(D_j\) on the new single rows.  The \(\pi\)-entries
  give the factors \(P_\ell\).  They have the same value on all rows of one
  original decomposition block.  Panels (c)--(d) show how the mate
  equations in (b) determine the entries.}
  \label{fig:group-perfect-universal-odd-lift}
\end{figure}
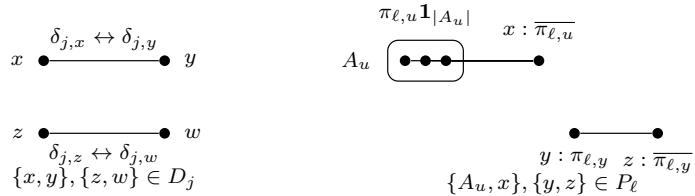

The right endpoints in Proposition~\ref{prop:six-block-interval} are at
least $2n$.  The two left endpoints are at most $n+7$ and $n+14$;
see Appendix~\ref{app:row-decomposition-constructions}.

\subsection{The \texorpdfstring{$n+5$}{n+5} construction and the three
parity classes}

Figure~\ref{fig:proof-architecture} lists the ranges needed in each parity
class.  Only the endpoint from the good column construction remains to be
proved.  Theorem~\ref{thm:existence-intervals} then checks that the ranges
overlap.  The figure uses the endpoint bounds above;
Proposition~\ref{prop:six-block-interval} gives the full intervals.

\begingroup
\newcommand{\ProofLink}[2]{\hyperref[#2]{\textcolor{LinkBlue}{#1}}}
\renewcommand{\baselinestretch}{1}
\begin{figure}[!htb]
\centering
\resizebox{.92\linewidth}{!}{%
\begin{tikzpicture}[
  x=1cm,y=1cm,
  every node/.style={font=\rmfamily\fontsize{9}{11}\selectfont},
  source/.style={draw=black!75,rounded corners=2pt,fill=black!3,
    text width=7.15cm,align=left,minimum height=1.55cm,inner sep=5pt},
  output/.style={draw=black,line width=.8pt,rounded corners=2pt,
    fill=black!10,text width=3.45cm,align=center,
    minimum height=.82cm,inner sep=4pt},
  title/.style={font=\rmfamily\bfseries\fontsize{10}{12}\selectfont,anchor=west},
  arr/.style={-{Latex[length=1.8mm]},line width=.55pt,draw=black!80}
]

\node[title] at (-6.35,6.90) {(a) Odd $n\ge11$};
\node[source] (os) at (-2.20,5.72)
  {\ProofLink{Padding}{cor:padded-endpoints}: $\{n+3\}$\\
   \ProofLink{Square kernels}{cor:low-kernel-intervals}: $[n+4,n+19]$\\
   \ProofLink{Six decomposition blocks}{prop:six-block-interval}: $[n+7,2n]$\\
   \ProofLink{Uniform tail}{thm:uniform-tail}: $[2n,2^n-6]$};
\node[output] (oo) at (4.10,5.72)
  {Target: $[n+3,2^n-6]\subseteq\Theta_n$};
\draw[arr] (os.east) -- (oo.west);

\draw[black!35] (-6.35,4.60) -- (6.20,4.60);
\node[title] at (-6.35,4.30)
  {(b) $n\equiv2\pmod4$, $n\ge22$};
\node[source] (ts) at (-2.20,3.15)
  {\ProofLink{Padding}{cor:padded-endpoints}: $\{n+4\}$\\
   \ProofLink{Square kernels}{cor:low-kernel-intervals}: $[n+5,n+38]$\\
   \ProofLink{Six decomposition blocks}{prop:six-block-interval}: $[n+14,2n]$\\
   \ProofLink{Uniform tail}{thm:uniform-tail}: $[2n,2^n-6]$};
\node[output] (to) at (4.10,3.15)
  {Target: $[n+4,2^n-6]\subseteq\Theta_n$};
\draw[arr] (ts.east) -- (to.west);

\draw[black!35] (-6.35,2.03) -- (6.20,2.03);
\node[title] at (-6.35,1.73)
  {(c) $n\equiv0\pmod4$, $n\ge20$};
\node[source,minimum height=1.82cm] (zs) at (-2.20,.48)
  {\ProofLink{Square kernels}{cor:low-kernel-intervals}:
     $\{n+4\}\cup[n+7,n+40]$\\
   \ProofLink{Good column lift}{thm:n-plus-five-endpoint}: $\{n+5\}$\\
   \ProofLink{Padding}{cor:padded-endpoints}: $\{n+6\}$\\
   \ProofLink{Six decomposition blocks}{prop:six-block-interval}: $[n+14,2n]$\\
   \ProofLink{Uniform tail}{thm:uniform-tail}: $[2n,2^n-6]$};
\node[output] (zo) at (4.10,.48)
  {Target: $[n+4,2^n-6]\subseteq\Theta_n$};
\draw[arr] (zs.east) -- (zo.west);
\end{tikzpicture}}
\caption{Roadmap for the three parity classes and the proof of
Theorem~\ref{thm:existence-intervals}.  In each panel, the left box lists the
ranges to be joined, and the right box records the target interval.  For the
six-block construction, only the bound needed for the overlap is shown.}
\label{fig:proof-architecture}
\end{figure}
\endgroup

Appendix~\ref{app:good-column-lift} gives the construction from the good
column structure in Definition~\ref{def:good-column-structure}.  Its starting
point is the $13\times8$ matrix in
Theorem~\ref{thm:good-column-thirteen-by-eight-matrix}.
The conditions in \eqref{eq:cert-good-column-acceptance} verify its UOM
property and the two cover bounds.
Equation~\eqref{eq:cert-good-column-minima} gives the exact minima.

\begin{theorem}
\label{thm:n-plus-five-endpoint}
For every $n\ge20$ with $n\equiv0\pmod4$,
\begin{equation}
 n+5\in\Theta_n.
 \label{eq:n-plus-five-endpoint}
\end{equation}
\end{theorem}

\begin{theorem}
\label{thm:existence-intervals}
\begin{align}
 [n+3,2^n-6]&\subseteq\Theta_n
   &&(n\ge11\text{ odd}),
   \label{eq:group-perfect-final-odd-tail}\\
 [n+4,2^n-6]&\subseteq\Theta_n
   &&(n\ge22,\ n\equiv2\pmod4),
   \label{eq:group-perfect-final-two-mod-four-tail}\\
 [n+4,2^n-6]&\subseteq\Theta_n
   &&(n\ge20,\ n\equiv0\pmod4).
   \label{eq:group-perfect-final-zero-mod-four-tail}
\end{align}
\end{theorem}

\begin{proof}
Each claimed interval has three parts: the low sizes, a six-block
interval, and the upper interval \eqref{eq:intro-two-n-tail}.
It is enough to show that the six-block interval meets the other two.
We first prove the endpoint bounds needed for these two overlaps.
The parameters are those in \eqref{eq:six-parameters}.

For odd $n$, the least multiple of six at least $n+1$ is at most $n+5$.
Thus
\[
 6t_{\mathrm{odd}}+2\le n+7,
 \qquad 16t_{\mathrm{odd}}\ge \tfrac83(n+1)>2n.
\]
For even $n$, $6t_{\mathrm{even}}$ is the least multiple of twelve
which is at least $n+1$.  Therefore
\[
 6t_{\mathrm{even}}+2\le n+14,
 \qquad 16t_{\mathrm{even}}\ge \tfrac83(n+1)>2n.
\]
For \eqref{eq:group-perfect-final-odd-tail},
\eqref{eq:odd-endpoint} and \eqref{eq:odd-kernel-interval} give
$n+3$ and $[n+4,n+19]$.
The interval \eqref{eq:six-odd} starts at or below $n+7$.
Its upper endpoint is greater than $2n$.
It therefore joins this range to \eqref{eq:intro-two-n-tail}.

For \eqref{eq:group-perfect-final-two-mod-four-tail},
\eqref{eq:two-endpoint} and \eqref{eq:two-kernel-interval} give
$n+4$ and $[n+5,n+38]$.
The interval \eqref{eq:six-even} starts at or below $n+14$.
Its upper endpoint is greater than $2n$.
It joins this range to \eqref{eq:intro-two-n-tail}.

For \eqref{eq:group-perfect-final-zero-mod-four-tail},
\eqref{eq:zero-kernel-interval}, \eqref{eq:n-plus-five-endpoint}, and
\eqref{eq:zero-six-endpoint} give $[n+4,n+40]$.
The same endpoint bounds for \eqref{eq:six-even} join this interval
to \eqref{eq:intro-two-n-tail}.
\end{proof}
\FloatBarrier
\section{Nonexistence for small and large sizes}
\label{sec:nonexistence}

This section gives the exclusions needed for the complete spectrum.
The new proofs apply to arbitrary UOMs.
They use the fibres in Definition~\ref{def:equality-fibre}(i) and the mate
rule in \eqref{eq:mate-rule}.  The main counts concern unordered pairs of rows.

For large row numbers, the following results hold for every
$n\ge3$:
\begin{equation}
\begin{gathered}
 2^n-5,\ 2^n-3,\ 2^n-2,\ 2^n-1\notin\Theta_n,\\
 2^n-4,\ 2^n\in\Theta_n.
\end{gathered}
\label{eq:finite-high-obstructions}
\end{equation}
Chen and Djokovic proved the nonexistence of
$2^n-5$~\cite{chen2018nonexistence}.  The other statements and the standard
constructions can be found in \cite{johnston2014structure}.

For small row numbers, Johnston's theorem \cite{johnston2013minimum} gives
the minimum values used in Theorem~\ref{thm:complete-spectrum}.  The parity
result in \cite{johnston2014structure} also shows that $n+2$ cannot occur
when $n$ is odd.

\medskip
\noindent\emph{Small fibres.}
Two short names for fibres will be useful.  By
Definition~\ref{def:equality-fibre}(i), a fibre is a set of rows.  A
fibre with two rows is called a doubleton fibre.  One with three rows is a
triple fibre.

\begin{theorem}
\label{thm:even-obstruction}
If $n\ge14$ and $n\equiv2\pmod4$, then
\begin{equation}
 n+3\notin\Theta_n.
 \label{eq:even-obstruction}
\end{equation}
\end{theorem}

The proof in Appendix~\ref{even:sec-global-obstruction} uses a fibre
covering inequality followed by a count of row pairs.
Only the transversal matching statement in
Lemma~\ref{even:lem-complete-transversal-certificate} requires a finite
enumeration.  The recurrence and its
exact counts are \eqref{eq:cert-rgs-recurrence}--\eqref{eq:cert-fixed-set-shape}
in Appendix~\ref{app:certificate-contracts}.
The program \nolinkurl{check_finite_obstructions.py} checks these counts.

It remains to consider the boundary case with ten columns.

\begin{theorem}
\label{thm:ten-boundary}
\begin{equation}
 13\notin\Theta_{10}.
 \label{eq:ten-boundary}
\end{equation}
\end{theorem}

The proof in Appendix~\ref{sec:13-by-10} reduces the possible column profiles
to three finite families.  The exact state count
\eqref{eq:cert-O3-counts} and the lower bounds on repetition loss
\eqref{eq:cert-two-special-minima}--\eqref{eq:cert-one-special-counts}
exclude these families.

In both proofs, the symbolic arguments reduce every possible UOM to the
listed matching families.  The finite searches check these remaining
cases.  Thus no assumption on a special form of UOM is needed.
\section{Finite cases}
\label{sec:finite-ledger}

Some small dimensions are not covered by the infinite constructions.  This
section settles them.  The $23$ matrices used here are collected in
\nolinkurl{finite_uoms.txt}, available in the \certsite.

Each entry gives the size, a short construction note, and the matrix itself.
The mate rule in \eqref{eq:mate-rule} checks pairwise orthogonality.
The fibre cover test in Theorem~\ref{thm:fibre-cover}(i) checks
unextendibility.

\Needspace{8\baselineskip}
\begin{theorem}
\label{thm:finite-new-certificates}
The matrices in {\normalfont\nolinkurl{finite_uoms.txt}} give the following inclusions:
\begin{align}
 \{17,18,19,21,22,23\}&\subseteq\Theta_{12},
 \label{eq:finite12}\\
 \{18,19,21,22,23\}&\subseteq\Theta_{14},
 \label{eq:finite14}\\
 \{21,22,23,25\}&\subseteq\Theta_{16},
 \label{eq:finite16}\\
 \{22,23,25,26,27,29,30,31\}&\subseteq\Theta_{18}.
 \label{eq:finite18}
\end{align}
\end{theorem}

\begin{proof}
For every matrix $X$ in \nolinkurl{finite_uoms.txt}, the mate
comparisons give \eqref{eq:cert-orth}.
The recurrence \eqref{eq:cert-fibre-dp} gives $[m]\notin\mathcal R_n$.
Equation~\eqref{eq:cert-uom-dp} therefore proves
\eqref{eq:finite12}--\eqref{eq:finite18}.

The following four matrices use the compression described in the proof of
Theorem~\ref{thm:square-transfer}(ii):
\begin{equation}
 21\times12,\qquad21\times14,\qquad22\times14,\qquad23\times14.
\label{eq:finite-compressed}
\end{equation}
Here $q=7$.  The first matrix uses inputs with $13$ and $8$ rows.
Its compression parameter is $s=3$.
The other three use an $8$-row input paired with an input having
$13$, $14$, or $15$ rows, respectively.  In these cases, $s=1$.
The relevant construction is the consequence stated
\hyperref[par:general-compression]{after that proof}.
It does not require inputs from the five point kernel.
For the following five sizes, the listed matrix itself is the finite witness:
\begin{equation}
 18\times14,\quad19\times14,\quad21\times16,
 \quad22\times16,\quad22\times18.
 \label{eq:finite-direct}
\end{equation}
The $21\times16$ record is a finite completion witness.
The $17\times12$, $18\times12$, and $19\times12$ records use the same
matching core.  Each adds a maximal completion clique to this core, as in
Theorem~\ref{thm:matching-core-full-link}(iii).
The $22\times12$ and $23\times12$ records use
Theorem~\ref{thm:row-decomposition-construction}.
The other records use compression or the six-block inputs of
Lemma~\ref{lem:certified-row-decompositions}(ii).
Their parameters are in the file.
\end{proof}

\begin{remark}
The file \nolinkurl{finite_uoms.txt} is available in the \certsite.  It contains
only the size, method note, and matrix before the next entry begins.  The
exact finite-set calculation for the fibre cover test is
\eqref{eq:cert-fibre-dp}--\eqref{eq:cert-uom-dp} in
Appendix~\ref{app:certificate-contracts}.
\end{remark}

\Needspace{10\baselineskip}
\begin{theorem}
\label{thm:finite-boundary-spectra}
\begin{align}
\Theta_{10}&=\{12\}\cup[14,1018]\cup\{1020,1024\},
\label{eq:spectrum10}\\
\Theta_{12}&=[16,4090]\cup\{4092,4096\},
\label{eq:spectrum12}\\
\Theta_{14}&=\{16\}\cup[18,16378]\cup\{16380,16384\},
\label{eq:spectrum14}\\
\Theta_{16}&=[20,65530]\cup\{65532,65536\},
\label{eq:spectrum16}\\
\Theta_{18}&=\{20\}\cup[22,262138]\cup\{262140,262144\}.
\label{eq:spectrum18}
\end{align}
\end{theorem}

\begin{proof}
It is enough to decide the sizes between the known minimum and
$2n$.  Equation~\eqref{eq:intro-two-n-tail} gives $[2n,2^n-6]$.
Equation~\eqref{eq:finite-high-obstructions} determines the larger sizes.
For $n=14,16,18$, a common four-input construction supplies most of the
remaining sizes.  The listed matrices and two applications of compression
fill the gaps.

The minimum in each case is given in \cite{johnston2013minimum}.

For \eqref{eq:spectrum10}, the minimum is $12$.
The six-block construction in
Appendix~\ref{app:row-decomposition-constructions} permits $t=2$.
Choose $k=5$ from \eqref{eq:K66-decompositions}.
The output has $5+k=10$ columns.
Equation~\eqref{eq:six-parent-sums} gives
$2\Lambda_6=\{12\}\cup[14,32]$ as its row numbers.
This set meets the interval $[20,1018]$ from
Theorem~\ref{thm:uniform-tail}.
The remaining value $13$ is excluded by Theorem~\ref{thm:ten-boundary}.

For \eqref{eq:spectrum12}, the minimum is $16$.
The sizes still needed below $2n=24$ are $17$ through $23$.
Equation~\eqref{eq:finite12} gives $17,18,19,21,22,23$.
The standard construction for row numbers divisible by four gives $20$
\cite{johnston2014structure}.

For $n=14,16,18$, use four inputs with six blocks each.
The choices in \eqref{eq:K66-decompositions} allow every
$k\in[3,18]$ when $t=4$.
Take $k=9,11,13$, respectively.  Equation~\eqref{eq:six-parent-sums}
then gives
\[
 \{24\}\cup[26,64]\subseteq\Theta_n
 \qquad(n=14,16,18).
\]

For \eqref{eq:spectrum14}, the minimum is $16$.  Theorem~\ref{thm:even-obstruction}
excludes $17$.  Equation~\eqref{eq:finite14} gives $18,19,21,22,23$.
The sizes $20$ and $25$ follow from two applications of compression.
The inputs are submatrices of the $25\times16$ record in
\nolinkurl{finite_uoms.txt}.  To check the compression hypotheses, first
consider its first nine columns.  They restrict to a
$10\times9$ UOM on rows $1$--$10$ and a $15\times9$ UOM on rows
$11$--$25$.  These assertions follow from
\eqref{eq:cert-orth} and \eqref{eq:cert-uom-dp}.
Every fibre of the first input is a singleton.
In the second input, the first nine rows form an orthogonal $9\times9$
core with singleton fibres.  The other six rows form one completion block.
The block count in the proof of Theorem~\ref{thm:square-transfer}(i)
therefore gives ten blocks satisfying \eqref{eq:block-cover-condition}
for this input.  The singleton blocks of the first input have the same
property.

In each application, give the two input copies disjoint column alphabets.
The compression
\hyperref[par:general-compression]{after Theorem~\ref{thm:square-transfer}}
applies with $q=9$ and $s=5$.
Using the first input twice gives $20\times14$.
Using one of each gives $25\times14$.

The common four-input range supplies $24,26,27$ for $n=14$.
This accounts for all sizes below the upper interval.

For \eqref{eq:spectrum16}, the minimum is $20$.
Equation~\eqref{eq:finite16} supplies $21,22,23,25$.
The four-input range supplies $24$ and $[26,31]$.

For \eqref{eq:spectrum18}, the minimum is $20$.
Equation~\eqref{eq:even-obstruction} excludes $21$.
Equation~\eqref{eq:finite18} supplies $22,23,25,26,27,29,30,31$.
The four-input range supplies $24,28$ and $[32,35]$.
Together with \eqref{eq:intro-two-n-tail} and
\eqref{eq:finite-high-obstructions}, these checks prove
\eqref{eq:spectrum10}--\eqref{eq:spectrum18}.
\end{proof}

The finite cases above and Theorem~\ref{thm:existence-intervals} complete
the proof of Theorem~\ref{thm:complete-spectrum}.
\begin{proof}
The constructed intervals give the middle part of each spectrum.
The remaining sizes follow from the minimum results and
Section~\ref{sec:nonexistence}.
For $n\le8$, equation~\eqref{eq:small-spectra} follows from the known
results recalled in Section~\ref{sec:introduction}.  The remaining odd
base case is $n=9$.  The known minimum and parity results give
$10\in\Theta_9$ and $11\notin\Theta_9$.
The sizes still needed below the upper interval starting at $18$ are
$12$ through $17$.  Take two input UOMs with six
decomposition blocks.  By \eqref{eq:six-parent-sums}, their possible row
sums are
\[
 2\Lambda_6=\{12\}\cup[14,32].
\]
For $t=2$, the second six-block construction in
Appendix~\ref{app:row-decomposition-constructions} uses the choices in
\eqref{eq:K66-decompositions}.  These permit every
$k\in[1,6]$.  Choose $k=4$.  The output has $5+4=9$ columns.  Thus
$\{12\}\cup[14,32]\subseteq\Theta_9$.
The matrix $M_{13,9}$ in Appendix~\ref{app:explicit-matrices} supplies $13$.
Equation~\eqref{eq:intro-two-n-tail} gives $[18,506]$.
Equation~\eqref{eq:finite-high-obstructions} determines the six larger
sizes.  This proves \eqref{eq:odd-spectrum} for $n=9$.

For odd $n\ge11$, \eqref{eq:group-perfect-final-odd-tail} gives the middle
interval.  The minimum and parity results in
\cite{johnston2013minimum,johnston2014structure} determine the smaller sizes.
Equation~\eqref{eq:finite-high-obstructions} determines the larger sizes.
These results prove \eqref{eq:odd-spectrum}.

Suppose that $n\equiv2\pmod4$.  Theorem~\ref{thm:finite-boundary-spectra}
gives the cases $n=10,14,18$.  For $n\ge22$,
\eqref{eq:group-perfect-final-two-mod-four-tail} gives the middle interval.
The minimum result in \cite{johnston2013minimum} and
Theorem~\ref{thm:even-obstruction} determine the smaller sizes.
Together with \eqref{eq:finite-high-obstructions}, this proves
\eqref{eq:two-spectrum}.

Suppose that $n\equiv0\pmod4$.  Theorem~\ref{thm:finite-boundary-spectra}
gives the cases $n=12,16$.  For $n\ge20$,
\eqref{eq:group-perfect-final-zero-mod-four-tail} gives the middle interval.
The minimum result in \cite{johnston2013minimum} determines the smaller sizes.
The larger sizes are given by \eqref{eq:finite-high-obstructions}.
Thus \eqref{eq:zero-spectrum} holds.
All $n\ge1$ are covered.
\end{proof}

As an example, the file \nolinkurl{20_qubit_uoms.txt} in the
\certsite\ gives the complete $20$-qubit size spectrum from
\eqref{eq:zero-spectrum}.
It lists an explicit $m\times20$ UOM for every $m\in[24,150]$.

\section{Conclusion}

Theorem~\ref{thm:complete-spectrum} determines all possible multiqubit UPB
sizes.  The constructions use row decompositions and completion graphs.
Fibre covers prove unextendibility.  Counts of row pairs exclude the
remaining sizes.  Extending these methods to higher local dimensions is
a natural next question.  Classifying all multiqubit UPBs of a fixed
size also remains open.

\section*{AI-Statement}

The authors used ChatGPT (GPT-5.6 Sol) for language polishing and LaTeX
editing during the development and preparation of this paper.
Some ideas used in the proof of Theorem~\ref{thm:even-obstruction} arose
during interactions with GPT-5.6 Sol.
All mathematical statements and proofs were independently verified by the
authors. The authors take responsibility for the final manuscript.

\section*{Acknowledgments}

Authors were supported by the NNSF of China (Grant No. 12471427).
\appendix
\section{Constructions from row decompositions and the upper interval}
\label{app:row-decomposition-constructions}

This appendix gives the inputs and added columns needed to apply
Theorem~\ref{thm:row-decomposition-construction} with two, four, or six
decomposition blocks.  The constructions give intervals of row numbers.
A direct sum then extends these intervals to the upper end of the spectrum.
Row decompositions and added column systems are defined in
Definitions~\ref{def:row-decomposition} and
\ref{def:added-column-system}.  The program
\nolinkurl{check_construction_data.py} checks the finite inputs.  It uses the
orthogonality and fibre cover conditions in
\eqref{eq:cert-orth}--\eqref{eq:cert-uom-dp}, followed by the minimal support
and repetition conditions in \eqref{eq:cert-min-supports} and
\eqref{eq:cert-repetition-capacity} of
Appendix~\ref{app:certificate-contracts}.

\subsection{Two decomposition blocks}

\begin{lemma}
\label{lem:two-block-test}
Suppose that every input UOM with $N$ columns has the row decomposition
$A_0\disj A_1$ and satisfies
\begin{equation}
 \kappa_X(A_0)>N/2,\qquad \kappa_X(A_1)>N/2.
 \label{eq:two-block-cost}
\end{equation}
Here $\kappa_X$ is the cover cost in
Definition~\ref{def:equality-fibre}(iii) and \eqref{eq:cover-cost}.  Then the
joint block capacity of any finite list of these input UOMs is at most one,
where an input UOM is allowed to appear more than one time.
\end{lemma}

\begin{proof}
We claim that a feasible choice in \eqref{eq:joint-capacity} cannot cover
two blocks.  Indeed, suppose that it does.
If they belong to one input, the chosen fibres cover its whole
row set.  This contradicts \eqref{eq:fibre-cover}.
If they belong to distinct inputs, their supports must be disjoint by
\eqref{eq:joint-capacity}.  Equation~\eqref{eq:two-block-cost} makes their
total size greater than $N$.
This is impossible within $N$ original columns.
Thus the joint capacity is at most one.
\end{proof}

Lemma~\ref{lem:two-block-test} gives the capacity bound required by
Theorem~\ref{thm:row-decomposition-construction}(i).
Use the matching added columns described after
Definition~\ref{def:added-column-system}.  For $t$ input UOMs,
part (iii) of the theorem supplies $2(t-1)$ added columns.
The row numbers are sums of the input row numbers.
The verified inputs with four and five columns give
\begin{align}
 [3n-6,5n-10]&\subseteq\Theta_n&&(n\ge6\text{ even}),
 \label{eq:two-block-even}\\
 \{3n-9\}\cup[3n-7,5n-15]&\subseteq\Theta_n
 &&(n\ge5\text{ odd}).
 \label{eq:two-block-odd}
\end{align}
Here the first set is the sumset $t[6,10]$, where $n=4+2(t-1)$, and the
second set is $t\{6,8,9,10\}$, where $n=5+2(t-1)$.

\subsection{Four decomposition blocks and repeated input matrices}

Let $X$ be an input UOM with four decomposition blocks.  Take
$I\subseteq[4]$.  Let $\Sfam_X(I)$ be the fibre support family in
Definition~\ref{def:equality-fibre}(ii), using the block notation introduced
after Definition~\ref{def:row-decomposition}.  Denote its inclusion-minimal
members by $\mathcal M_X(I)$.  Let $\mathcal C$ be a family of
verified input UOMs, where every matrix has $N$ columns and four
decomposition blocks.

Theorem~\ref{thm:row-decomposition-construction}(i) requires
$\bcap^{\rm rep}(\mathcal C)\le3$.
The program \nolinkurl{check_construction_data.py} verifies it from all
minimal supports in \eqref{eq:cert-min-supports} using
\eqref{eq:cert-repetition-capacity}.
The recurrence allows an input to be chosen again at a later state.
It therefore includes repeated inputs.  Its proof is in
Appendix~\ref{app:certificate-contracts}.

Let $\alpha_j$ be the minimum size of a fibre support covering $j$ blocks.
The following table gives these lower bounds and the column restrictions
when equality holds:
\begin{equation}
\begin{gathered}
\begin{array}{c|c}
N&[\alpha_1,\alpha_2,\alpha_3]\\ \hline
4&[1,3,4]\\
5&[2,3,4]\\
6&[2,3,5]\\
7&[2,4,6]
\end{array}\\[.7em]
\begin{array}{c|l}
N&\text{additional equality check}\\ \hline
4&\text{every size-$1$ support of one block contains column }4\\
5&\text{none}\\
6&\text{every size-$3$ support of two blocks contains column }6\\
7&\text{none}.
\end{array}
\end{gathered}
\label{eq:four-support-profiles}
\end{equation}
The test uses the full list of minimal supports.  The three numbers
$\alpha_1,\alpha_2,\alpha_3$ give useful lower bounds, but they alone do
not determine the repetition capacity.

The capacity bound permits any list of the verified inputs, including
repetitions.  It remains to determine the sums of their row numbers:
\begin{align}
t\{8,9,10,12\}&=[8t,12t-2]\cup\{12t\},
\label{eq:four-sum1}\\
t[8,12]&=[8t,12t],
\label{eq:four-sum2}\\
t\{8,9,11,12\}&=[8t,12t]&&(t\ge2),
\label{eq:four-sum3}\\
t\{8,10,11,12,13\}&=\{8t\}\cup[8t+2,13t].
\label{eq:four-sum4}
\end{align}
These identities follow by induction on $t$ after subtracting $8t$ from
both sides.  The four sets which occur in the induction are
$\{0,1,2,4\}$, $[0,4]$, $\{0,1,3,4\}$, and
$\{0,2,3,4,5\}$, in the same order as the four equations above.

These sumsets give an interval near the lower end of the spectrum.
The two-block intervals extend it upwards.
For $n=N+4(t-1)$, Theorem~\ref{thm:row-decomposition-construction},
\eqref{eq:four-sum1}--\eqref{eq:four-sum4}, and the two-block intervals give
\begin{equation}
\begin{array}{c|c}
n\pmod4&\text{continuous interval in }\Theta_n\\ \hline
0&[2n,5n-10]\\
1&[2n-2,5n-15]\\
2&[2n-4,5n-10]\quad(n\ge10)\\
3&[2n-4,5n-15].
\end{array}
\label{eq:decomposition-block-interval}
\end{equation}

Figure~\ref{twon:fig-four-group-lift} shows the complete block matrix and
the four one-factors used in the added columns.  The term one-factor has the
meaning in Definition~\ref{def:graph-notation}(i).

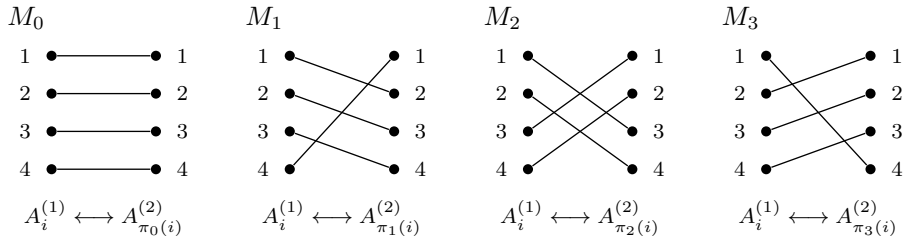
\begin{figure}[!htbp]
  \centering
  \begingroup
  \setlength{\arraycolsep}{4.8pt}
  \renewcommand{\arraystretch}{1.16}
  \small
  \begin{equation*}
    X^{(v)}=
    \begin{bmatrix}X^{(v)}_1\\X^{(v)}_2\\X^{(v)}_3\\X^{(v)}_4\end{bmatrix},
    \qquad
    R\!\left(X^{(v)}_i\right)=A^{(v)}_i,
    \qquad
    r_{v,i}=|A^{(v)}_i|,
    \qquad v\in\{1,2\},
    \qquad\text{\rm(a)}
  \end{equation*}

  \vspace{.2em}

  \begin{align*}
    \pi_e(i)&=1+\bigl((i+e-1)\bmod4\bigr),
    &\rho_e&=\pi_e^{-1},\\
    M_e&=\bigl\{\{A^{(1)}_i,A^{(2)}_{\pi_e(i)}\}:i\in[4]\bigr\}.
    \tag*{\rm(b)}
  \end{align*}
  For every $e$ and $i$, the symbols $\alpha_{e,i}$ and
  $\bar\alpha_{e,i}$ form a new mate pair.
  Different edges use different mate pairs.

  \vspace{.25em}

  \resizebox{.985\linewidth}{!}{$
  Y=
  \left[
  \begin{array}{c|cccc}
    \multicolumn{1}{c|}{\text{$N$ original columns}}
       &c_0&c_1&c_2&c_3\\ \hline
    X^{(1)}_1
       &\alpha_{0,1}\mathbf 1_{r_{1,1}}
       &\alpha_{1,1}\mathbf 1_{r_{1,1}}
       &\alpha_{2,1}\mathbf 1_{r_{1,1}}
       &\alpha_{3,1}\mathbf 1_{r_{1,1}}\\
    X^{(1)}_2
       &\alpha_{0,2}\mathbf 1_{r_{1,2}}
       &\alpha_{1,2}\mathbf 1_{r_{1,2}}
       &\alpha_{2,2}\mathbf 1_{r_{1,2}}
       &\alpha_{3,2}\mathbf 1_{r_{1,2}}\\
    X^{(1)}_3
       &\alpha_{0,3}\mathbf 1_{r_{1,3}}
       &\alpha_{1,3}\mathbf 1_{r_{1,3}}
       &\alpha_{2,3}\mathbf 1_{r_{1,3}}
       &\alpha_{3,3}\mathbf 1_{r_{1,3}}\\
    X^{(1)}_4
       &\alpha_{0,4}\mathbf 1_{r_{1,4}}
       &\alpha_{1,4}\mathbf 1_{r_{1,4}}
       &\alpha_{2,4}\mathbf 1_{r_{1,4}}
       &\alpha_{3,4}\mathbf 1_{r_{1,4}}\\ \hline
    X^{(2)}_1
       &\bar\alpha_{0,\rho_0(1)}\mathbf 1_{r_{2,1}}
       &\bar\alpha_{1,\rho_1(1)}\mathbf 1_{r_{2,1}}
       &\bar\alpha_{2,\rho_2(1)}\mathbf 1_{r_{2,1}}
       &\bar\alpha_{3,\rho_3(1)}\mathbf 1_{r_{2,1}}\\
    X^{(2)}_2
       &\bar\alpha_{0,\rho_0(2)}\mathbf 1_{r_{2,2}}
       &\bar\alpha_{1,\rho_1(2)}\mathbf 1_{r_{2,2}}
       &\bar\alpha_{2,\rho_2(2)}\mathbf 1_{r_{2,2}}
       &\bar\alpha_{3,\rho_3(2)}\mathbf 1_{r_{2,2}}\\
    X^{(2)}_3
       &\bar\alpha_{0,\rho_0(3)}\mathbf 1_{r_{2,3}}
       &\bar\alpha_{1,\rho_1(3)}\mathbf 1_{r_{2,3}}
       &\bar\alpha_{2,\rho_2(3)}\mathbf 1_{r_{2,3}}
       &\bar\alpha_{3,\rho_3(3)}\mathbf 1_{r_{2,3}}\\
    X^{(2)}_4
       &\bar\alpha_{0,\rho_0(4)}\mathbf 1_{r_{2,4}}
       &\bar\alpha_{1,\rho_1(4)}\mathbf 1_{r_{2,4}}
       &\bar\alpha_{2,\rho_2(4)}\mathbf 1_{r_{2,4}}
       &\bar\alpha_{3,\rho_3(4)}\mathbf 1_{r_{2,4}}
  \end{array}
  \right].
  $}

  \vspace{.8em}

  For $a\in[r_{v,i}]$,
  \begin{equation*}
    Y_{(v,i,a),j}=
    \begin{cases}
      X^{(v)}_{(i,a),j},
        &1\le j\le N,\\
      \alpha_{e,i},
        &v=1,\ j=N+e+1,\\
      \bar\alpha_{e,\rho_e(i)},
        &v=2,\ j=N+e+1,
    \end{cases}
    \qquad e\in\{0,1,2,3\}.
    \qquad\text{\rm(c)}
  \end{equation*}

  \vspace{.35em}

  \begin{tikzpicture}[
    x=1cm,y=1cm,
    pt/.style={circle,fill=black,inner sep=0pt,minimum size=3.8pt},
    edge/.style={draw=black,line width=.5pt},
    lab/.style={font=\scriptsize}
  ]
    \node[font=\footnotesize,anchor=west] at (0,2.48)
      {(d) the one-factor used for each added column};
    \foreach \e/\x in {0/0,1/3.15,2/6.30,3/9.45}{
      \node[font=\footnotesize,anchor=west] at (\x,2.02) {$M_{\e}$};
      \foreach \i/\yy in {1/1.52,2/1.02,3/.52,4/.02}{
        \node[pt] (L\e\i) at (\x+.72,\yy) {};
        \node[pt] (R\e\i) at (\x+2.10,\yy) {};
        \node[lab,anchor=east] at (\x+.58,\yy) {$\i$};
        \node[lab,anchor=west] at (\x+2.24,\yy) {$\i$};
      }
    }
    \foreach \i in {1,...,4}{
      \pgfmathtruncatemacro{\j}{mod(\i-1+0,4)+1}
      \draw[edge] (L0\i)--(R0\j);
      \pgfmathtruncatemacro{\j}{mod(\i-1+1,4)+1}
      \draw[edge] (L1\i)--(R1\j);
      \pgfmathtruncatemacro{\j}{mod(\i-1+2,4)+1}
      \draw[edge] (L2\i)--(R2\j);
      \pgfmathtruncatemacro{\j}{mod(\i-1+3,4)+1}
      \draw[edge] (L3\i)--(R3\j);
    }
    \node[lab] at (1.41,-.62) {$A_i^{(1)}\longleftrightarrow A_{\pi_0(i)}^{(2)}$};
    \node[lab] at (4.56,-.62) {$A_i^{(1)}\longleftrightarrow A_{\pi_1(i)}^{(2)}$};
    \node[lab] at (7.71,-.62) {$A_i^{(1)}\longleftrightarrow A_{\pi_2(i)}^{(2)}$};
    \node[lab] at (10.86,-.62) {$A_i^{(1)}\longleftrightarrow A_{\pi_3(i)}^{(2)}$};
  \end{tikzpicture}
  \endgroup
  \caption{The construction from four decomposition blocks in every input
  UOM.  The original part of a row is the same as the corresponding row of
  its input UOM.  The entries in the four added columns are given by
  formula (c), and panel (d) gives the one-factor which is used for each
  added column.}
  \label{twon:fig-four-group-lift}
\end{figure}

\subsection{Weighted added column systems for six blocks}

\medskip
\noindent\emph{Paired triples.}
This construction gives every even number of added columns in
$[2(t-1),6(t-1)]$, with total weight $6(t-1)$.
The number of columns changes by splitting selected columns into three.
Their total weight stays the same.
To construct the initial columns, divide the six decomposition blocks of
each input UOM into two triples.  For
$t$ input UOMs there are $2t$ triples.  Pair the two triples from each
input.  Remove these pairs from the complete graph on the triples:
\begin{equation}
 Q_t=K_{2t}-tK_2.
 \label{eq:cocktail}
\end{equation}
By Theorem~\ref{thm:standard-graph-tools}(i), the deleted matching belongs
to a one-factorization of $K_{2t}$.  Remove its factor.  Each of the other
$2t-2$ factors gives a union of $K_{3,3}$ components.
These components contain every decomposition block vertex in
\eqref{eq:demand-graph}.  Over all factors, they give every edge of that
demand graph.
For every component, put one symbol on one part and its mate on the other.
Use a different mate pair for each component.
This gives one column of weight three for the whole factor.
The column can instead be split as follows.
Factor every component into three one-factors.
Use the first one-factor of every component in one column, then do the
same for the second and third one-factors.
This gives three columns of weight one.
Make the same choice for all components of a factor.  If the columns
for $a$ factors are each split into three, then
\begin{equation}
 k=2(t-1)+2a,\qquad0\le a\le2(t-1),\qquad
 \sum_c\mu_c=6(t-1).
 \label{eq:six-even-counts}
\end{equation}
The possible values of $k$ are all the even numbers in
$[2(t-1),6(t-1)]$.

\medskip
\noindent\emph{Factorized input pairs.}
When $t$ is even, the number of added columns can be any integer in
$[t-1,6(t-1)]$.  Their total weight is $6(t-1)$.
For this construction, divide $K_t$ into $t-1$ one-factors using
Theorem~\ref{thm:standard-graph-tools}(i).  Here each vertex represents one
input.  Each edge of $K_t$ represents one copy of $K_{6,6}$ between the
blocks of its two inputs.  For a
fixed factor, label both parts of each copy by $\mathbb Z_6$.  Use
the cyclic one-factors
\[
 M_e=\{L_iR_{i+e}:i\in\mathbb Z_6\},\qquad e\in\mathbb Z_6.
\]
These six factors partition $K_{6,6}$.
The union $M_0\cup M_2\cup M_4$ consists of two copies of $K_{3,3}$.
The same holds for $M_1\cup M_3\cup M_5$.
The union $M_e\cup M_{e+3}$ consists of three copies of $K_{2,2}$.
These decompositions allow any number of groups from one to six.
All six factors can be kept together, divided into two triples,
or divided into three opposite pairs.  Keeping one triple together and
the other three factors separate gives four groups.
Keeping one opposite pair together and the other four factors separate
gives five groups.  Keeping all factors separate gives six groups.
For each group, use a different mate pair on every complete bipartite
component.  Make the same grouping choice for every input pair in the fixed
factor of $K_t$.  Put corresponding groups in the same column.
The number of columns and their weights are
\begin{equation}
\begin{array}{c|c}
1&6\\2&3+3\\3&2+2+2\\4&3+1+1+1\\
5&2+1+1+1+1\\6&1+1+1+1+1+1.
\end{array}
\label{eq:K66-decompositions}
\end{equation}
Choose one line of
\eqref{eq:K66-decompositions} for each factor of $K_t$.
Every $k\in[t-1,6(t-1)]$ is a sum of $t-1$ integers from $[1,6]$.
The total weight is always $6(t-1)$.

The added columns above satisfy the first bound in
\eqref{eq:row-decomposition-hypothesis}: their total weight is $6(t-1)$.
For the second bound, the six-block inputs satisfy
\eqref{eq:block-cover-condition}.  Thus
\eqref{eq:block-cover-condition-capacity} bounds their joint capacity by
five.  Theorem~\ref{thm:row-decomposition-construction}(i) gives UOMs.

For Proposition~\ref{prop:six-block-interval}, the required number of
added columns is $k=n-5$.  The row sums are
\begin{equation}
 \begin{aligned}
 t\Lambda_6
   &=6t+t(\{0\}\cup[2,10])\\
   &=\{6t\}\cup[6t+2,16t],
 \end{aligned}
 \label{eq:six-row-sums}
\end{equation}
which proves \eqref{eq:six-parent-sums}.
For \eqref{eq:six-odd}, take $t=t_{\mathrm{odd}}$ from
\eqref{eq:six-parameters}.  Then $k=n-5$ is even.
The bounds $6t\ge n+1$ and $t\le(n+5)/6$ give
\[
 2(t-1)\le(n-1)/3\le n-5\le6(t-1).
\]
Thus \eqref{eq:six-even-counts} supplies $k$ columns.

For \eqref{eq:six-even}, take $t=t_{\mathrm{even}}$.
This is even, with $6t\ge n+1$ and $t\le(n+12)/6$.
Hence
\[
 t-1\le(n+6)/6\le n-5\le6(t-1).
\]
The choices in \eqref{eq:K66-decompositions} supply every such $k$.
In each case \eqref{eq:six-row-sums} gives the required rows in
$5+k=n$ columns.

\subsection{Closure and the upper interval}

\begin{lemma}
\label{lem:direct-sum}
If $a,b\in\Theta_N$, then $a+b\in\Theta_{N+1}$.
\end{lemma}

\begin{proof}
Stack UOMs of sizes $a\times N$ and $b\times N$.
Add one column with $x$ on the first input and $\bar x$ on the second.
Rows within either input remain orthogonal.
Rows from distinct inputs are orthogonal in the added column.

An extension cannot cover both inputs in that column, by
\eqref{eq:mate-rule}.  Its first $N$ entries would therefore extend one
whole input, contrary to Definition~\ref{def:uom}(iii).
\end{proof}

The following rule gives another way to combine UOMs.
It replaces each row of one UOM by a group of rows.

\begin{lemma}
\label{lem:row-replacement}
Let $X$ be a $t\times s$ UOM, and for every $i\in[t]$ let $Y_i$ be an
$a_i\times r$ UOM.  Replace row $i$ of $X$ by $a_i$ copies of this row,
and put $Y_i$ beside these copies in $r$ new columns.  Then the resulting
matrix is a
\begin{equation}
 \left(\sum_{i=1}^t a_i\right)\times(s+r)\text{ UOM}.
\end{equation}
\end{lemma}

\begin{proof}
Rows in the same replacement group are orthogonal in the $Y_i$ columns.
Rows in distinct groups are orthogonal in the $X$ columns.

For unextendibility, suppose that an extension exists.
Its last $r$ entries would extend one of the $Y_i$.
Indeed, its first $s$ entries leave some row $i$ of $X$ uncovered by
\eqref{eq:fibre-cover}.  They therefore leave that whole replacement
group uncovered.  Its last $r$ entries must then extend $Y_i$.
This is impossible.
\end{proof}

The new matrix is shown in Figure~\ref{fig:row-replacement-picture}, where
the source of every entry is also given.

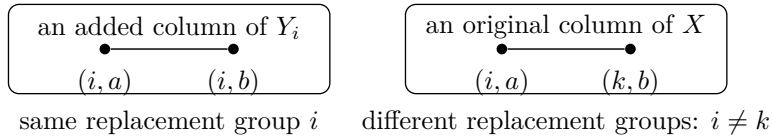
\begin{figure}[!htbp]
  \centering
  \begingroup
  \setlength{\arraycolsep}{4.2pt}
  \renewcommand{\arraystretch}{1.13}
  \small
  \begin{equation*}
    X=(X_{ij})_{t\times s},
    \qquad
    Y_i=\bigl((Y_i)_{a\ell}\bigr)_{a_i\times r},
    \qquad
    Z=X[Y_1,\ldots,Y_t].
    \qquad\text{\rm(a)}
  \end{equation*}

  \vspace{.4em}

  \resizebox{.985\linewidth}{!}{$
  \begin{array}{c@{\qquad}c}
  \begin{array}{c}
    \text{original row}\vphantom{\displaystyle\sum_{a=1}^{a_i}}\\[4pt]
    X_{1,*}\\[-1pt]\vdots\\[-1pt]
    X_{i,*}\\[-1pt]\vdots\\[-1pt]
    X_{t,*}
  \end{array}
  &
  Z=
  \left[
  \begin{array}{c|cccc|cccc}
    &\multicolumn{4}{c|}{\text{$s$ original columns}}
    &\multicolumn{4}{c}{\text{$r$ added columns}}\\
    \text{row}&1&2&\cdots&s&1&2&\cdots&r\\ \hline
    (1,1)&X_{11}&X_{12}&\cdots&X_{1s}
      &(Y_1)_{11}&(Y_1)_{12}&\cdots&(Y_1)_{1r}\\
    \vdots&\vdots&\vdots&&\vdots&\vdots&\vdots&&\vdots\\
    (1,a_1)&X_{11}&X_{12}&\cdots&X_{1s}
      &(Y_1)_{a_1,1}&(Y_1)_{a_1,2}&\cdots&(Y_1)_{a_1,r}\\ \hline
    \vdots&\vdots&\vdots&&\vdots&\vdots&\vdots&&\vdots\\[-2pt]
    (i,1)&X_{i1}&X_{i2}&\cdots&X_{is}
      &(Y_i)_{11}&(Y_i)_{12}&\cdots&(Y_i)_{1r}\\
    \vdots&\vdots&\vdots&&\vdots&\vdots&\vdots&&\vdots\\
    (i,a_i)&X_{i1}&X_{i2}&\cdots&X_{is}
      &(Y_i)_{a_i,1}&(Y_i)_{a_i,2}&\cdots&(Y_i)_{a_i,r}\\ \hline
    \vdots&\vdots&\vdots&&\vdots&\vdots&\vdots&&\vdots\\[-2pt]
    (t,1)&X_{t1}&X_{t2}&\cdots&X_{ts}
      &(Y_t)_{11}&(Y_t)_{12}&\cdots&(Y_t)_{1r}\\
    \vdots&\vdots&\vdots&&\vdots&\vdots&\vdots&&\vdots\\
    (t,a_t)&X_{t1}&X_{t2}&\cdots&X_{ts}
      &(Y_t)_{a_t,1}&(Y_t)_{a_t,2}&\cdots&(Y_t)_{a_t,r}
  \end{array}
  \right]
  \end{array}
  $}

  \vspace{.85em}

  More precisely, the entries of $Z$ are given by
  \begin{equation*}
    Z_{(i,a),j}=
    \begin{cases}
      X_{ij},&1\le j\le s,\\
      (Y_i)_{a,j-s},&s<j\le s+r,
    \end{cases}
    \qquad i\in[t],\quad a\in[a_i].
    \qquad\text{\rm(b)}
  \end{equation*}
  The replacement group $i$ contains $a_i$ copies of the original row
  $X_{i,*}$.  Its $a$th copy has the $a$th row of $Y_i$ beside it.

  \vspace{.55em}

  \begin{tikzpicture}[
    x=1cm,y=1cm,
    pt/.style={circle,fill=black,inner sep=0pt,minimum size=4pt},
    group/.style={draw=black,line width=.5pt,rounded corners=4pt},
    edge/.style={draw=black,line width=.55pt},
    lab/.style={font=\footnotesize}
  ]
    \node[lab,anchor=west] at (.25,2.12) {(c) two cases for the orthogonality in $Z$};

    \node[group,minimum width=4.25cm,minimum height=1.18cm]
      (same) at (3.00,.95) {};
    \node[pt] (ia) at (2.15,.95) {};
    \node[pt] (ib) at (3.85,.95) {};
    \draw[edge] (ia)-- node[lab,above] {an added column of $Y_i$} (ib);
    \node[lab,anchor=north] at (2.15,.82) {$(i,a)$};
    \node[lab,anchor=north] at (3.85,.82) {$(i,b)$};
    \node[lab] at (3.00,-.02) {same replacement group \(i\)};

    \node[group,minimum width=4.25cm,minimum height=1.18cm]
      (different) at (8.25,.95) {};
    \node[pt] (ic) at (7.40,.95) {};
    \node[pt] (kb) at (9.10,.95) {};
    \draw[edge] (ic)-- node[lab,above] {an original column of $X$} (kb);
    \node[lab,anchor=north] at (7.40,.82) {$(i,a)$};
    \node[lab,anchor=north] at (9.10,.82) {$(k,b)$};
    \node[lab] at (8.25,-.02) {different replacement groups: \(i\ne k\)};
  \end{tikzpicture}
  \endgroup
  \caption{The row replacement construction.  In the left part, the
  original row $X_{i,*}$ is repeated, and in the right part the rows of
  $Y_i$ are inserted.  Formula (b) gives all entries of the
  $(\sum_i a_i)\times(s+r)$ matrix.  The two cases used to prove
  orthogonality are shown in panel (c).}
  \label{fig:row-replacement-picture}
\end{figure}

\begin{lemma}
\label{lem:upper-bootstrap}
If
\begin{equation}
 [L,2^p-6]\subseteq\Theta_p,\qquad L\le2^p-9,
 \label{eq:bootstrap-hyp}
\end{equation}
then
\begin{equation}
 [2L,2^{p+1}-6]\subseteq\Theta_{p+1}.
 \label{eq:bootstrap-out}
\end{equation}
\end{lemma}

\begin{proof}
Three direct sums give the lower, middle, and upper parts of
\eqref{eq:bootstrap-out}.
Apply Lemma~\ref{lem:direct-sum} to \eqref{eq:bootstrap-hyp} and the
two sizes $2^p-4,2^p$ in \eqref{eq:finite-high-obstructions}:
\begin{align*}
[L,2^p-6]+[L,2^p-6]&=[2L,2^{p+1}-12],\\
[L,2^p-6]+(2^p-4)&=[L+2^p-4,2^{p+1}-10],\\
[L,2^p-6]+2^p&=[L+2^p,2^{p+1}-6].
\end{align*}
It remains to check that these intervals leave no gap.
The first two intervals overlap or are adjacent because
$L+2^p-4\le2^{p+1}-11$ follows from $L\le2^p-7$.
The last two overlap or are adjacent because
$L+2^p\le2^{p+1}-9$ follows from $L\le2^p-9$.
Both bounds on the lower endpoint follow from \eqref{eq:bootstrap-hyp}.
Their union is \eqref{eq:bootstrap-out}.
\end{proof}

The following induction proves \eqref{eq:intro-two-n-tail} in
Theorem~\ref{thm:uniform-tail}.
\begin{proof}
The induction uses two intervals in each dimension.
Lemma~\ref{lem:upper-bootstrap} extends the interval in the preceding
dimension to the upper end.
The row-decomposition intervals in \eqref{eq:decomposition-block-interval}
fill the lower end.  Their overlap will give \eqref{eq:intro-two-n-tail}.

For the initial dimensions, we show that
$[L_n,2^n-6]\subseteq\Theta_n$ for $4\le n\le10$,
with $L_n\le2n$.  For $4\le n\le8$, \eqref{eq:small-spectra} gives
the endpoints in the following table.  For $n=9$, apply
Lemma~\ref{lem:upper-bootstrap} to $[15,250]\subseteq\Theta_8$.
It gives $[30,506]\subseteq\Theta_9$.
The interval $[16,30]$ from \eqref{eq:decomposition-block-interval}
therefore gives $L_9=16$.
Another application of Lemma~\ref{lem:upper-bootstrap} gives
$[32,1018]\subseteq\Theta_{10}$.
The interval $[16,40]$ from \eqref{eq:decomposition-block-interval}
then gives $L_{10}=16$.  Thus the initial values are
\begin{equation}
\begin{array}{c|rrrrrrr}
n&4&5&6&7&8&9&10\\ \hline
L_n&8&8&12&10&15&16&16.
\end{array}
\label{eq:tail-bases}
\end{equation}

For the induction step, it is enough that the lower interval meets the
upper one.  We show that $4n-4$ belongs to both.
Let $n\ge11$.  Suppose that the interval in dimension $n-1$ starts at
$L\le2(n-1)$.  Then $L\le2^{n-1}-9$.
Lemma~\ref{lem:upper-bootstrap} applies, so
\eqref{eq:bootstrap-out} gives an upper interval whose lower endpoint is
at most $4n-4$.
The lower intervals in \eqref{eq:decomposition-block-interval} end at
$5n-10$ for even $n$ and $5n-15$ for odd $n$.
Both endpoints are at least $4n-4$.
Their lower endpoints are at most $2n$.
Thus the lower and upper intervals give \eqref{eq:intro-two-n-tail},
completing the induction.
\end{proof}
\section{Square kernel lemmas and an extension for odd dimensions}
\label{app:square-kernel}

This appendix proves the graph factorization used in the odd extension.
It then verifies the initial matrix and its extension to odd dimensions.
The last subsection states the finite check for the kernel with five points.
Square kernels are defined in Definition~\ref{def:square-kernel}.
The graph terms are defined in Definition~\ref{def:graph-notation}(i)--(iii).
The extension
also uses row decompositions and the block cover condition from
Definitions~\ref{def:row-decomposition} and
\ref{def:block-cover-condition}.

\subsection{Proof of Lemma~\ref{lem:split-factorization}}

\begin{proof}
We first partition $E(H)$ into matchings of equal size.
Each matching misses $r$ vertices of $B$.
A second factorization assigns edges from $S$ to these missing vertices.
This completes every matching to a one-factor of
\eqref{eq:split-factorization}.

To obtain matchings of equal size,
Theorem~\ref{thm:standard-graph-tools}(v) gives at most
$b-r+1\le b$ colours.
Lemma~\ref{lem:balanced-color-classes} gives $b$ classes satisfying
\eqref{eq:balanced-color-classes}.
Regularity gives
\begin{equation}
 |E(H)|=\frac{b(b-r)}2.
 \label{eq:split-edge-count}
\end{equation}
By \eqref{eq:split-edge-count}, the average class size is the integer
$(b-r)/2$.  We claim that every class has this size.
Indeed, by \eqref{eq:balanced-color-classes}, the sizes differ by
at most one.  Since their average is an integer, they must all have that
size.  Each class is a matching on $b-r$ vertices.
It misses $r$ vertices.

Each class must now be completed at its missing vertices.
The completions must use every edge of $K_{S,B}$ once.
Form a bipartite graph $J$ on the $b$ colours and the $b$ vertices of $B$.
Join $c$ to $v$ if colour $c$ misses $v$.
This graph is regular, as required by
Theorem~\ref{thm:standard-graph-tools}(iii).
Indeed, each colour vertex has degree $r$.
Each $v\in B$ also has degree $r$, since its $b-r$ edges in $H$ have
distinct colours.

Theorem~\ref{thm:standard-graph-tools}(iii) factors $J$ into $r$
one-factors.  Label them by $s\in S$.
If the factor labelled $s$ contains $cv$, give $sv$ colour $c$.
For each $s$, this assigns every edge from $s$ to $B$ once.
For each $c$, it matches $S$ to the vertices missed by colour $c$.
Each completed colour class is therefore a one-factor of
\eqref{eq:split-factorization}.
\end{proof}

\subsection{The initial \texorpdfstring{$8\times5$}{8 by 5} matrix}

The extension starts from the input UOM $Y_5$ in \eqref{eq:Y5}, with the
row decomposition in \eqref{eq:Y5-decomposition-blocks}.  Its integer
labels are local to each column, as explained before \eqref{eq:Y5}.

\begin{lemma}
\label{lem:Y5-profile}
The matrix $Y_5$ is a UOM.  Moreover, its row decomposition in
\eqref{eq:Y5-decomposition-blocks} satisfies the block cover condition.
\end{lemma}

\Needspace{10\baselineskip}
\begin{proof}
For the block cover condition \eqref{eq:block-cover-condition}, the
count depends on how many chosen blocks have two rows.
In each case, the number of columns must be at least the number of chosen
blocks.  Taking all six blocks then proves unextendibility.

For orthogonality, the following table gives the unique witness column for every row pair
of \eqref{eq:Y5}:
\begin{equation}
\begin{array}{c|rrrrrrr}
 &2&3&4&5&6&7&8\\ \hline
1&2&1&1&4&5&3&2\\
2&&1&1&5&4&2&3\\
3&&&2&3&2&4&5\\
4&&&&2&3&5&4\\
5&&&&&2&1&1\\
6&&&&&&1&1\\
7&&&&&&&2
\end{array}
\label{eq:Y5-witnesses}
\end{equation}
In the listed column, the two entries are mates.
Thus \eqref{eq:Y5-witnesses} proves \eqref{eq:row-orthogonal} for every pair.

It remains to prove \eqref{eq:block-cover-condition}.
Only $A_2,A_4$ in \eqref{eq:Y5-decomposition-blocks} have two rows.
In the first two columns, the fibres with more than one row are
\begin{align}
\F_1&=\{\{1,2\},\{3,4\},\{5,6\},\{7,8\}\},
\label{eq:Y5-f1}\\
\F_2&=\{\{1,7\},\{2,8\},\{3,5\},\{4,6\}\};
\label{eq:Y5-f2}
\end{align}
every fibre in columns three, four, and five contains only one row.

Choose any $k$ decomposition blocks.  Let $d$ denote the number of
chosen blocks which contain two rows.  The union of the chosen blocks
contains $k+d$ target rows.  The three possibilities are $d=0,1,2$.

If $d=0$, each fibre covers at most one target row.
The chosen rows lie in $\{1,3,6,8\}$.
The lists \eqref{eq:Y5-f1}--\eqref{eq:Y5-f2} show that every fibre meets
this set in at most one row.  Thus fewer than $k$ columns cannot cover the
chosen rows.

If $d=2$, fewer than $k$ columns cannot cover the $k+2$ target rows.
Only the first and second columns can cover a second row.
Thus any $h$ columns with $0\le h\le k-1$ cover at most
$h+2\le k+1<k+2$ rows.

If $d=1$, a cover with at most $k-2$ columns covers at most $k$ rows.
To cover $k+1$ rows with $k-1$ columns, the fibres chosen from the
first two columns would each have to cover two target rows without overlap.
We claim that these two fibres must overlap.
Indeed, if the chosen two-row block is $\{2,7\}$, the pair of fibres belongs to
\begin{equation}
 \{\{1,2\},\{7,8\}\}
 \times\{\{1,7\},\{2,8\}\}.
 \label{eq:Y5-overlap-27}
\end{equation}
If the chosen two-row block is $\{4,5\}$, the pair of fibres belongs to
\begin{equation}
 \{\{3,4\},\{5,6\}\}
 \times\{\{3,5\},\{4,6\}\}.
 \label{eq:Y5-overlap-45}
\end{equation}
Each pair in \eqref{eq:Y5-overlap-27} or \eqref{eq:Y5-overlap-45}
has a common row.  The first two fibres therefore cover at most three target
rows in total.  Each remaining fibre covers at most one target row.
Hence $k-1$ columns cover at most $k$ rows, a contradiction.

All cases give \eqref{eq:block-cover-condition}.
In particular, the six blocks cannot be covered by five columns.
Equation~\eqref{eq:fibre-cover} therefore makes $Y_5$ a UOM.
\end{proof}

\subsection{Proof of Theorem~\ref{thm:odd-padding}: extension for odd
\texorpdfstring{$q$}{q}}

\begin{proof}
The first construction extends $Y_5$ to size $(q+3)\times q$.
The second gives $(q+1)\times q$ using singleton blocks.
In both constructions, graph edges specify the pairs of rows that need
mate symbols.  A factorization assigns these edges to columns.
A fibre cover of $h$ blocks will need at least $h$ columns, as required by
\eqref{eq:block-cover-condition}.
Taking all blocks then gives unextendibility.

\smallskip
\noindent\emph{Extension of $Y_5$.}
The case $q=5$ is Lemma~\ref{lem:Y5-profile}.
Now let $q\ge11$ be odd.  Set $t=q-5$.
Then $t$ is even with $t\ge6$.  Add $t$ new rows to $Y_5$.
Each new row forms one decomposition block.
The new matrix has $6+t=q+1$ blocks in total.

First assign mate symbols to some pairs of new rows in the original
columns.  Theorem~\ref{thm:standard-graph-tools}(i) gives
$K_t$ a factorization with $t-1$
one-factors.  Choose five and denote them by
$D_1,\ldots,D_5$.  For every edge of $D_j$, choose a new mate pair.
Put its two symbols on the endpoint rows in original column $j$.
The pairs of new rows still requiring mate symbols form
\begin{equation}
 H=K_t-\disj_{j=1}^5D_j.
 \label{eq:padding-H}
\end{equation}
The graph $H$ is $(t-6)$-regular.  It remains to cover these edges and
all pairs consisting of an original block and a new row.
Lemma~\ref{lem:split-factorization}
splits $K_{6,t}\cup H$ into $t$ one-factors.
Make one new column for each factor.  Use a new mate pair for every edge.
If an edge joins an
original decomposition block $A$ to a new row $b$, put one symbol on all
rows of $A$ and its mate on $b$.  For an edge joining two new rows, put the
two mate symbols on those rows.  The factors now cover every pair for which
at least one row is new.  The old row pairs remain orthogonal by
Lemma~\ref{lem:Y5-profile}.  Thus the new matrix is orthogonal.

The block cover condition means that any $h$ blocks of the new
decomposition need at least $h$ columns.
Every fibre in a new column is exactly one decomposition block.
Suppose that a fibre cover for the union of $h$ blocks uses $a$ new
columns.  These columns cover at most $a$ complete blocks.
Every other chosen block must be covered entirely by the original columns.
Let $u$ and $v$ be the numbers of original and new blocks still to be
covered, respectively.  Then $u+v\ge h-a$.

We claim that the remaining $u+v$ blocks need at least $u+v$ original
columns.  Indeed, Lemma~\ref{lem:Y5-profile} requires $u$ columns for the original
blocks.  A fibre containing a new row is a singleton with a new symbol.
It covers no original row.  A cover uses at most one fibre per column
by Theorem~\ref{thm:fibre-cover}.
Thus the $v$ new blocks require $v$ further columns, separate from those
used for the original blocks.
Thus every cover uses at least $a+u+v\ge h$ columns, proving
\eqref{eq:block-cover-condition}.
Taking all $q+1$ blocks in \eqref{eq:fibre-cover} proves unextendibility.
The size of the matrix is
\begin{equation}
 (8+t)\times(5+t)=(q+3)\times q.
\end{equation}

\smallskip
\noindent\emph{Singleton construction.}
For the input UOM with singleton decomposition blocks, split $K_{q+1}$ into
$q$ one-factors by Theorem~\ref{thm:standard-graph-tools}(i).
Use one factor in every column.  Put a new mate pair on the endpoint rows
of every edge.  Every row pair then has mate symbols in one column.
Each fibre contains one row.  Any $h\le q$ rows need $h$
columns.  All $q+1$ rows cannot be covered by only $q$ columns.  By
Theorem~\ref{thm:fibre-cover}, the resulting matrix is a
$(q+1)\times q$ UOM.  The same count proves the block cover condition
for its row decomposition.
\end{proof}

\subsection{The kernel with five points}

The program \nolinkurl{check_five_point_cliques.py} verifies
\eqref{eq:cert-five-required} using the clique and maximality tests in
\eqref{eq:cert-five-clique} and \eqref{eq:cert-five-maximal}.
The first test checks every pair in each supplied set.
The second checks that no vertex outside the set can be added.
Thus, for each $c\in\Lambda_0$, the supplied $c$-vertex set is an
inclusion-maximal clique in the full graph on $120$ vertices.
For each witness, Lemma~\ref{lem:full-link}(i) gives the corresponding
UOM used in the proof.
\section{A lift preserving good column structure and the
\texorpdfstring{$13\times8$}{13 by 8} initial matrix}
\label{app:good-column-lift}

The construction starts from an explicit $13\times8$ matrix.
It combines two one-factors into one column.
The two cover bounds in \eqref{eq:good-column-conditions} remain valid
after the lift.  This gives \eqref{eq:n-plus-five-endpoint}.

\subsection{Restricted cover costs and good column structure}

\begin{definition}
\label{def:restricted-cover-cost}
For a $q$-column matrix $X$, $J\subseteq[q]$, and $T\subseteq R(X)$,
define
\begin{equation}
 \begin{aligned}
 \kappa_X^J(T)
 &=\min\bigl\{|D|:\ D\subseteq J,\\
 &\qquad T\subseteq\displaystyle\bigcup_{j\in D}F_{X,j}(a_j),\\[-1mm]
 &\qquad F_{X,j}(a_j)\in\F_{X,j}\quad(j\in D)\bigr\},\\
 \min\varnothing&=\infty.
 \end{aligned}
 \label{eq:restricted-cover-cost}
\end{equation}
\qed
\end{definition}

When all columns are allowed, $\kappa_X^{[q]}=\kappa_X$, with $\kappa_X$
as in Definition~\ref{def:equality-fibre}(iii) and \eqref{eq:cover-cost}.
The empty set has cover cost
$\kappa_X^J(\varnothing)=0$.

\Needspace{11\baselineskip}
\begin{definition}
\label{def:good-column-structure}
Let $q$ be even, and let
\begin{equation}
 R(X)=A_1\disj\cdots\disj A_{q+2}
 \label{eq:good-column-decomposition}
\end{equation}
be a row decomposition into nonempty blocks.  Choose $g\in[q]$.  The pair
$(A_1,\ldots,A_{q+2};g)$ is a good column structure of $X$ if, for every
nonempty $I\subseteq[q+2]$,
\begin{equation}
 \kappa_X^{[q]}\!\left(\bigcup_{i\in I}A_i\right)\ge |I|-1,
 \qquad
 \kappa_X^{[q]\setminus\{g\}}\!\left(\bigcup_{i\in I}A_i\right)
 \ge |I|.
 \label{eq:good-column-conditions}
\end{equation}
\qed
\end{definition}

An orthogonal matrix with a good column structure is a UOM.
Indeed, take $I=[q+2]$ in the first inequality of
\eqref{eq:good-column-conditions}.  A cover of all rows would need at
least $q+1$ columns.  The matrix $X$ has only $q$ columns, so
Theorem~\ref{thm:fibre-cover}(ii) applies.

\subsection{A compressible pair of one-factors}

\begin{lemma}
\label{lem:compressible-factor-pair}
If $t\equiv0\pmod4$, then $K_t$ has a one-factorization which contains two
one-factors $M_0,M_1$ satisfying
\begin{equation}
 M_0\cup M_1\cong \frac t4 K_{2,2}.
 \label{eq:compressible-factor-pair}
\end{equation}
\end{lemma}

\begin{proof}
Write $B=U\disj V$, where $|U|=|V|=t/2$.
Divide the vertices into pairs as follows:
\begin{equation*}
 U=\mathop{\dot\bigcup}_{i=1}^{t/4}\{u_i,u_i'\},
 \qquad
 V=\mathop{\dot\bigcup}_{i=1}^{t/4}\{v_i,v_i'\}.
\end{equation*}
Set
\begin{align*}
 M_0&=\mathop{\dot\bigcup}_{i=1}^{t/4}\{u_iv_i,u_i'v_i'\},&
 M_1&=\mathop{\dot\bigcup}_{i=1}^{t/4}\{u_iv_i',u_i'v_i\}.
\end{align*}
Each group of four vertices gives one $K_{2,2}$.
This proves \eqref{eq:compressible-factor-pair}.
We claim that $M_0,M_1$ belong to a one-factorization of $K_t$.
Indeed, $K_{U,V}-M_0-M_1$ is bipartite and $(t/2-2)$-regular.
Theorem~\ref{thm:standard-graph-tools}(iii) factors it into $t/2-2$
one-factors.
Since $t/2$ is even, Theorem~\ref{thm:standard-graph-tools}(i) gives
factorizations of $K_U$ and $K_V$.
Pairing factors with the same index gives $t/2-1$ further one-factors
on $B$.
Together with $M_0,M_1$, these partition the edges within and between
$U,V$.  Their number is
\begin{equation*}
 2+(t/2-2)+(t/2-1)=t-1.
\end{equation*}
\end{proof}

\subsection{A lift preserving good column structure}

\begin{theorem}
\label{thm:good-column-even-lift}
Let $X$ be an orthogonal $m\times q$ formal matrix with $q$ even, and put
$r=q+2$.  Suppose that $(A_1,\ldots,A_r;g)$ is a good column structure.
If
\begin{equation}
 t\equiv0\pmod4,\qquad t\ge r,
 \label{eq:good-column-lift-hypothesis}
\end{equation}
then there is an orthogonal $(m+t)\times(q+t)$ matrix $\widehat X$ such
that the blocks
\begin{equation}
 A_1,\ldots,A_r,\{b_1\},\ldots,\{b_t\}
 \label{eq:lifted-decomposition-blocks}
\end{equation}
form a row decomposition and, together with $g$, give $\widehat X$ a good
column structure.  In particular, $\widehat X$ is a UOM.
\end{theorem}

\begin{proof}
To preserve \eqref{eq:good-column-conditions}, we control how many
blocks one fibre can cover.
Each fibre in a new column will be one whole block.
Among original columns, only $g$ can cover two new rows at once.
For any $h$ chosen blocks, the required lower bound is $h$ columns
without $g$.  If $g$ is allowed, it is $h-1$.

\noindent\emph{Graph data.}
Let $B=\{b_1,\ldots,b_t\}$ be the set of new rows, each forming one
decomposition block.  Regard the original blocks
$S=\{A_1,\ldots,A_r\}$ as graph vertices.  Choose the factorization in
Lemma~\ref{lem:compressible-factor-pair}.
Put $Q=M_0\cup M_1$.  Since $t-1\ge q+1$, the factorization also
contains the following distinct factors besides $M_0$ and $M_1$:
\begin{equation}
 N_j\qquad(j\in[q]\setminus\{g\})
 \label{eq:ordinary-removed-factors}
\end{equation}
On the vertex set $B$, define the graph $H$ by
\begin{equation}
 E(H)=E(K_B)\setminus
 \left(E(Q)\cup\bigcup_{j\ne g}E(N_j)\right).
 \label{eq:chosen-column-residual-graph}
\end{equation}
Deleting these $q+1$ edge-disjoint one-factors leaves
\begin{equation}
 \deg_H(b)=t-1-(q+1)=t-q-2=t-r\qquad(b\in B).
 \label{eq:chosen-column-residual-degree}
\end{equation}
Both $|S|=r$ and $|B|=t$ are even.
Equation~\eqref{eq:chosen-column-residual-degree} gives the required degree.
Lemma~\ref{lem:split-factorization} therefore gives
\begin{equation}
 K_{S,B}\cup H=P_1\disj\cdots\disj P_t,
 \label{eq:chosen-column-split-factors}
\end{equation}
where every $P_\ell$ is a one-factor on $S\disj B$.

\smallskip
\noindent\emph{The lifted matrix.}
Define the matrix $\widehat X$ as follows.  On the original rows and
columns, it agrees with $X$.  Each new mate pair uses symbols that have
not yet appeared in that column.

For $j\ne g$, put the matching $N_j$ in original column $j$.
Assign a different new mate pair to each edge.
Put its two symbols on the two endpoint rows.
Every fibre on the new rows then contains one row.
In the chosen column $g$, use one new mate pair for every
$K_{2,2}$ component of $Q$.  Put one symbol on the two vertices in one
part and put its mate on the two vertices in the other part.  Thus, every
fibre on the new rows in column $g$ contains two rows.

Put $P_\ell$ in new column $q+\ell$ for every $\ell\in[t]$.  If an
edge has the form $A_i b$, put one symbol on every row of $A_i$ and put
its mate on $b$.  If an edge has the form $bb'$, put the two symbols of a
new mate pair on $b$ and $b'$.  Different factor edges use different mate
pairs.  It follows that
\begin{equation}
 \begin{gathered}
 \text{every fibre in a new column}\\
 \text{is exactly one block in \eqref{eq:lifted-decomposition-blocks}.}
 \end{gathered}
 \label{eq:new-column-one-group-fibre}
\end{equation}

\smallskip
\noindent\emph{Orthogonality.}
Original row pairs remain orthogonal in $X$.
An original block and a new row are joined in
\eqref{eq:chosen-column-split-factors}, which supplies a witness column.
For two new rows, \eqref{eq:chosen-column-residual-graph} assigns their
edge to $Q$, some $N_j$, or $H$.
The pair is orthogonal in column $g$, column $j$, or a new column,
respectively.

\smallskip
\noindent\emph{Cover bounds.}
To prove \eqref{eq:good-column-conditions}, choose a nonempty set of $h$
blocks from \eqref{eq:lifted-decomposition-blocks}.
Fix a fibre cover of their union.
Let $a$ be the number of new columns used.  By
\eqref{eq:new-column-one-group-fibre}, these new columns cover at most $a$
complete blocks.  A chosen block not covered by these fibres must be
covered entirely by the original columns.
Let $u$ and $v$ be the numbers of original and new blocks remaining,
respectively.  Then
\begin{equation}
 u+v\ge h-a.
 \label{eq:uncovered-decomposition-blocks}
\end{equation}

We first show that a cover without column $g$ uses at least $h$ columns.
Indeed, the second bound in \eqref{eq:good-column-conditions} requires
$u$ original columns for the $u$ original blocks.
A fibre on $B$ in any original column $j\ne g$ is a singleton with no
original row.  The $v$ new blocks therefore need $v$ further columns.
Together with \eqref{eq:uncovered-decomposition-blocks}, this gives
\begin{equation}
 a+u+v\ge h.
 \label{eq:lifted-without-chosen-column-bound}
\end{equation}

We next show that allowing column $g$ still requires at least $h-1$
columns.
Let $c$ be the number of original columns in the cover.
If $g$ is not used, then $c\ge u+v$ by the preceding count.
Suppose that the chosen fibre in column $g$ lies on the original rows.
This fibre does not meet $B$.
The first inequality in \eqref{eq:good-column-conditions} gives the
bound for the $u$ remaining original blocks.
Each of the $v$ remaining new blocks needs a separate column.
Hence
\begin{equation*}
 c\ge
 \begin{cases}
 (u-1)+v,&u>0,\\
 v,&u=0,
 \end{cases}
\end{equation*}
so $c\ge u+v-1$.
The remaining case is that the chosen fibre in column $g$ lies on $B$.
It covers at most two new blocks.
The $u$ remaining original blocks must then be covered without using $g$.
This gives
\begin{equation}
 c\ge1+u+\max\{0,v-2\}\ge u+v-1.
 \label{eq:chosen-column-fibre-on-new-rows}
\end{equation}
In every case, $c\ge u+v-1$, so
\begin{equation}
 a+c\ge a+u+v-1\ge h-1.
 \label{eq:lifted-all-column-bound}
\end{equation}
Equations~\eqref{eq:lifted-without-chosen-column-bound} and
\eqref{eq:lifted-all-column-bound} give the two bounds in
\eqref{eq:good-column-conditions}.
For all $(q+t)+2$ blocks, \eqref{eq:lifted-all-column-bound} requires
at least $q+t+1$ columns.
Equation~\eqref{eq:fibre-cover} therefore shows that $\widehat X$ is a UOM.
\end{proof}

\subsection{The verified \texorpdfstring{$13\times8$}{13 by 8} initial matrix}

The lift starts from the matrix below.  Its integers are local symbols, with
mate pairs
\begin{equation}
 1\leftrightarrow2,\quad 3\leftrightarrow4,\quad
 5\leftrightarrow6,\quad 7\leftrightarrow8,\quad
 9\leftrightarrow10.
 \label{eq:x8-mate-map}
\end{equation}

\begin{equation}
X_8=
\left(\begin{array}{rrrrrrrr}
 9&1&3&2&5&8&1&3\\
 3&4&2&9&7&8&1&4\\
 4&5&9&3&6&4&3&1\\
 7&4&6&6&2&7&7&2\\
 3&1&2&9&7&7&1&1\\
 4&3&9&3&6&4&4&5\\
 5&2&7&4&9&9&8&3\\
 8&8&10&8&10&2&2&5\\
 10&6&8&10&1&6&5&6\\
 1&5&5&1&8&10&4&6\\
 6&3&4&7&4&5&3&2\\
 2&7&1&5&3&3&6&4\\
 3&1&2&9&7&1&2&1
\end{array}\right)
\label{eq:good-column-matrix-x8}
\end{equation}
For this matrix, take the following ten decomposition blocks:
\begin{equation}
\begin{aligned}
 A_1&=\{1,13\},&A_2&=\{2,5\},&A_3&=\{3,6\},\\
 A_4&=\{4\},&A_5&=\{7\},&A_6&=\{8\},\\
 A_7&=\{9\},&A_8&=\{10\},&A_9&=\{11\},\\
 A_{10}&=\{12\}
\end{aligned}
\label{eq:good-column-decomposition-x8}
\end{equation}

\begin{theorem}
\label{thm:good-column-thirteen-by-eight-matrix}
The matrix $X_8$ in \eqref{eq:good-column-matrix-x8} is a UOM.  With
$g=8$, the blocks in \eqref{eq:good-column-decomposition-x8} give it a
good column structure.
\end{theorem}

\begin{proof}
The program \nolinkurl{check_construction_data.py} checks the $78$ row
pairs of $X_8$ by \eqref{eq:cert-orth}.
It also verifies \eqref{eq:cert-uom-dp}.

For the good column structure, it uses exactly the ten blocks in
\eqref{eq:good-column-decomposition-x8} and $g=8$.
The table in \eqref{eq:cert-good-column-minima} groups the $1023$
nonempty block sets by their number of blocks.
For each group, it gives the smallest cover cost with each allowed set
of columns.  Each minimum satisfies the corresponding bound in
\eqref{eq:cert-good-column-acceptance}.
Thus both bounds in \eqref{eq:good-column-conditions} hold for every
nonempty block set.

The chosen column has orthogonality graph
\begin{equation}
 K_{\{3,5,13\},\{4,11\}}
 \disj K_{\{1,7\},\{2,12\}}
 \disj K_{\{6,8\},\{9,10\}}.
 \label{eq:x8-chosen-column-graph}
\end{equation}
No other row decomposition is used in this proof.
\end{proof}

\begin{corollary}
\label{cor:good-column-zero-class-family}
For every $n\ge20$ satisfying $n\equiv0\pmod4$, there exists an
$(n+5)\times n$ UOM.  Equivalently,
\begin{equation}
 n+5\in\Theta_n.
 \label{eq:good-column-zero-class-family}
\end{equation}
\end{corollary}

\begin{proof}
To obtain \eqref{eq:good-column-zero-class-family}, apply the lift to
the $13\times8$ UOM in
Theorem~\ref{thm:good-column-thirteen-by-eight-matrix}.
Set $t=n-8$.  Then $t\equiv0\pmod4$ and $t\ge12>8+2$, which verifies
\eqref{eq:good-column-lift-hypothesis}.
Theorem~\ref{thm:good-column-even-lift} gives $13+t=n+5$ rows
in $8+t=n$ columns.
\end{proof}
\section{A global obstruction in even dimensions}
\label{even:sec-global-obstruction}

Recall \eqref{eq:intro-theta} and the convention $2^n\in\Theta_n$.  This
appendix excludes the integer immediately above the minimum when
$n\equiv2\pmod4$.  The argument applies to every UOM.  It does not use the
constructions in Sections~\ref{sec:master-constructions} and
\ref{sec:existence-synthesis}.

\subsection{Fibres and the covering inequality}
\label{even:subsec-fibres}

Let $X$ be an $m\times n$ formal matrix.
For a symbol $x$ in column $j$, use the fibre $F_{X,j}(x)$ from
Definition~\ref{def:equality-fibre}(i).  A symbol not occurring in $X$ has
$F_{X,j}(x)=\varnothing$.  The set $F_{X,j}(\bar x)$ is the mate fibre.  The
fibre cover test in Theorem~\ref{thm:fibre-cover} permits at most one fibre
from each column.  An entry $\bar x$ in column $j$ of an extension row
covers exactly $F_{X,j}(x)$.

\begin{lemma}
\label{even:lem-union-inequality}
Let $X$ be an $(n+3)\times n$ UOM.  If
$F_1,\ldots,F_i$ are fibres selected from $i$ distinct
columns, then
\begin{equation}
 \left|F_1\cup\cdots\cup F_i\right|\leq i+2.
 \label{even:eq-union-inequality}
\end{equation}
\end{lemma}

\begin{proof}
A violation of \eqref{even:eq-union-inequality} would give a full fibre
cover.  Indeed, if the selected union contains at least $i+3$ rows,
at most $(n+3)-(i+3)=n-i$ rows remain.
Assign them to distinct unused columns.
In each assigned column, choose the fibre containing that row.
Together with $F_1,\ldots,F_i$, these fibres satisfy \eqref{eq:fibre-cover},
contrary to unextendibility.
\end{proof}

\subsection{Two finite and elementary matching lemmas}
\label{even:subsec-local-lemmas}

The next two lemmas use matching and transversal matching in the sense of
Definition~\ref{def:graph-notation}(i) and
Definition~\ref{def:graph-notation}(iv) in
Section~\ref{sec:preliminaries}.

\begin{lemma}
\label{even:lem-complete-transversal-certificate}
Let $\mathcal M=(M_1,\ldots,M_s)$ be a list of matchings of size
three.  The same $M_i$ may occur more than once.
\begin{enumerate}[label=(\roman*)]
\item If $\mathcal M$ has no
transversal matching of size three, then
\begin{equation}
 s\leq4.
 \label{even:eq-complete-transversal-bound}
\end{equation}

\item If, in addition, $s=4$ and every edge in every $M_i$ meets the same
three-element set $A$, then the four matchings can be reordered so that
\begin{equation}
 M_1=M_2=M,\qquad M_3=M_4=N,
 \label{even:eq-complete-transversal-extremal}
\end{equation}
where $M\cup N$ is a six-cycle between $A$ and another
three-element set $D$.
\end{enumerate}
\end{lemma}

\begin{proof}
For (i), it is enough to exclude lists of length five.
If a longer list has no transversal matching of size three, the same
holds for its first five members.
The program \nolinkurl{check_finite_obstructions.py} uses
\eqref{eq:cert-transversal-predicate} in the recurrence
\eqref{eq:cert-rgs-recurrence}.
The last count in \eqref{eq:cert-rgs-counts} is zero.
Thus no such list of length five exists, which gives
\eqref{even:eq-complete-transversal-bound}.

For (ii), each matching defines a one-to-one map from $A$ to vertices
outside $A$.  Indeed, its three disjoint edges meet the three-element set $A$.
Each edge must therefore use exactly one vertex there.
The fixed-set search therefore applies.
Equation~\eqref{eq:cert-fixed-set-counts} lists its counts.
Every final list has the form \eqref{eq:cert-fixed-set-shape}, which is
\eqref{even:eq-complete-transversal-extremal}.

Every list with the required properties is included in these searches
up to vertex relabelling.
Indeed, old vertices keep their labels at each step.
New vertices can be given consecutive labels without changing any intersections.
The searches test every matching allowed by these conditions, up to
vertex relabelling.
Thus every list has a representative in the recurrence.
\end{proof}

\begin{lemma}
\label{even:lem-two-intersecting-triples}
Let $\mathcal T$ be a family of at least three distinct
three-element sets such that
\[
 |T\cap T'|\geq2
 \qquad
 (T,T'\in\mathcal T).
\]
\begin{enumerate}[label=(\roman*)]
\item Exactly one of the following cases holds:
\begin{enumerate}[label=(\alph*)]
\item all members of $\mathcal T$ contain a fixed two-element set;
\item all members of $\mathcal T$ are contained in a fixed four-element
set.
\end{enumerate}

\item Suppose also that every $T_i\in\mathcal T$ has a two-element set $E_i$
which is disjoint from $T_i$ and meets every other set in $\mathcal T$.
Then
\begin{equation}
 |\mathcal T|\leq4.
 \label{even:eq-triple-family-bound}
\end{equation}

\item No three pairwise disjoint two-element sets can each meet
every member of $\mathcal T$.
\end{enumerate}
\end{lemma}

\begin{proof}
Part (i) gives two possible forms of the triple family.
For each form, parts (ii) and (iii) count the vertices that the
two-element sets must use.

For (i), choose two distinct triples.
We claim that if some triple does not contain their common pair,
then all triples lie in a four-element set.  Indeed, label the chosen triples as
\[
 T_1=\{1,2,3\},
 \qquad
 T_2=\{1,2,4\}.
\]
If every other set contains $\{1,2\}$, then the first case holds.
Otherwise, some triple omits at least one of $1,2$.
By the two intersection conditions, this triple must be either
$\{1,3,4\}$ or $\{2,3,4\}$.
Any triple containing a vertex outside $\{1,2,3,4\}$ must also contain
$\{1,2\}$ to meet both $T_1$ and $T_2$ in two vertices.
It would then meet this third triple in only one vertex.
Thus every triple is contained in $\{1,2,3,4\}$, giving (i)(b).
The two cases cannot both hold.  Three distinct triples containing a
fixed pair have a union of at least five vertices.

For (ii), case (i)(a) permits at most three triples under the additional
condition on $E_i$.  Indeed, write the triples as
\[
 C\cup\{a_i\},
 \qquad |C|=2.
\]
The set $E_i$ is disjoint from $C\cup\{a_i\}$.
It must meet every other triple, so it contains every $a_j$ with $j\neq i$.  Since
$|E_i|=2$, there can be at most three triples.  For the second case, a
four-element set has only four subsets of size three.  This proves
\eqref{even:eq-triple-family-bound}.

For (iii), first consider (i)(a).  A two-element set avoiding $C$ would
have to contain all the distinct $a_i$ to meet every triple.
There are at least three of them.  Thus every two-element set meeting all
the triples must meet $C$.  Three disjoint sets would require three
vertices of $C$, contrary to $|C|=2$.

In case (i)(b), three disjoint sets meeting every triple would use more
than four vertices of the containing set.  The count depends on whether
there are three or four triples.
If there are three triples, a set using only one vertex there must use
their common vertex.  At most one of the disjoint sets can do so.
The other two each use at least two vertices, requiring at least five
in total.  If there are four triples, each set needs two vertices there,
requiring six.  Both counts are impossible.
\end{proof}

\Needspace{13\baselineskip}
\subsection{Proof of Theorem~\ref{thm:even-obstruction}}
\label{even:subsec-global-theorem}

\begin{proof}
The aim is to show that the columns cover too few distinct row pairs.
The fibre inequality \eqref{even:eq-union-inequality} restricts the
column types that need to be considered.
There are three cases: no triple fibre; two triples in one column;
or at least one triple fibre, with at most one in each column.
In each case, either the capacity is below the required count in
\eqref{even:eq-required-pairs}, or repeated pairs leave too few distinct
pairs.

Suppose that $X$ is a $(4k+5)\times(4k+2)$ UOM.
The hypothesis $n\ge14$ gives $k\ge3$.
Set
\[
 m=4k+5,\qquad n=4k+2.
\]
Orthogonality requires every row pair to be covered.

\smallskip
\noindent\emph{Column facts.}
For column $j$, let $C_{r,j}$ denote the number of fibres having
size $r$.  Let $e_j$ denote the number of unordered row pairs which
are orthogonal in this column.  Each covered pair is counted once in this
column.

We first show that every nonempty fibre has a nonempty mate fibre.
Indeed, suppose that $F_{X,j}(x)$ is nonempty but $F_{X,j}(\bar x)$ is empty.
Choose a row $v$ with entry $x$ in column $j$.  Replace that entry by
$\bar x$.
The new row is orthogonal to $v$ in column $j$.
Every other row was orthogonal to $v$ in another column, since $\bar x$
did not occur in column $j$.
The entries in those columns are unchanged.  Thus the new row extends $X$.
This contradiction proves mate completeness.

Each column has a positive odd number of doubletons.
Indeed, \eqref{even:eq-union-inequality} with $i=1$ shows that every fibre
has at most three rows.  Mate completeness makes the number of nonempty fibres even.
The row count is
\[
 4k+5=C_{1,j}+2C_{2,j}+3C_{3,j}.
\]
Subtracting the even fibre count shows that $C_{2,j}+2C_{3,j}$ is odd.
Thus $C_{2,j}$ is odd and nonzero, proving
\begin{equation}
 C_{2,j}\equiv1\pmod2,
 \qquad
 2\leq\max_x|F_{X,j}(x)|\leq3.
 \label{even:eq-small-fibre-parity}
\end{equation}

By Definition~\ref{def:equality-fibre}(i), the number of row pairs covered
by column $j$ is
\begin{equation}
 e_j=\sum_{\{x,\bar x\}}|F_{X,j}(x)|\,|F_{X,j}(\bar x)|.
 \label{even:eq-column-edge-count}
\end{equation}
In this sum, every unordered mate pair is used once.
For fixed fibre sizes, a maximum is obtained by pairing the largest sizes
together.  Indeed, if $a\ge b\ge c\ge d$, then
$ab+cd\ge ac+bd$ and $ab+cd\ge ad+bc$.
Successive exchanges give this pairing rule.
Put $f=C_{2,j}$ when $C_{3,j}=0$.
Put $h=C_{2,j}$ when $C_{3,j}\in\{1,2\}$.
Both numbers are positive and odd by \eqref{even:eq-small-fibre-parity}.
Pairing the largest fibres first gives
\begin{equation}
\begin{array}{c|c}
 \text{number of $3$-fibres in column $j$}
 &\text{upper bound for $e_j$}\\ \hline
 0&2k+2+f,\\
 1&2k+5+h,\\
 2&2k+8+h.
\end{array}
\label{even:eq-column-capacity-table}
\end{equation}
\medskip
\noindent\emph{Column types.}
Type $O_f$ means that the column has no triple fibre and has $f$
doubleton fibres.  Type $S_h$ means one triple fibre and $h$ doubletons.
Type $T_h$ means two triple fibres and $h$ doubletons.
For $O_f$, pair doubletons together except for one paired with a singleton.
For $S_h$, pair the triple with a doubleton.  Pair the remaining
doubletons together.
For $T_h$, pair the triples together.  Use the $O_f$ rule on the rest.
For example, the $O_f$ maximum is
\[
 4\frac{f-1}{2}+2+\frac{4k+4-2f}{2}=2k+2+f.
\]

By \eqref{even:eq-small-fibre-parity}, every column contains a doubleton
fibre.  One column cannot contain three triple fibres.
Indeed, take a doubleton fibre from another column.  It can meet at most
two of these three triple fibres, so it is disjoint from at least one of
them.  Their union then has size five, which contradicts
\eqref{even:eq-union-inequality} for $i=2$.

By \eqref{even:eq-column-capacity-table}, an $O_1$ column has capacity
at most $2k+3$.  Put
\begin{equation}
 B=(4k+2)(2k+3).
 \label{even:eq-baseline-capacity}
\end{equation}
The number of different unordered row pairs in the matrix is
\begin{equation}
 R=\binom{4k+5}{2}=B+2k+4.
 \label{even:eq-required-pairs}
\end{equation}
Every one of these pairs must be covered by some column.  Equation
\eqref{even:eq-column-edge-count} therefore gives
\begin{equation}
 \sum_{j=1}^{4k+2}e_j\geq R.
 \label{even:eq-total-column-capacity}
\end{equation}
Capacity counts include repeated pairs.
Equation~\eqref{eq:cert-repetition-loss} defines repetition loss by counting
every occurrence after the first.  The number of distinct pairs is
$\sum_j e_j-\delta$.
Lost capacity means the amount by which the columns fall below their
separate upper bounds.  If these bounds sum to $R+\varepsilon$, lost
capacity and repetition loss together must be at most $\varepsilon$.

\medskip
\noindent\textbf{Subcase 1: no column has a triple fibre.}
\phantomsection\label{even:subcase-no-triple}
In this subcase, we claim that every column has type $O_1$ or $O_3$.
Indeed, let $f_j=C_{2,j}$.  Every $f_j$ is positive and odd.
Suppose that some column has $f_j\geq5$.  If all the other columns had $f=1$, then
\eqref{even:eq-column-capacity-table} and
\eqref{even:eq-required-pairs} would give
\[
 (4k+3)+(4k+1)(2k+3)=R-4.
\]
Here $f_j\le2k+1$, since the doubletons are disjoint and $f_j$ is odd.
Thus the first column contributes at most $2k+2+f_j\le4k+3$.
A second column must have at least three doubleton fibres.
Choose one doubleton in a third column.  Among the three disjoint
doubletons in the second column, one avoids this set.  Among the five or
more disjoint doubletons in the first column, one further avoids all four
selected vertices.  These three disjoint doubletons come from distinct columns.
Their union has size six, contrary to
\eqref{even:eq-union-inequality} with $i=3$.

Only $O_1$ and $O_3$ remain.  At least five $O_3$ columns are required.
Indeed, let $s$ be their number.
Equation~\eqref{even:eq-column-capacity-table} gives
\[
 B+2s\geq R,
 \qquad\text{so}\qquad
 s\geq k+2\geq5.
\]
The three doubletons of each $O_3$ column form a matching.
Equation~\eqref{even:eq-union-inequality} excludes a transversal matching
of size three.
Equation~\eqref{even:eq-complete-transversal-bound} gives $s\le4$,
contrary to $s\ge k+2\ge5$.
\medskip
\noindent\textbf{Subcase 2: some column has two triple fibres.}
\phantomsection\label{even:subcase-two-triple}
In this subcase, the capacity bound can reach the required total only
if no pair is repeated.  The doubleton fibres will force a repeated pair.
Denote the two disjoint triple fibres by $A$ and $D$.
Their column has type $T_h$.  No other column can have a triple
fibre.  Otherwise, by \eqref{even:eq-union-inequality}, a third triple
would have to meet both $A$ and $D$ in at least two vertices, which is
impossible.

Every other column must in fact have type $O_1$.
To prove this, \eqref{even:eq-union-inequality} first shows that every
doubleton in an $O_f$ column meets both $A$ and $D$.
It is therefore an edge of $K_{3,3}[A,D]$.
Such a column contains one or three doubletons.
Every doubleton in the $T_h$ column lies
outside $A\cup D$.  Suppose that an $O_f$ column contains three
doubletons.  Choose a doubleton in another $O_f$ column.
One edge of the matching of three doubletons avoids it.
Finally, choose a doubleton from the $T_h$ column, outside $A\cup D$.  These three edges form
a transversal matching of size three, contrary to
\eqref{even:eq-union-inequality}.  Therefore, every $O_f$ column has type
$O_1$.

We now show that the total capacity is at most $R$.
Indeed, the $T_h$ column contains at most $2k-1$ doubleton fibres.  By
\eqref{even:eq-column-capacity-table}, the total capacity is therefore at
most
\begin{equation}
 (4k+7)+(4k+1)(2k+3)=R.
 \label{even:eq-two-triple-equality}
\end{equation}
Equation~\eqref{even:eq-total-column-capacity} forces equality in
\eqref{even:eq-two-triple-equality}, with no repeated row pair.
Reaching the bound in \eqref{even:eq-column-capacity-table} also requires
$A$ to be paired with $D$.  Thus this column covers $K_{3,3}[A,D]$.

We claim that some row pair is repeated.
Indeed, the unique doubletons in any two $O_1$ columns must meet.  Otherwise,
these two doubletons and a doubleton from the $T_h$ column outside
$A\cup D$ would form a transversal matching of size three, which is excluded
by \eqref{even:eq-union-inequality}.  By
Theorem~\ref{thm:standard-graph-tools}(vi), a pairwise intersecting family
of edges in a bipartite graph is a star.  After exchanging $A$ and $D$ if
necessary, let its centre be $a_0\in A$.  Thus every $O_1$ doubleton
contains $a_0$.  Some $O_1$ column must cover $a_0a_1$ for
$a_1\in A\setminus\{a_0\}$: the column containing the two triples cannot
cover a pair within either triple.
In this column, the fibre $\{a_0,d\}$ is paired with $\{a_1\}$.
The resulting $K_{2,1}$ and the $K_{3,3}$ from the $T_h$ column both
cover $a_1d$.  This contradicts the equality case above.

\medskip
\noindent\textbf{Subcase 3: triple fibres occur, but every column has
at most one triple fibre.}
\phantomsection\label{even:subcase-at-most-one-triple}
Let $q$ denote the number of $S_h$ columns.  All the other columns have
type $O_f$.  Since triple fibres occur, $q\ge1$.
For $q\ge3$, Lemma~\ref{even:lem-two-intersecting-triples}(ii) bounds the
number of such columns.
For $q=2$ and $q=1$, bounds on the $O_3$ columns give the remaining
capacity comparisons.

\smallskip
\noindent\emph{Restrictions for $q\geq2$.}
The triples from different $S_h$ columns are distinct.
Equation~\eqref{even:eq-union-inequality} shows that they meet in at least
two vertices.  Suppose that two of these triples were equal.
A doubleton in one column is disjoint from its own triple.
Its union with the triple in the other column would have size five,
contrary to \eqref{even:eq-union-inequality}.

Every $S_h$ column has exactly one doubleton fibre.  Otherwise, suppose
that it had three.  These doubletons are disjoint from the triple in the
same column.  Another triple has at most one vertex outside that triple.
One of the three doubletons therefore avoids the other triple.
Their union has size five, contrary to \eqref{even:eq-union-inequality}.  The unique doubleton in an
$S_h$ column is disjoint from its own triple.  By
\eqref{even:eq-union-inequality}, it must meet every triple from every
other $S_h$ column.

\smallskip
\noindent\emph{The case $q\geq3$.}
Only $q=3,4$ are possible.  Every remaining column has type $O_1$.
Indeed, Lemma~\ref{even:lem-two-intersecting-triples}(ii) gives
$q\leq4$.  Equation~\eqref{even:eq-union-inequality} also shows that every
doubleton in an $O_f$ column meets every triple from an $S_h$ column.  Lemma~\ref{even:lem-two-intersecting-triples}(iii) excludes three
disjoint doubletons meeting all these triples.
Hence every $O_f$ column has type $O_1$.

For $q=3$, Equations~\eqref{even:eq-baseline-capacity} and
\eqref{even:eq-column-capacity-table} give
\[
 B+9<R.
\]
For $q=4$, the total capacity is at most $B+12$.
We claim that the sum of the lost capacity and repetition loss is at
least three.  This gives at most $B+9<R$ covered pairs.
The proof of Lemma~\ref{even:lem-two-intersecting-triples}(ii)
excludes a common pair
in this case.  Part (i) then gives the following four triples:
\[
 T_i=U\setminus\{u_i\}
 \qquad (i=1,2,3,4)
\]
of a four-element set $U$.  If an $S_1$ column reaches its capacity
bound, then $T_i$ is paired with a doubleton containing $u_i$.  This column covers every $u_iu_j$ with $j\ne i$.
Two such $S_1$ columns $i,j$ therefore both cover $u_i u_j$.  Let $\ell$ denote the
amount by which the four $S_1$ columns are below their separate capacity
bounds.  If $\ell\le2$, at least $4-\ell$ of these columns reach their bounds.
They repeat at least
\[
 \binom{4-\ell}{2}
\]
pairs.  Adding the lost capacity gives at least three units of total
loss.  If $\ell\geq3$, the lost capacity alone gives this bound.  Hence
the number of different row pairs covered is at most
\[
 B+12-3=B+9<R.
\]

\smallskip
\noindent\emph{The case $q=2$.}
We next show that at most two remaining columns have type $O_3$.
Indeed, write the triples as
\[
 A=C\cup\{a\},
 \qquad
 A'=C\cup\{b\},
 \qquad
 C=\{x,y\}.
\]
By \eqref{even:eq-union-inequality}, every doubleton in an $O_f$ column
meets both triples.  In each column, each of $x,y$ belongs to at most one doubleton.  A doubleton avoiding $\{x,y\}$ must be $\{a,b\}$ to meet both
triples.  Hence an $O_f$ column has at most three doubletons.
If it has three, its matching has the form
\begin{equation}
 \{ab,\;xp_i,\;yq_i\},
 \qquad p_i\neq q_i.
 \label{even:eq-two-special-rich-shape}
\end{equation}
Thus such a column has type $O_3$.
The three edges are disjoint, so $p_i,q_i$ lie outside $A\cup A'$.
There can be at most two $O_3$ columns.  To see this, suppose that there are three.  Take
$ab$ from the third column.  For the other two columns $i,j$, both cross
choices avoid this edge.  The two edges in each choice must therefore meet;
otherwise, the three selected edges form a transversal matching.
The cross choices $(xp_i,yq_j)$ and $(yq_i,xp_j)$ both have intersecting edges only when
\[
 p_i=q_j,\qquad q_i=p_j.
\]
Thus $(p_i,q_i)$ is the reverse of $(p_j,q_j)$.  If this
were true for all three pairs of columns, then the second and third
ordered pairs would be equal and at the same time would be reverses of
each other.  It would follow that one column has $p_i=q_i$, contrary to
\eqref{even:eq-two-special-rich-shape}.

Let $s$ denote the number of $O_3$ columns.  The preceding argument gives
$s\leq2$.  Equations~\eqref{even:eq-baseline-capacity} and
\eqref{even:eq-column-capacity-table} now give
\begin{equation}
 B+6+2s\leq B+10.
 \label{even:eq-two-special-capacity}
\end{equation}
For $k\ge4$, \eqref{even:eq-required-pairs} makes this bound smaller
than $R$.  Only $k=3$, $s=2$ can reach $R$.
In that case both $S_1$ columns must reach their bounds.
The column with triple $A$ pairs it with a doubleton containing $b$.
The column with triple $A'$ pairs it with a doubleton containing $a$.
Both therefore cover $ab$.
The capacity equals the required count, so this repetition leaves some
pair in \eqref{even:eq-required-pairs} uncovered,
contrary to orthogonality.

\smallskip
\noindent\emph{The case $q=1$.}
Let $A$ be the unique triple fibre.  Let $h$ be the number of doubletons in
its $S_h$ column.  Let $s$ be the number of $O_3$ columns.  By
\eqref{even:eq-union-inequality}, every doubleton in an $O_f$ column meets
$A$.  Since these doubletons are disjoint, each such column has at most
three of them.  The parity condition \eqref{even:eq-small-fibre-parity}
leaves one or three.
An $O_3$ column therefore gives a matching of size three from $A$
to three vertices outside $A$.  Lemma~\ref{even:lem-complete-transversal-certificate}
gives
\begin{equation}
 s\leq4.
 \label{even:eq-one-special-rich-bound}
\end{equation}

If $h\ge5$, then $s=0$.
Indeed, suppose otherwise that $h\ge5$ and $s\ge1$.
Since \eqref{even:eq-one-special-rich-bound} gives $s\le4$, at least one of the
$4k+1\ge13$ other columns has type $O_1$.
Choose its doubleton.  Then choose a disjoint edge in an $O_3$ column.  The four selected vertices meet at most four among the
$h$ disjoint doubletons in the $S_h$ column, so one doubleton in that
column remains disjoint.  The three edges form a transversal matching of size
three, contrary to \eqref{even:eq-union-inequality}.
Thus $s=0$.  By \eqref{even:eq-column-capacity-table}, the bound
$h\le2k+1$ gives total capacity at most
\[
 B+2k+3<R.
\]

If $h=3$, then $s\le1$.
If $h\ge3$ and $s\ge2$, choose disjoint edges from two $O_3$ columns.
Their two endpoints outside $A$ meet at most two
doubletons in the $S_h$ column.  All these doubletons lie outside $A$.
A third disjoint edge therefore remains.  This contradicts \eqref{even:eq-union-inequality} with $i=3$.
Thus $s\le1$.  Equation~\eqref{even:eq-column-capacity-table} bounds
the capacity by
\[
 B+7<R.
\]

It remains to consider $h=1$.  Equations~\eqref{even:eq-column-capacity-table},
\eqref{even:eq-baseline-capacity}, and
\eqref{even:eq-required-pairs} give the following capacity condition:
\[
 B+3+2s\geq B+2k+4.
\]
Together with \eqref{even:eq-one-special-rich-bound}, this condition leaves
only the following case:
\begin{equation}
 k=3,\qquad h=1,\qquad s=4.
 \label{even:eq-one-special-critical}
\end{equation}
The capacity in this case is $R+1$, so at most one unit can be lost.
We claim that the repetition loss is at least three.
By
\eqref{even:eq-complete-transversal-extremal} in
Lemma~\ref{even:lem-complete-transversal-certificate}, the four matchings
coming from the $O_3$ columns have the form
\[
 M,M,N,N,
\]
where $M\cup N$ is a six-cycle between $A$ and a three-element set $D$.
The doubleton $E$ in the $S_1$ column is contained in $D$.
Indeed, every two-element subset
of $D$ is the pair of outside endpoints of two disjoint edges, one from
$M$ and the other from $N$.  If $E$ did not meet such a pair, these three
edges would form a transversal matching of size three.  This forces $E$ to
meet every two-element subset of $D$, so
\begin{equation}
 |E\cap D|=2.
 \label{even:eq-special-edge-meets-D}
\end{equation}

The $S_1$ column must pair $A$ with $E$.
Otherwise, its capacity is at least two below the bound.
The total capacity has only one extra unit.
This column must therefore contain all the edges of $K_{3,2}[A,E]$.  At
most one $O_3$ column can be below its capacity bound.  At least three
$O_3$ columns therefore pair two doubletons.  Each pair gives a $K_{2,2}$.  By
\eqref{even:eq-special-edge-meets-D}, at least one of the two outside
endpoints of each such pair lies in $E\cap D$.  This $K_{2,2}$ then contains an $A$--$E$ edge
which is already covered by the $K_{3,2}$ from the $S_1$ column.
Each of these three columns adds at least one repeated occurrence.
Thus the repetition loss is at least three, even if the same pair occurs
in several of these columns.  The capacity exceeds $R$ by only one.
This is a contradiction.

All possible column profiles have now been excluded.  The assumed UOM
does not exist.
\end{proof}

\begin{remark}
\label{even:rem-computation-scope}
The proof of Theorem~\ref{thm:even-obstruction} is valid for every UOM.
Its only computer check is the exhaustive matching calculation in
\eqref{eq:cert-transversal-predicate}--\eqref{eq:cert-fixed-set-shape}, carried
out by \nolinkurl{check_finite_obstructions.py} and used in
Lemma~\ref{even:lem-complete-transversal-certificate}.  A search restricted
to added columns, a square kernel, or a construction from input UOMs is
not sufficient for proving this theorem.
\end{remark}

\subsection{Consequences for the spectrum in even dimensions}
\label{even:subsec-spectrum-consequences}

Together with \eqref{eq:group-perfect-final-two-mod-four-tail},
Theorem~\ref{thm:even-obstruction} gives the
$n\equiv2\pmod4$ case of Theorem~\ref{thm:complete-spectrum}.  Thus, for
every $n\ge22$ in this class,
\[
 \Theta_n=\{n+2\}\cup[n+4,2^n-6]\cup\{2^n-4,2^n\}.
\]
\section{The boundary case
\texorpdfstring{$13\times10$}{13 by 10}}
\label{sec:13-by-10}

Theorem~\ref{thm:even-obstruction} treats $n\ge14$, so only
$(m,n)=(13,10)$ remains.  The proof starts from
Lemma~\ref{even:lem-union-inequality}, the parity condition
\eqref{even:eq-small-fibre-parity}, and the capacity bounds in
\eqref{even:eq-column-capacity-table}.  Three finite searches exclude the
last column types.  They are carried out together by
\nolinkurl{check_finite_obstructions.py}.  Their exhaustive outcomes are
recorded in \eqref{eq:cert-O3-counts},
\eqref{eq:cert-two-special-minima}, and
\eqref{eq:cert-one-special-counts} of
Appendix~\ref{app:certificate-contracts}.

\Needspace{13\baselineskip}
\subsection{Proof of Theorem~\ref{thm:ten-boundary}}

\begin{proof}
The proof uses the same row-pair count as
Theorem~\ref{thm:even-obstruction}, with $k=2$.
Two triple fibres in one column are excluded first.
The number $q$ of columns containing a triple then gives the cases
that remain.  In each case, the loss forced by repeated pairs must not
exceed the loss allowed by the capacity count.
The three finite checks in \eqref{eq:cert-O3-counts},
\eqref{eq:cert-two-special-minima}, and \eqref{eq:cert-one-special-counts}
exclude the cases left by these counts.

Suppose that $X$ is a $13\times10$ UOM.
The mate completeness argument preceding
\eqref{even:eq-small-fibre-parity} also holds at $k=2$.
Thus that parity condition and \eqref{even:eq-union-inequality} still apply.
In particular, doubletons from three distinct columns cannot be disjoint.
There are
\[
 R=\binom{13}{2}=78
\]
row pairs which must be covered.  Use the column types $O_f,S_h,T_h$ from the preceding proof.
Let $e(P)$ count the row pairs covered by a column of type $P$.  For $k=2$,
Equation~\eqref{even:eq-column-capacity-table} becomes
\[
 e(O_f)\leq6+f,\qquad
 e(S_h)\leq9+h,\qquad
 e(T_h)\leq12+h.
\]
An $O_1$ column has capacity $7$.  An $O_3$ column has capacity $9$ exactly when two doubletons are mates.
Otherwise, its capacity is $8$.
An $S_1$ column has capacity $10$ exactly when its triple and doubleton
are mates.  Otherwise, its capacity is $8$.  Falling below the bound loses one
unit for an $O_3$ column and two units for an $S_1$ column.
If two doubletons have singleton mates, pairing them together adds
$4+1-(2+2)=1$ row pair.
The term $+1$ counts the pair formed by the two singletons.
For a triple and a doubleton with singleton mates, the same change adds
$6+1-(3+2)=2$ row pairs.

First exclude two triple fibres in one column.
The argument in \hyperref[even:subcase-two-triple]{Subcase 2} uses
\eqref{even:eq-union-inequality} to make the other nine columns $O_1$.
It applies at $k=2$: the $T_h$ doubletons lie outside the two triples,
so an $O_3$ column would give three disjoint doubletons.
Equation~\eqref{even:eq-two-triple-equality} becomes $15+9\cdot7=78$.
Equality requires the two triples to be mates.
No row pair can then be repeated.
The $O_1$ doubletons must then form a star by
Theorem~\ref{thm:standard-graph-tools}(vi).
Covering a pair inside one triple repeats a pair between the two triples,
as in that subcase.  This excludes equality.
Three triple fibres are already excluded by
\eqref{even:eq-union-inequality} with $i=2$.
Let $q$ count the columns with one triple.
The remaining cases are $q=0$, $q\ge3$, $q=2$, and $q=1$.
\subsubsection{No triple fibres}
First exclude $O_5$.
Here $f\in\{1,3,5\}$.
If an $O_5$ column were the only column not of type $O_1$, the capacity
would be $11+9\cdot7=74<78$ by \eqref{even:eq-column-capacity-table}.
A second column must therefore have at least three doubletons.
Choose a doubleton in a third column.  Choose a disjoint one in the second.
The five doubletons in the first column cannot all meet these four vertices.
A third disjoint doubleton remains, contrary to
\eqref{even:eq-union-inequality} with $i=3$.
Thus only $O_1,O_3$ remain.
We claim that exactly four $O_3$ columns are required.
Each must have capacity $9$.  No pair can be repeated.
Indeed, let $s$ denote their number.
The capacity count gives
\[
 70+2s\geq78,
\]
while Lemma~\ref{even:lem-complete-transversal-certificate} gives
$s\leq4$.  Hence $s=4$.  All four columns must reach capacity $9$.
No row pair can be repeated.  Their doubletons give four matchings of
size three without a transversal matching of size three.
The search uses the $90$ such four-matching lists in
\eqref{eq:cert-rgs-counts}.
For each matching, a column at capacity $9$ has
\[
 3\cdot7\cdot15=315
\]
ways to select the mate pairs.  First select the doubleton paired with a
singleton.  Then select this singleton.  Finally, pair the other six
singletons.  Let $B_1,\ldots,B_4$ be the resulting edge sets on $[13]$.
These graphs must satisfy
\[
\begin{aligned}
 B_i\cap B_j&=\varnothing &&(i<j),\\
 \sum_{i=1}^4\bigl(\deg_{B_i}(v)-1\bigr)&\le2 &&(v\in[13]).
\end{aligned}
\]
The first condition means that no row pair is repeated.
For the second, mate completeness gives degree at least one at every row
in each of the ten columns.  With no repeated pair, the total degree at
each row is $12$.  Subtracting one degree unit for each of the ten columns
leaves only two.  The excess degree from the four selected columns is
therefore at most two.
Every such choice occurs in the finite domain.
None satisfies both conditions.  This is the last zero in \eqref{eq:cert-O3-counts}.
Thus the case without a triple fibre is impossible.

\subsubsection{At least three \texorpdfstring{$S_h$}{S-h} columns}
For $q\ge3$, the restrictions in
\hyperref[even:subcase-at-most-one-triple]{Subcase 3} use only
\eqref{even:eq-union-inequality}, so they apply at $k=2$.
The triples are distinct.  Every $S_h$ column is $S_1$.
Lemma~\ref{even:lem-two-intersecting-triples}(ii)--(iii) gives
$q\in\{3,4\}$.  It also shows that every other column is $O_1$.
If $q=3$, the maximum capacity $79$ allows only one repetition.
We will show that at least three pairs are repeated.
Each $S_1$ column must reach capacity $10$.
A smaller value loses two units, which is already too much.
The two forms of the triple family in
Lemma~\ref{even:lem-two-intersecting-triples}(i) give the repeated pairs.
First suppose that the triples share a pair $C$.
Write them as $C\cup\{a_i\}$ for $i=1,2,3$.  The doubleton paired with the $i$-th triple must then be
$\{a_j,a_k\}$.  The $i$-th and $j$-th columns both cover $a_i a_j$.
Thus three pairs are repeated.

For the other type of triple family, write
$T_i=U\setminus\{u_i\}$, $i=1,2,3$, inside a four-element set $U$.
The doubleton paired with $T_i$ must contain $u_i$.  It follows that each
of the three pairs $u_i u_j$, $1\leq i<j\leq3$, is again covered twice.
The maximum capacity $79$ permits only one repetition.
This is a contradiction.

If $q=4$, the maximum capacity is $82$.
Only four units can be lost.
The sum of the lost capacity and repetition loss is at least five.
Indeed, the common-pair case in the proof of
Lemma~\ref{even:lem-two-intersecting-triples}(ii) permits at most three triples.
Thus part (i) puts these four triples in a four-element set.
Write them as
\[
 T_i=U\setminus\{u_i\},\qquad i=1,2,3,4,
\]
In an $S_1$ column at capacity $10$, the doubleton paired with $T_i$
contains $u_i$.  Suppose that $a$ of
the four $S_1$ columns do not reach capacity $10$.  Since each such
column loses two capacity units, the four available units give $a\leq2$.
For each pair of the
remaining columns $i,j$, the pair $u_i u_j$ is covered twice.  Thus the
sum of the lost capacity and repetition loss is at least
\[
 2a+\binom{4-a}{2}
 =
 \begin{cases}
 6,&a=0,\\
 5,&a=1,\\
 5,&a=2.
 \end{cases}
\]
The maximum capacity is $82$.  Each value above is at least five, so fewer
than $78$ different pairs can be covered.

\subsubsection{Exactly two \texorpdfstring{$S_h$}{S-h} columns}
The restrictions for $q\ge2$ in
\hyperref[even:subcase-at-most-one-triple]{Subcase 3} make both triple columns $S_1$.
They also show that the triples are distinct.
Their intersection has two rows.
We next show that at most two other columns have type $O_3$.
Indeed, write the triples as $A=C\cup\{a\}$ and $A'=C\cup\{b\}$,
where $C=\{x,y\}$.
Equation~\eqref{even:eq-union-inequality} makes each doubleton in another
column meet both triples.
The matching argument leading to \eqref{even:eq-two-special-rich-shape}
therefore applies unchanged.
Every $O_3$ column has matching $\{ab,xp_i,yq_i\}$ with $p_i\ne q_i$.
At most two such columns can occur.
If $s$ counts them, \eqref{even:eq-two-special-capacity} at $k=2$
gives maximum capacity $76+2s$.
Thus only $s=1,2$ remain.
If $s=1$, the total capacity is only $78$, so no pair can be repeated.
But the pair $ab$ is covered twice.
Indeed, both $S_1$ columns must reach capacity
$10$.  The $O_3$
column must reach capacity $9$.  The doubleton paired with $A$ contains $b$.
The one paired with $A'$ contains $a$.  Therefore, the two
complete bipartite graphs arising from the $S_1$ columns both cover
$ab$.  This is impossible because the capacity count does not allow any
repeated pair.

If $s=2$, the finite check shows that every allowed column choice loses
more than the two units permitted by the maximum capacity $80$.
To state the check, write the two doubletons from the $S_h$ columns as
\[
 E=\{b,u\},\qquad E'=\{a,v\},
\]
where $u,v$ are selected from the other nine rows.  The second finite
search contains all $419\,904$ labelled choices.  Exactly $4\,104$ have
no transversal matching of size three.  Let
$s_0$ be the number of $S_1$ columns below capacity $10$.  Let $b_0$
be the number of $O_3$ columns below capacity $9$.  Capacity permits only
$2s_0+b_0\le2$.  For the four possible pairs
\[
 (s_0,b_0)=(0,0),(0,1),(0,2),(1,0),
\]
the forced subgraphs have minimum repetition losses $5,3,1,1$,
respectively, by \eqref{eq:cert-two-special-minima}.
The search retains the forced $K_{3,2}$ or $K_{2,2}$ in each column at
capacity.  Equation~\eqref{eq:repetition-loss-monotonicity} makes these
minima lower bounds for the complete columns.
After subtracting the lost capacity, only $2,1,0,0$ repetitions are allowed
in the four respective cases.
Each forced loss is larger than the loss allowed in the same case.
This excludes $q=2$.
\subsubsection{Exactly one \texorpdfstring{$S_h$}{S-h} column}
First reduce to $h=1$ with three or four $O_3$ columns.
Let $A$ be the triple.  Write $h$ for its column's doubleton count.
Let $s$ be the number of $O_3$ columns.
Equation~\eqref{even:eq-union-inequality} makes each $O_3$ matching join
$A$ to three outside rows.
Equation~\eqref{even:eq-complete-transversal-bound} gives $s\le4$.
The $q=1$ argument in
\hyperref[even:subcase-at-most-one-triple]{Subcase 3} also applies at $k=2$.
For $h\ge5$, there is still an $O_1$ column to choose, since there
are nine other columns and at most four $O_3$ columns.
The argument gives $s=0$.  The capacity is then at most $77$.
If $h=3$, it gives $s\le1$.  The same capacity bound applies.
For $h=1$, \eqref{even:eq-column-capacity-table} gives $73+2s$.
This reaches $78$ only for $s=3,4$.
For $s=3$, the maximum total is only $79$.
We claim that the sum of the lost capacity and repetition loss is at
least two.
First, the $S_1$ column must reach capacity $10$.
Otherwise, its loss of two units already gives a contradiction.
Denote its doubleton by $E$.  For an $O_3$ matching $M_i$, define
\[
 H_i=\{e\in M_i:e\cap E=\varnothing\}.
\]
For two $O_3$ columns at capacity $9$, we claim that either at least one
repeats a pair from the $S_1$ column, or the two columns have two common pairs.
Indeed,
\eqref{even:eq-union-inequality} does not allow a transversal
matching of size three containing $E$.  Thus every edge of
$H_i$ meets every edge of $H_j$.  An $O_3$ column reaching capacity $9$
pairs two edges of its matching.  If one of these two edges meets $E$,
the resulting $K_{2,2}$ repeats an edge of the $K_{3,2}[A,E]$ from the
$S_1$ column.  Suppose that two $O_3$ columns reaching capacity $9$ avoid
such a repeated edge.  Each of them must then choose two edges from its
set $H_i$.  Each selected edge must meet both disjoint edges from the
other matching.  Its endpoints therefore lie one on each of those edges.
The two choices are the alternating matchings of a four-cycle.
Their $K_{2,2}$ graphs have two common edges:
the pair within $A$ and the pair between the two outside vertices.

If all three $O_3$ columns reach capacity $9$, there are two possibilities.
If at least two repeat a pair from the $S_1$ column, they give repetition
loss at least two.  Otherwise, at least two avoid such a repetition.
The four-cycle argument shows that these two columns have two common
pairs.  Again the repetition loss is at least two.
The total capacity is $79$, which allows loss at most one.
If one $O_3$ column has capacity $8$, the other two have capacity $9$.
Either one repeats a pair from the $S_1$ column, or they have two common
pairs.  Both alternatives contradict total capacity $78$, which allows
no repeated pair.  Two columns of capacity $8$ would already give total
capacity at most $77$.

For $s=4$, the maximum total is $81$, so only three units can be lost.
The loss depends on whether the $S_1$ column reaches capacity $10$.
Equation~\eqref{even:eq-complete-transversal-extremal} in
Lemma~\ref{even:lem-complete-transversal-certificate} gives
\[
 M,M,N,N,
\]
where $M\cup N$ is a six-cycle between $A$ and a three-element set $D$.
The doubleton $E$ in the $S_1$ column must meet every two-element subset of
$D$.  Indeed, $E$ is disjoint from $A$.
The argument for \eqref{even:eq-special-edge-meets-D}
uses this fact, the six-cycle, and the absence of three disjoint
doubletons from distinct columns.
These conditions hold here, so it gives
$|E\cap D|=2$.

If the $S_1$ column reaches capacity $10$, the total loss is at least
four units.  Indeed, let
$b_0\leq3$ be the number of the four $O_3$ columns that do not reach
capacity $9$.  Every $K_{2,2}$ coming from an $O_3$ column of capacity
$9$ repeats one edge of the $K_{3,2}[A,E]$ coming from the $S_1$ column.
Therefore, the sum of the lost capacity and repetition loss is at
least
\[
 b_0+(4-b_0)=4,
\]
while the maximum capacity $81$ requires this total to be at most three.

If the $S_1$ column does not reach capacity $10$, it loses two of the
three allowed units.
At most one $O_3$ column can also fail to reach capacity $9$.
The last finite check rules out these two possibilities.
In the first choice, all four $O_3$ columns reach their bounds.
In the second, exactly one $O_3$ column falls below its bound.  Each
$O_3$ column at capacity $9$ pairs two of its three doubletons.
There are three choices for this pair.  The last finite search gives the two exact minima in
\eqref{eq:cert-one-special-counts}: the $81$ choices in the first case have
repetition loss at least four.
The $108$ choices in the second case have repetition loss at least two.
These counts use only the $K_{2,2}$
subgraphs from the $O_3$ columns at capacity $9$.
Equation~\eqref{eq:repetition-loss-monotonicity} shows that these lower
bounds also hold for the complete columns.
After subtracting the lost capacity, the first choice permits repetition
loss at most one.  The second permits none.
Both are impossible.
This rules out every possible column type.  A $13\times10$ UOM does not
exist.
\end{proof}
\section{Finite computations used in the proofs}
\label{app:certificate-contracts}

This appendix specifies the three finite checks used in the proofs:
construction data, five point cliques, and matching obstructions.
The programs use integers and finite sets, represented by integer bit masks.
The additional matrices in \nolinkurl{finite_uoms.txt} are checked by
\eqref{eq:cert-orth} and \eqref{eq:cert-uom-dp}.
The program \nolinkurl{check_finite_uoms.py} in
\nolinkurl{audit/finite_uom_check} checks these conditions directly from
the matrix entries.
\subsection{Construction data}

\medskip
\noindent\emph{Basic matrix checks.}
The program \nolinkurl{check_construction_data.py} checks the finite inputs
for the constructions based on row decompositions and on a good column
structure.  It reads
\nolinkurl{construction_data.json}.  Each column has its own alphabet.
The mate rule is \eqref{eq:mate-rule}.  For a symbol $a$ in column $j$,
let $F_{X,j}(a)$ be the fibre in
Definition~\ref{def:equality-fibre}(i) and \eqref{eq:equality-fibre}.  The basic checks are
\begin{align}
 \mathsf{Orth}(X)
 &\Longleftrightarrow
 \forall r<s\ \exists j\quad X_{rj}=\overline{X_{sj}},
 \label{eq:cert-orth}\\
 \mathsf{MC}(X)
 &\Longleftrightarrow
 \forall j,a\quad F_{X,j}(a)\ne\varnothing
       \Longrightarrow F_{X,j}(\bar a)\ne\varnothing.
 \label{eq:cert-mate-complete}
\end{align}
The first line checks row orthogonality in
Definition~\ref{def:uom}(i)--(ii).
The second condition is used only where mate completeness is required.

\smallskip
\noindent\emph{Fibre-cover recurrence.}
To test unextendibility, the program finds every row set that can be
covered by at most one fibre from each column.
For each column, put
$\mathcal F_{X,j}=\{F_{X,j}(a):F_{X,j}(a)\ne\varnothing\}$.  The fibre
cover calculation is
\begin{equation}
 \mathcal R_0=\{\varnothing\},\qquad
 \mathcal R_j=\{U\cup F:U\in\mathcal R_{j-1},\ 
              F\in\mathcal F_{X,j}\cup\{\varnothing\}\}.
 \label{eq:cert-fibre-dp}
\end{equation}
Repeated sets are removed after each column.
An orthogonal matrix is unextendible exactly when no set in the final
list contains all rows:
\begin{equation}
 \mathsf{UOM}_{\rm DP}(X)
 \Longleftrightarrow
 \mathsf{Orth}(X)\ \wedge\ [m]\notin\mathcal R_n.
 \label{eq:cert-uom-dp}
\end{equation}
$\mathsf{UOM}_{\rm DP}(X)$ therefore checks Definition~\ref{def:uom}(iii) by
Theorem~\ref{thm:fibre-cover}(ii).

\smallskip
\noindent\emph{Minimal supports and repetition.}
Let $\mathcal A=(A_1,\ldots,A_q)$ be a row decomposition in
Definition~\ref{def:row-decomposition}, and set
$T_I=\bigcup_{i\in I}A_i$.  The notation $\mathsf{Cover}_X(I,J)$ means
that one fibre can be chosen in each column $j\in J$ so that the union of
the chosen fibres contains $T_I$.  This means that $J$ is a fibre support of $T_I$ in
Definition~\ref{def:equality-fibre}(ii).
Here the rows are represented by their numbers.  The complete list of minimal column supports is
\begin{equation}
 \begin{aligned}
 \mathcal A_X^{\min}(I)=\{J\subseteq[n]:{}&
       \mathsf{Cover}_X(I,J),\\
      &\mathsf{Cover}_X(I,J')\text{ fails for every }J'\subsetneq J\}.
 \end{aligned}
 \label{eq:cert-min-supports}
\end{equation}
Theorem~\ref{thm:row-decomposition-construction}(i) needs joint block
capacity below $q$.
The next recurrence tests whether a list of inputs can reach this bound.
The same input may occur more than once.
All matrices in $\mathcal C$ have the same number $N$ of columns and the
same number $q$ of decomposition blocks.
Every nonempty proper subset $I\subsetneq[q]$ is used.
The full row set has no support because the input matrix is a UOM.
For a used column set $U$, the state value $d(U)$ records the largest
total number of selected blocks found so far.
A total of $q$ or more is recorded as $q$.
Initially, put $d(\varnothing)=0$ and
$d(U)=-\infty$ for $U\ne\varnothing$.  For every $X\in\mathcal C$, every
$I\ne\varnothing$ with $I\subsetneq[q]$, every
$J\in\mathcal A_X^{\min}(I)$, and every $U$ disjoint from $J$, make the
update
\[
 d(U\cup J)\leftarrow
 \max\{d(U\cup J),\min(q,d(U)+|I|)\}.
\]
A nonempty block set has no empty support.
Each transition therefore increases $|U|$.
Processing states in increasing order of $|U|$ gives a finite calculation.

We now check that the recurrence gives the repeated-input capacity in
\eqref{eq:repetition-capacity}, with values above $q$ recorded as $q$.
Each transition adds one input whose support avoids all previously used columns,
as required by \eqref{eq:joint-capacity}.
Conversely, replace each support in a feasible list by a minimal subset.
The corresponding transitions then reproduce that list.
The same input can be chosen in several transitions, so repetitions are
included.
Thus the largest reachable value is
\begin{equation}
 \bcap^{\rm rep}_q(\mathcal C)
   :=\min\{q,\bcap^{\rm rep}(\mathcal C)\}
   =\max_{U\subseteq[N]}d(U).
 \label{eq:cert-repetition-capacity}
\end{equation}
The values used in Lemma~\ref{lem:certified-row-decompositions} are obtained
from this finite recurrence.

\smallskip
\phantomsection\label{app:cert-good-column}
\noindent\emph{The input with a good column structure.}
The same program checks the fixed $13\times8$ matrix $X_8$, its ten stated
decomposition blocks, and the chosen column $g=8$.  For every nonempty
$I\subseteq[10]$ it verifies the restricted cover costs from
Definition~\ref{def:restricted-cover-cost}:
\begin{equation}
 \begin{gathered}
 \mathsf{Orth}(X_8),\qquad \mathsf{MC}(X_8),\qquad
 \mathsf{UOM}_{\rm DP}(X_8),\\
 \kappa_{X_8}^{[8]}(T_I)\ge |I|-1,\qquad
 \kappa_{X_8}^{[8]\setminus\{8\}}(T_I)\ge |I|.
 \end{gathered}
 \label{eq:cert-good-column-acceptance}
\end{equation}
For reference, the minima over the $1023$ possible nonempty sets $I$ are
\begin{equation}
 \begin{array}{c|cccccccccc}
 |I|&1&2&3&4&5&6&7&8&9&10\\ \hline
 \min\kappa^{[8]}_{X_8}(T_I)
  &1&1&2&3&4&5&6&7&8&\infty\\
 \min\kappa^{[8]\setminus\{8\}}_{X_8}(T_I)
  &1&2&3&4&5&6&7&\infty&\infty&\infty.
 \end{array}
 \label{eq:cert-good-column-minima}
\end{equation}
These checks concern the stated decomposition of $X_8$.
The later lifted matrices follow from
Theorem~\ref{thm:good-column-even-lift}.

\subsection{The five point completion graph}
\label{app:cert-five-point}

Let $A=\mathbb Z_5$ and let $F_a$ be the five near-one-factors in
\eqref{eq:five-kernel}.  Following Definition~\ref{def:completion-graph}(ii)
and the adjacency rule \eqref{eq:five-adjacency}, the program
\nolinkurl{check_five_point_cliques.py} constructs the completion graph on
the $120$ permutations in $\Sym(A)$.  Put
\[
 W(\sigma,\tau)=\{a\in A:\{\sigma(a),\tau(a)\}\in F_a\}.
\]
For every supplied set $C\subseteq\Sym(A)$, it checks
\begin{align}
 \mathsf{Clique}(C)&\Longleftrightarrow
   \forall\sigma\ne\tau\in C\quad W(\sigma,\tau)\ne\varnothing,
 \label{eq:cert-five-clique}\\
 \mathsf{Max}(C)&\Longleftrightarrow
   \forall\pi\in\Sym(A)\setminus C\ \exists\sigma\in C\quad
                         W(\pi,\sigma)=\varnothing.
 \label{eq:cert-five-maximal}
\end{align}
The file \nolinkurl{five_point_cliques.json} contains one witness for each
order actually used in the proof.  The required statement is
\begin{equation}
 \begin{aligned}
 \mathsf{FivePointReq}
 &=\bigwedge_{c\in\Lambda_0}\ \exists C_c\subseteq\Sym(A):\\
 &\quad |C_c|=c,\quad \mathsf{Clique}(C_c),\quad \mathsf{Max}(C_c),\\
 &\quad \Lambda_0=\{1\}\cup[4,19].
 \end{aligned}
 \label{eq:cert-five-required}
\end{equation}
One maximal clique of each required order is enough for
Lemma~\ref{lem:full-link}(i).

\subsection{Finite matching searches}
\label{app:cert-transversal}

\medskip
\noindent\emph{Transversal matching search.}
The program \nolinkurl{check_finite_obstructions.py} checks the matching
families below.  Let
$\mathbf M=(M_0,\ldots,M_{s-1})$ be an ordered list of matchings with
three edges.  Repetitions are allowed.
Write $V(\mathbf M)$ for the vertices appearing in the list.
The following condition tests for a transversal matching in
Definition~\ref{def:graph-notation}(iv):
\begin{equation}
 \begin{aligned}
 \mathsf{TM}_3(\mathbf M)\Longleftrightarrow{}&
 \exists\,0\le i<j<k<s,\\[-1mm]
 &\exists e_i\in M_i,e_j\in M_j,e_k\in M_k:\\[-1mm]
 &e_i,e_j,e_k\text{ are pairwise disjoint}.
 \end{aligned}
 \label{eq:cert-transversal-predicate}
\end{equation}
Fix $M_0=\{01,23,45\}$.  The $48$ automorphisms of $M_0$ reduce the second
matching to $27$ representatives.
At later levels, previously used vertices keep their labels.
A new vertex receives the next unused label.
For $0\le h\le6$, choose $6-h$ vertices from $\mathbb Z_v$
and add $v,\ldots,v+h-1$.  Let $\mathcal G(v)$ contain all perfect
matchings on every six-element set obtained in this way.  The recurrence is
\begin{equation}
 \begin{aligned}
 \mathcal P_{s+1}=\{(\mathbf M,M):{}&\mathbf M\in\mathcal P_s,\\
 &M\in\mathcal G(1+\max V(\mathbf M)),\\
 &\neg\mathsf{TM}_3(\mathbf M,M)\}.
 \end{aligned}
 \label{eq:cert-rgs-recurrence}
\end{equation}
Every list without a transversal matching of size three occurs after
vertex relabelling.  The level counts are
\begin{equation}
 (|\mathcal P_2|,|\mathcal P_3|,|\mathcal P_4|,|\mathcal P_5|)
 =(27,62,90,0).
 \label{eq:cert-rgs-counts}
\end{equation}
The last count excludes a list of length five.
This proves Lemma~\ref{even:lem-complete-transversal-certificate}(i).

\smallskip
\noindent\emph{A fixed three-point set.}
The next search gives the form of the four matchings in
Lemma~\ref{even:lem-complete-transversal-certificate}(ii).
Take $A=\{0,1,2\}$.
Represent a matching by a one-to-one map from $A$ to the vertices outside
$A$.  Start with $(3,4,5)$.  At each step the range may use old outside
vertices and consecutively labelled new vertices.  A list is retained
exactly when there are no indices $i<j<k$ and no permutation
$(a_i,a_j,a_k)$ of $A$ for which
$f_i(a_i),f_j(a_j),f_k(a_k)$ are distinct.  The resulting counts are
\begin{equation}
 (|\mathcal J_1|,|\mathcal J_2|,|\mathcal J_3|,|\mathcal J_4|)
 =(1,48,6,6),
 \label{eq:cert-fixed-set-counts}
\end{equation}
Each final list has the form
\begin{equation}
 \{M,M,N,N\},\qquad M\cup N=C_6[A,D],\qquad |A|=|D|=3.
 \label{eq:cert-fixed-set-shape}
\end{equation}

\smallskip
\phantomsection\label{app:cert-13x10}
\noindent\emph{The $13\times10$ searches.}
The same program treats the three finite families left in the
$13\times10$ proof.  The first check asks whether four $O_3$ columns can have pairwise
disjoint sets of covered row pairs.  The other two give lower bounds on
the loss from repeated pairs.  Appendix~\ref{sec:13-by-10} compares these
bounds with the total column capacity minus $78$.
For finite edge sets $B_1,\ldots,B_t$, put
\begin{equation}
 \delta(B_1,\ldots,B_t)=\sum_{i=1}^t|B_i|-
          \left|\bigcup_{i=1}^tB_i\right|.
 \label{eq:cert-repetition-loss}
\end{equation}
In the case with four $O_3$ columns, the input is the $90$ four-matching
families in $\mathcal P_4$.  Each matching has $315$ realizations which
reach capacity nine.  The search examines $28\,350$ single columns,
$2\,400\,840$ edge-disjoint pairs, and $13\,296\,960$ edge-disjoint
triples.  At depth three it also applies
$\sum_i(\deg_{B_i}(v)-1)\le2$ at every row $v$; no edge-disjoint fourth
choice remains.  The recorded tuple is
\begin{equation}
 (28\,350,\ 2\,400\,840,\ 13\,296\,960,\ 0).
 \label{eq:cert-O3-counts}
\end{equation}
For two $S_1$ and two $O_3$ columns, the complete labelled domain has
$419\,904$ members.
Among them, $4\,104$ have no transversal matching of size three.
For the loss calculation, let $B_i$ contain only the forced
$K_{3,2}$ in an $S_1$ column at capacity $10$, or the forced $K_{2,2}$
in an $O_3$ column at capacity $9$.  Set $B_i=\varnothing$ for a column
below capacity.  The remaining edges of each column are omitted.
The loss on these forced subgraphs is a lower bound for the loss on the
complete columns.
Indeed, let $E_i$ be the complete edge set of column $i$.
Then $B_i\subseteq E_i$.
For each row pair $e$, let $r_B(e)$ and $r_E(e)$ count the sets containing it.
Thus $r_B(e)\le r_E(e)$, so
\begin{equation*}
 \begin{aligned}
 \delta(B_1,\ldots,B_t)
 &=\sum_e\max\{r_B(e)-1,0\}\\
 &\le\sum_e\max\{r_E(e)-1,0\}\\
 &=\delta(E_1,\ldots,E_t).
 \end{aligned}
 \tag{\ref*{eq:cert-O3-counts}a}\label{eq:repetition-loss-monotonicity}
\end{equation*}
The exact minima for the forced subgraphs are
\begin{equation}
 \begin{array}{c|cccc}
 (s_0,b_0)&(0,0)&(0,1)&(0,2)&(1,0)\\ \hline
 \min\delta&5&3&1&1.
 \end{array}
 \label{eq:cert-two-special-minima}
\end{equation}
Here $s_0$ is the number of the two $S_1$ columns below capacity $10$, and
$b_0$ is the number of the two $O_3$ columns below capacity $9$.
For $(s_0,b_0)=(1,0)$, the program omits the first $S_1$
subgraph.  Interchanging the two $S_1$ columns preserves the search domain,
so this also covers the choice in which the second one is below capacity.
Finally, for one $S_1$ column and four $O_3$ columns, keep only the
$K_{2,2}$ subgraphs of the $O_3$ columns at capacity $9$.
The same inequality gives a lower bound for the complete columns.
The two searches give
\begin{equation}
 \begin{aligned}
 81\text{ patterns:}&\quad \min\delta=4,\\
 108\text{ patterns:}&\quad \min\delta=2.
 \end{aligned}
 \label{eq:cert-one-special-counts}
\end{equation}
In this last case the $S_1$ column is below capacity $10$.  The first search
has $3^4=81$ choices because all four $O_3$ columns reach capacity $9$.
The second has $4\cdot3^3=108$ choices because exactly one $O_3$ column is
below capacity $9$.
These are exhaustive finite domains, not random samples.  Their construction
is written directly in the program, together with the capacity and
edge-disjointness tests used in Appendix~\ref{sec:13-by-10}.

The three programs, their two JSON input files, and
\nolinkurl{finite_uoms.txt} are available in the \certsite.
\section{Example matrices}
\label{app:explicit-matrices}

The following three matrices come from the main constructions.
All entries are column-local symbols, as in
Section~\ref{sec:preliminaries}.
In every column, the mate pairs are
$1\leftrightarrow2$, $3\leftrightarrow4$, $5\leftrightarrow6$ and so on.
These are mate pairs in the sense of \eqref{eq:mate-rule}.

\subsection{A matrix from a completion graph
\texorpdfstring{$M_{13,9}$}{M13,9}}

The first nine rows form a matching core from
Theorem~\ref{thm:matching-core-full-link}(i).
An inclusion-maximal clique gives the last four rows by
Lemma~\ref{lem:full-link}(i).  The core matchings are not
cyclic.  This also gives a $13\times9$ example for \cite{ChenChen2026}; see
Definition~\ref{def:completion-graph}(ii).
\begingroup
\small
\setlength{\arraycolsep}{3pt}
\[
M_{13,9}=\left(\begin{array}{rrrrrrrrr}
1&3&3&3&3&3&3&3&3\\
3&1&5&5&5&5&5&4&5\\
5&5&1&7&6&4&7&5&7\\
4&4&7&1&7&7&8&7&9\\
7&7&6&4&1&9&9&8&8\\
9&8&8&9&9&1&6&6&4\\
6&9&4&6&10&10&1&9&10\\
10&10&9&8&4&8&10&1&6\\
8&6&10&10&8&6&4&10&1\\ \hline
2&2&2&7&2&8&2&5&2\\
2&2&8&7&2&2&2&6&2\\
9&2&7&2&2&4&2&6&2\\
2&2&8&8&2&2&2&2&2
\end{array}\right).
\]
\endgroup
Equation~\eqref{eq:matching-core-maximality} therefore shows that
$M_{13,9}$ is a UOM.

\Needspace{18\baselineskip}
\subsection{A matrix from a matching core
\texorpdfstring{$M_{15,10}$}{M15,10}}

In this example, the first ten rows form the matching core of
Theorem~\ref{thm:matching-core-full-link}.  Five core matchings contain five
edges.  The other five contain four edges.
The last five rows come from an inclusion-maximal clique in the completion
graph of Definition~\ref{def:completion-graph}(ii).
Lemma~\ref{lem:full-link}(i) gives the corresponding completion rows.
\begingroup
\small
\setlength{\arraycolsep}{2.5pt}
\[
M_{15,10}=\left(\begin{array}{rrrrrrrrrr}
1&1&1&1&1&9&1&1&1&1\\
3&3&3&3&3&11&3&3&3&2\\
5&2&5&5&5&1&4&9&5&3\\
7&5&4&7&7&3&2&11&7&4\\
9&9&2&6&9&5&5&4&8&5\\
10&7&9&4&8&7&7&2&6&7\\
8&8&7&2&6&6&9&5&4&9\\
6&11&8&9&4&4&8&7&2&6\\
4&6&11&10&2&2&6&6&9&8\\
2&4&6&8&10&8&10&8&10&11\\ \hline
7&12&12&8&7&12&2&10&7&12\\
6&10&12&7&3&12&8&12&2&10\\
9&12&12&7&4&10&6&10&8&10\\
5&10&10&7&4&10&5&12&6&10\\
10&12&12&7&4&10&3&12&5&10
\end{array}\right).
\]
\endgroup
Lemma~\ref{lem:full-link}(i) therefore shows that $M_{15,10}$ is a UOM.

\subsection{A matrix with two added columns
\texorpdfstring{$M_{21,9}$}{M21,9}}

This matrix uses the $8\times7$ and $13\times7$ inputs in
Theorem~\ref{thm:row-decomposition-construction}.  Its last two columns are
added as in Definition~\ref{def:added-column-system}.
\begingroup
\small
\setlength{\arraycolsep}{2.8pt}
\[
M_{21,9}=\left(\begin{array}{rrrrrrr|rr}
1&1&1&1&1&1&1&1&1\\
2&3&3&3&3&3&3&3&3\\
3&2&4&5&5&5&5&3&3\\
4&4&2&7&7&7&7&3&3\\
5&5&5&2&4&6&8&1&1\\
6&7&7&4&2&8&6&1&1\\
7&6&8&6&8&2&4&3&3\\
8&8&6&8&6&4&2&1&1\\ \hline
23&21&24&26&28&27&25&2&4\\
25&23&21&24&26&28&27&4&2\\
27&25&23&21&24&26&28&2&4\\
28&27&25&23&21&24&26&4&2\\
26&28&27&25&23&21&24&2&4\\
24&26&28&27&25&23&21&4&2\\
22&22&22&22&22&22&22&4&2\\
21&24&26&26&26&25&26&2&4\\
21&24&28&28&27&25&25&2&4\\
24&21&27&24&25&25&23&4&2\\
21&22&27&26&25&23&27&4&2\\
21&28&28&28&25&25&26&4&2\\
25&21&27&23&27&25&25&4&2
\end{array}\right).
\]
\endgroup
Figure~\ref{fig:two-added-columns} shows this construction with two
decomposition blocks.

\begingroup
\small
\bibliographystyle{unsrt}
\bibliography{cankao}

@article{Ransford_2026,
  title={A 98-qubit trapped-ion quantum computer with all-to-all connectivity},
  volume={655},
  ISSN={1476-4687},
  url={http://dx.doi.org/10.1038/s41586-026-10676-4},
  DOI={10.1038/s41586-026-10676-4},
  number={8121},
  journal={Nature},
  publisher={Springer Science and Business Media LLC},
  author={Ransford, Anthony and Allman, M. S. and Arkinstall, Jake and others},
  year={2026},
  month={6},
  pages={81--86}
}

@article{Horodecki2009Quantum,
  author = {Ryszard Horodecki and Pawe{\l} Horodecki and Micha{\l} Horodecki and Karol Horodecki},
  title = {Quantum entanglement},
  journal = {Reviews of Modern Physics},
  volume = {81},
  number = {2},
  pages = {865--942},
  year = {2009},
  doi = {10.1103/RevModPhys.81.865}
}

@article{Bennett1993Teleporting,
  author = {Charles H. Bennett and Gilles Brassard and Claude Cr{\'e}peau and Richard Jozsa and Asher Peres and William K. Wootters},
  title = {Teleporting an unknown quantum state via dual classical and Einstein-Podolsky-Rosen channels},
  journal = {Physical Review Letters},
  volume = {70},
  number = {13},
  pages = {1895--1899},
  year = {1993},
  doi = {10.1103/PhysRevLett.70.1895}
}

@article{Demianowicz2024Completely,
  author = {Maciej Demianowicz and Kajetan Vogtt and Remigiusz Augusiak},
  title = {Completely entangled subspaces of entanglement depth k},
  journal = {Physical Review A},
  volume = {110},
  number = {1},
  pages = {012403},
  year = {2024},
  doi = {10.1103/PhysRevA.110.012403}
}

@article{Horodecki1998Mixed,
  author = {Micha{\l} Horodecki and Pawe{\l} Horodecki and Ryszard Horodecki},
  title = {Mixed-state entanglement and distillation: Is there a bound entanglement in nature?},
  journal = {Physical Review Letters},
  volume = {80},
  number = {24},
  pages = {5239--5242},
  year = {1998},
  doi = {10.1103/PhysRevLett.80.5239}
}

@article{Augusiak2011Bell,
  author = {Remigiusz Augusiak and J. Stasi{\'n}ska and C. Hadley and J. K. Korbicz and Maciej Lewenstein and Antonio Ac{\'i}n},
  title = {Bell inequalities with no quantum violation and unextendable product bases},
  journal = {Physical Review Letters},
  volume = {107},
  number = {7},
  pages = {070401},
  year = {2011},
  doi = {10.1103/PhysRevLett.107.070401}
}

@article{Alon2001Unextendible,
  author = {Noga Alon and L{\'a}szl{\'o} Lov{\'a}sz},
  title = {Unextendible product bases},
  journal = {Journal of Combinatorial Theory. Series A},
  volume = {95},
  number = {1},
  pages = {169--179},
  year = {2001},
  doi = {10.1006/jcta.2000.3160}
}

@article{Agrawal2019Genuinely,
  author = {Sristy Agrawal and Saronath Halder and Manik Banik},
  title = {Genuinely entangled subspace with all-encompassing distillable entanglement across every bipartition},
  journal = {Physical Review A},
  volume = {99},
  pages = {032335},
  year = {2019},
  doi = {10.1103/PhysRevA.99.032335}
}

@article{Shi2020Unextendible,
  author = {Fei Shi and Xiande Zhang and Lin Chen},
  title = {Unextendible product bases from tile structures and their local entanglement-assisted distinguishability},
  journal = {Physical Review A},
  volume = {101},
  number = {6},
  pages = {062329},
  year = {2020},
  doi = {10.1103/PhysRevA.101.062329}
}

@article{You2023Unextendible,
  author = {Siwen You and Chen Wang and Fei Shi and Sihuang Hu and Yiwei Zhang},
  title = {Unextendible product bases from tile structures in bipartite systems},
  journal = {Journal of Physics A: Mathematical and Theoretical},
  volume = {56},
  number = {1},
  pages = {015303},
  year = {2023},
  doi = {10.1088/1751-8121/aca9b6}
}

@article{He2024Strong,
  author = {Yiyun He and Fei Shi and Xiande Zhang},
  title = {Strong quantum nonlocality and unextendibility without entanglement in n-partite systems with odd n},
  journal = {Quantum},
  volume = {8},
  pages = {1349},
  year = {2024},
  doi = {10.22331/q-2024-06-06-1349}
}

@article{Shi2022Strongly,
  author = {Fei Shi and Mao-Sheng Li and Mengyao Hu and Lin Chen and Man-Hong Yung and Yan-Ling Wang and Xiande Zhang},
  title = {Strongly nonlocal unextendible product bases do exist},
  journal = {Quantum},
  volume = {6},
  pages = {619},
  year = {2022},
  doi = {10.22331/q-2022-01-18-619}
}

@article{Shi2023Graph,
  author = {Fei Shi and Ge Bai and Xiande Zhang and Qi Zhao and Giulio Chiribella},
  title = {Graph-theoretic characterization of unextendible product bases},
  journal = {Physical Review Research},
  volume = {5},
  number = {3},
  pages = {033144},
  year = {2023},
  doi = {10.1103/PhysRevResearch.5.033144}
}

@article{Shi2022Strong,
  author = {Fei Shi and Mao-Sheng Li and Lin Chen and Xiande Zhang},
  title = {Strong quantum nonlocality for unextendible product bases in heterogeneous systems},
  journal = {Journal of Physics A: Mathematical and Theoretical},
  volume = {55},
  number = {1},
  pages = {015305},
  year = {2022},
  doi = {10.1088/1751-8121/ac3c9a}
}

@article{bennett1999unextendible,
  title={Unextendible product bases and bound entanglement},
  author={Bennett, Charles H and DiVincenzo, David P and Mor, Tal and Shor, Peter W and Smolin, John A and Terhal, Barbara M},
  journal={Physical Review Letters},
  volume={82},
  number={26},
  pages={5385},
  year={1999},
  publisher={APS}
}

@article{johnston2014structure,
  title={The structure of qubit unextendible product bases},
  author={Johnston, Nathaniel},
  journal={Journal of Physics A: Mathematical and Theoretical},
  volume={47},
  number={42},
  pages={424034},
  year={2014},
  publisher={IOP Publishing}
}

@article{chen2018multiqubit,
  title={Multiqubit UPB: the method of formally orthogonal matrices},
  author={Chen, Lin and {\DJ}okovi{\'c}, Dragomir {\v{Z}}.},
  journal={Journal of Physics A: Mathematical and Theoretical},
  volume={51},
  number={26},
  pages={265302},
  year={2018},
  publisher={IOP Publishing}
}

@article{wang2020construction,
  title={The construction of 7-qubit unextendible product bases of size ten},
  author={Wang, Kai and Chen, Lin},
  journal={Quantum Information Processing},
  volume={19},
  pages={1--17},
  year={2020},
  publisher={Springer}
}

@article{SunWang2024,
  author  = {Yize Sun and Baoshan Wang},
  title   = {The construction and weakly local indistinguishability of
             multiqubit unextendible product bases},
  journal = {Quantum Information Processing},
  volume  = {23},
  number  = {5},
  pages   = {191},
  year    = {2024},
  doi     = {10.1007/s11128-024-04379-w}
}

@article{chen2018nonexistence,
  title={Nonexistence of {$n$}-qubit unextendible product bases of size {$2^n-5$}},
  author={Chen, Lin and {\DJ}okovi{\'c}, Dragomir {\v{Z}}.},
  journal={Quantum Information Processing},
  volume={17},
  pages={24},
  year={2018},
  publisher={Springer}
}

@article{johnston2013minimum,
  title={The minimum size of qubit unextendible product bases},
  author={Johnston, Nathaniel},
  journal={arXiv preprint arXiv:1302.1604},
  year={2013}
}

@article{bennett1999quantum,
  title={Quantum nonlocality without entanglement},
  author={Bennett, Charles H and DiVincenzo, David P and Fuchs, Christopher A and Mor, Tal and Rains, Eric and Shor, Peter W and Smolin, John A and Wootters, William K},
  journal={Physical Review A},
  volume={59},
  number={2},
  pages={1070},
  year={1999},
  publisher={APS}
}

@article{bravyi2004unextendible,
  title={Unextendible product bases and locally unconvertible bound entangled states},
  author={Bravyi, S. B.},
  journal={Quantum Information Processing},
  volume={3},
  pages={309--329},
  year={2004},
  publisher={Springer}
}

@article{bhat2006completely,
  title={A completely entangled subspace of maximal dimension},
  author={Bhat, B. V. Rajarama},
  journal={International Journal of Quantum Information},
  volume={4},
  number={02},
  pages={325--330},
  year={2006},
  publisher={World Scientific}
}

@article{wang20194,
  title={4 * 4 4$\times$ 4 unextendible product basis and genuinely entangled space},
  author={Wang, Kai and Chen, Lin and Zhao, Lijun and Guo, Yumin},
  journal={Quantum Information Processing},
  volume={18},
  pages={1--28},
  year={2019},
  publisher={Springer}
}

@article{demianowicz2018unextendible,
  title={From unextendible product bases to genuinely entangled subspaces},
  author={Demianowicz, Maciej and Augusiak, Remigiusz},
  journal={Physical Review A},
  volume={98},
  number={1},
  pages={012313},
  year={2018},
  publisher={APS}
}

@article{chen2017orthogonal,
  title={Orthogonal product bases of four qubits},
  author={Chen, Lin and {\DJ}okovi{\'c}, Dragomir {\v{Z}}.},
  journal={Journal of Physics A: Mathematical and Theoretical},
  volume={50},
  number={39},
  pages={395301},
  year={2017},
  publisher={IOP Publishing}
}

@article{chen2015minimum,
  title={The minimum size of unextendible product bases in the bipartite case (and some multipartite cases)},
  author={Chen, Jianxin and Johnston, Nathaniel},
  journal={Communications in Mathematical Physics},
  volume={333},
  pages={351--365},
  year={2015},
  publisher={Springer}
}

@article{PhysRevA.85.042113,
  title = {Tight Bell inequalities with no quantum violation from qubit unextendible product bases},
  author = {Augusiak, R. and Fritz, T. and Kotowski, Ma. and Kotowski, Mi. and Paw{\l}owski, M. and Lewenstein, M. and Ac{\'i}n, A.},
  journal = {Phys. Rev. A},
  volume = {85},
  issue = {4},
  pages = {042113},
  numpages = {12},
  year = {2012},
  month = {Apr},
  publisher = {American Physical Society},
  doi = {10.1103/PhysRevA.85.042113},
  url = {https://link.aps.org/doi/10.1103/PhysRevA.85.042113}
}

@article{tura2012four,
  title={Four-qubit entangled symmetric states with positive partial transpositions},
  author={Tura, J and Augusiak, R and Hyllus, P and Ku{\'s}, M and Samsonowicz, J and Lewenstein, M},
  journal={Physical Review A},
  volume={85},
  number={6},
  pages={060302},
  year={2012},
  publisher={APS}
}

@article{divincenzo2003unextendible,
  title={Unextendible product bases, uncompletable product bases and bound entanglement},
  author={DiVincenzo, David P and Mor, Tal and Shor, Peter W and Smolin, John A and Terhal, Barbara M},
  journal={Communications in Mathematical Physics},
  volume={238},
  pages={379--410},
  year={2003},
  publisher={Springer}
}

@article{feng2006unextendible,
  title={Unextendible product bases and 1-factorization of complete graphs},
  author={Feng, Keqin},
  journal={Discrete applied mathematics},
  volume={154},
  number={6},
  pages={942--949},
  year={2006},
  publisher={Elsevier}
}

@article{HoffmanRodger1992,
  author  = {D. G. Hoffman and C. A. Rodger},
  title   = {The chromatic index of complete multipartite graphs},
  journal = {Journal of Graph Theory},
  volume  = {16},
  pages   = {159--163},
  year    = {1992},
  doi     = {10.1002/jgt.3190160207}
}

@article{Vizing1964,
  author  = {V. G. Vizing},
  title   = {On an estimate of the chromatic class of a {$p$}-graph},
  journal = {Diskretnyi Analiz},
  volume  = {3},
  pages   = {25--30},
  year    = {1964},
  note    = {MR0180505}
}

@book{LovaszPlummer1986,
  author    = {L{\'a}szl{\'o} Lov{\'a}sz and Michael D. Plummer},
  title     = {Matching Theory},
  series    = {North-Holland Mathematics Studies},
  volume    = {121},
  publisher = {North-Holland},
  address   = {Amsterdam},
  year      = {1986},
  isbn      = {978-0-444-87916-5}
}

@article{ChenChen2026,
  author  = {J. Chen and L. Chen},
  title   = {On the existence of {$13\times9$} unextendible orthogonal matrices},
  journal = {International Journal of Theoretical Physics},
  volume  = {65},
  pages   = {147},
  year    = {2026},
  doi     = {10.1007/s10773-026-06352-y}
}

@misc{UPBOnlineSource,
  key = {Cheng, Caohan},
  howpublished = {Cheng, C.: Online source for {Systematic construction of multiqubit unextendible product bases}. {GitHub}. \url{https://github.com/chechche123/upb-research}}
}
\endgroup
\end{document}